%% file: panel_im_test.tex
\documentclass[12pt]{article}

\usepackage[reqno]{amsmath}
\usepackage{amsthm,amssymb,amsxtra,bbm,mathtools,mathrsfs}
\usepackage{booktabs,longtable,multirow,arydshln,cases,empheq}
\usepackage{enumerate,graphicx,wasysym,tikz,bm,dsfont,xcolor,float,pdflscape}
\usepackage[mathcal]{euscript}
\usepackage[longnamesfirst]{natbib}
\usepackage[margin=1in]{geometry}
\usepackage{caption,subcaption}
\usepackage[export]{adjustbox}
\usepackage{accents}

\newcommand{\dto}{\overset{d}{\to}}
\newcommand{\pto}{\overset{p}{\to}}

\newcommand{\textupandbold}[1]{\textup{\textbf{#1}}}

\theoremstyle{plain}
\newtheorem{thm}{Theorem}[section]

\newtheorem{lem}{Lemma}[section]
\newtheorem{ass}{Assumption}[section]

\theoremstyle{definition}

\newtheorem{rem}{Remark}[section]

\numberwithin{equation}{section}

\newcommand{\1}{\mathbbm{1}}
\newcommand{\E}{\mathbb{E}}
\newcommand{\N}{\mathbb{N}}
\newcommand{\R}{\mathbb{R}}

\newcommand{\calJ}{\mathcal{J}}

\newcommand{\calU}{\mathcal{U}}
\newcommand{\calX}{\mathcal{X}}
\newcommand{\calY}{\mathcal{Y}}

\newcommand{\mfa}{\mathfrak{a}}
\newcommand{\frakp}{{\mathfrak{p}}}
\newcommand{\frakt}{{\mathfrak{t}}}

\newcommand{\argmax}{\operatorname*{argmax}}

\newcommand{\cov}{\operatorname{cov}}
\newcommand{\defined}{\vcentcolon=}
\newcommand{\supp}{\operatorname{supp}}
\newcommand{\var}{\operatorname{var}}
\newcommand{\wpone}{\mathrm{w.p.1}}
\newcommand{\plim}{\operatorname*{plim}}
\newcommand{\No}{\mathrm{N}}

\newcommand{\Yibar}{\bar{Y}_i}
\newcommand{\sumi}{\sum_{i=1}^n}

\newcommand{\alphaio}{\alpha_{i0}}
\newcommand{\alphaiT}{\alpha_{i}}
\newcommand{\alphahatiT}{\hat{\alpha}_{i}}
\newcommand{\calXiT}{\calX_{i}}

\newcommand{\elliT}{\ell_{i}}

\newcommand{\elliTtheta}{\ell_{i\theta}}

\newcommand{\thetahat}{\hat{\theta}}
\newcommand{\thetahatp}{\thetahat^\frakp}
\newcommand{\thetahatstar}{\thetahat^\frakt}

\newcommand{\Xit}{X_{it}}
\newcommand{\Yit}{Y_{it}}
\newcommand{\thetabar}{\bar{\theta}}

\newcommand{\Op}{O_{\mathrm{p}}}
\newcommand{\op}{o_{\mathrm{p}}}
\newcommand{\Opstar}{O_{\mathrm{p}^*}}
\newcommand{\opstar}{o_{\mathrm{p}^*}}
\newcommand{\dtostar}{\overset{d^*}{\to}}

\newcommand{\sumt}{\sum_{t=1}^T}

\newcommand{\IB}{\mathrm{IB}}

\allowdisplaybreaks

\newcommand{\alphahat}{\hat\alpha}
\newcommand{\alphai}{\alpha_i}

\newcommand{\elliOne}{\ell_{i,1}}
\newcommand{\elliTwo}{\ell_{i,2}}
\newcommand{\elliThree}{\ell_{i,3}}
\newcommand{\ellitOne}{\ell_{it,1}}

\newcommand{\lk}{l_k}

\newcommand{\Iio}{I_{i0}}
\newcommand{\ellit}{\ell_{it}}
\newcommand{\Vhatp}{\widehat{V}^\frakp}
\newcommand{\Vhatstar}{\widehat{V}^\frakt}
\newcommand{\Vhat}{\widehat{V}}
\newcommand{\calW}{\mathcal{W}}
\newcommand{\calG}{\mathcal{G}}

\newcommand{\Whattarget}{\calW^\frakt}
\newcommand{\Whatprofile}{\calW^\frakp}

\begin{document}

\title{The Panel Information Matrix Test}
\author{Martin Schumann\\[0.5em]
{\normalsize Maastricht University}\\[0.25em]
{\normalsize \texttt{m.schumann@maastrichtuniversity.nl}}}
\date{\today}

\maketitle

\begin{abstract}
\noindent
The Information Matrix (IM) test is a natural specification check for likelihood-based models, yet cross-sectional implementations are often badly sized, confound neglected heterogeneity with distributional misspecification, and require third derivatives. We develop the panel information matrix (PIM) test for models with fixed effects. Explicit fixed effects isolate functional-form misspecification from time-invariant unobserved heterogeneity. Although the profile maximum likelihood estimator inherits incidental-parameter bias of order $O(1/T)$, the leading bias of the profile PIM is of order $\sqrt{n}/T$, so it is asymptotically negligible whenever $n/T^2\to 0$, which includes the rectangular regime $n/T\to\mathrm{const}$, and stabilizes under parabolic asymptotics $n/T^2\to\rho\in(0,\infty)$. This robustness arises because the indicator fluctuates at rate $\sqrt{n}$ rather than $\sqrt{nT}$. The same rate argument makes third-derivative corrections asymptotically unnecessary, so the test uses only first- and second-order likelihood derivatives. Simulations confirm that asymptotic PIM critical values suffice in moderately long panels, even when Wald tests for $\theta$ are badly sized, while a parametric bootstrap largely removes size distortions when $\rho>0$.
\end{abstract}


\section{Introduction}\label{sec:introduction}

Likelihood-based inference requires correct specification of the conditional density of the outcome. Under a correctly specified model, counterfactual and policy analysis can identify estimands that are not identified from experimental or quasi-experimental variation alone \citep{heckman-vytlacil-2005}. The Information Matrix (IM) test of \citet{white1982maximum} is a natural check of that hypothesis, yet it is rarely used. Three difficulties are well known. The test has poor size in small samples \citep{taylor1987size,orme1990small,chesher1991asymptotic}. It often rejects because of neglected heterogeneity rather than because the functional form of the density is wrong \citep{chesher1984testing}. In cross-sections, its asymptotic variance depends on the estimation error of the parameters, so the statistic requires third derivatives of the log-likelihood \citep{chesher1983information,lancaster1984covariance}, which are cumbersome to derive and often numerically unstable.

This paper develops the asymptotic theory of the IM test in panel models with fixed effects. The resulting procedure is the panel information matrix (PIM) test. A central feature of the panel setting is that the information-matrix indicator is an average across individuals of within-individual moments, so that its sampling variability is of order $n^{-1/2}$ rather than $(nT)^{-1/2}$. Relative to the classical cross-sectional IM test, this $\sqrt{n}$ scaling has two main consequences: the profile PIM is first-order robust to incidental parameters, and third-derivative corrections in the asymptotic variance become redundant.

Including fixed effects isolates functional-form misspecification of the idiosyncratic density from time-invariant unobserved heterogeneity and thereby addresses the critique in \citet{chesher1984testing}. For estimation of the common parameter $\theta$, however, profiling out the fixed effects creates an incidental parameter problem \citep{neyman-scott}, as the profile maximum likelihood estimator has a first-order bias of order $O(1/T)$, and Wald-type inference based on that estimator is misleading under rectangular asymptotics $n/T\to\mathrm{const}$ \citep{hahn-newey-nonlinear-panels,hahn-kuersteiner-nonlinear-panels}. The PIM statistic behaves differently. Incidental-parameter effects enter the profile indicator through the score and information biases of the profile likelihood, both of order $O(T^{-1})$ \citep{lee-phillips2015modelselection,schumann-info}. After $\sqrt{n}$ scaling, the leading bias of the profile indicator is therefore of order $\sqrt{n}/T$. This term is $o(1)$ whenever $n/T^2\to 0$, which includes the rectangular regime. Consequently, panels that are too short for reliable inference on $\theta$ can still support a first-order valid specification test. Moreover, the infeasible target likelihood, and feasible alternatives that are first-order equivalent to it, yield a PIM statistic whose null limit is central $\chi^2$ without a restriction on $n/T^2$.

When $n/T^2\to\rho\in(0,\infty)$, the profile PIM converges to a noncentral chi-square whose noncentrality is proportional to $\rho$. In that regime, and more generally in short panels, we recommend a parametric bootstrap along the lines of \citet{horowitz1994bootstrap} and \citet{higgins2024bootstrap}. Under high-level conditions stated below, the bootstrap replicates the limiting law of the PIM statistic. Monte Carlo experiments for static logit, static probit, and dynamic probit show that asymptotic critical values already work well for moderate $T$, even when Wald tests for $\theta$ remain badly sized, while the bootstrap further improves size when $T$ is small. Power against fixed alternatives is reported in the appendix.

The same $\sqrt{n}$ scaling also implies that third-derivative corrections for parameter uncertainty are asymptotically unnecessary. In cross-sections, plug-in estimation error enters the asymptotic variance of the IM indicator on the same scale as the indicator itself \citep{lancaster1984covariance}. In panels, the estimation error for $\theta$ shrinks with both $n$ and $T$, while the indicator continues to fluctuate at rate $\sqrt{n}$. The PIM can therefore be computed from first- and second-order likelihood derivatives only, which is especially attractive when the likelihood must be evaluated numerically, as in many structural panel models \citep{rust1987,aguirregabiria-mira-2007,dearing-blevins-2025}. 

Recent work has revived the IM test by characterizing and decomposing its moment restrictions in specific models. \citet{amengual2024multivariate} relate the multivariate-normality IM statistic to Hermite polynomial moment tests and provide a directional decomposition. Related results for multinomial logit and incomplete-data models are developed in \citet{amengual2025logit} and \citet{amengual2026mixtures}. These contributions are cross-sectional. To the best of our knowledge, the asymptotic properties of the IM test in panel models with fixed effects, and in particular its robustness to incidental parameters, have not been studied.

The remainder of the paper is organized as follows. Section~\ref{sec:pseudo-lik} introduces the panel pseudo-likelihoods and the associated score- and information-bias expansions. Section~3 develops the limiting distributions of the target and profile PIM statistics. Section~\ref{sec:bootstrap} develops the parametric bootstrap. Section~\ref{sec:simulations} reports Monte Carlo evidence. Section~\ref{sec:conclusion} concludes. Assumptions, intermediate results, and proofs are collected in Sections~\ref{sec-ass}--\ref{sec-proofs}. Supplementary simulation results, a Neyman--Scott illustration, implementation details, and the illustrative applications of Section~\ref{sec:illustrative-example} are deferred to the appendices.

\section{Pseudo-likelihoods in panel data models with fixed effects}\label{sec:pseudo-lik}
Let $\Yit$ denote outcomes and $\Xit$ a (column) vector of explanatory variables for $i=1, \ldots, n$ individuals and $t=1, \ldots, T$ time periods,
where $n,T \ge 2$. 
The time-invariant fixed effect, denoted by $\alphaio$, is an individual-specific unobserved random variable
whose distribution has known support $\calJ\defined \supp(\alphaio) \subset \R$, but is otherwise unknown. By modeling unobserved heterogeneity through $\alphaio$, the  PIM test developed here isolates misspecification of the idiosyncratic density from neglected heterogeneity of the type emphasized by \citet{chesher1984testing}. Let further  $\calXiT=(X_{i1},\ldots,X_{iT})$ denote the time series of the observed covariates and $\calY_{i0}=(Y_{i,1-L},\ldots,Y_{i0})$ for a fixed lag length $L\in\N$. 
Following \citet{higgins2024bootstrap}, we assume that the likelihood for individual $i$ is the density, conditional on $\calXiT$ and $\calY_{i0}$, with respect to some appropriate dominating measure. Here, $\theta\in\Theta\subset\R^p$ is the common parameter of interest with true value $\theta_0$ in the interior of $\Theta$, and $\alpha_i$ is a candidate value for the incidental parameter $\alphaio$. The period log-likelihood contribution is
\[
\ellit(\theta,\alpha_i)=\log f_{it}(\Yit|Y_{it-1},\ldots,Y_{it-L},\calXiT;\theta,\alpha_i),
\]
and the scaled log-likelihood for the $i$-th stratum is $\elliT(\theta,\alpha_i)=T^{-1}\sumt\ellit(\theta,\alpha_i)$. The joint distribution of $(\calXiT,\alphaio)$ is unrestricted, which allows for arbitrary correlation between the explanatory variables and the fixed effect. The operator $\E_0$ refers to integration with respect to the product of the conditional densities in the sample, $\prod_{i=1}^n\prod_{t=1}^T f_{it}(\Yit|Y_{it-1},\ldots,Y_{it-L},\calXiT;\theta_0,\alphaio)$, whereas unconditional expectations are denoted by $\E$.\footnote{When the likelihood for individual $i$ is evaluated at non-random parameters, $\E_0$ reduces to an integral with respect to $\prod_{t=1}^T f_{it}(\Yit|Y_{it-1},\ldots,Y_{it-L},\calXiT;\theta,\alpha_i)$.} Precise regularity conditions on derivatives and moments are collected in Appendix~\ref{sec-ass}.

\subsection{Profile and target likelihoods}\label{sec:score-info-bias}
Since $\alpha_i$ is unobserved, it must be eliminated from the likelihood, resulting in a pseudo-likelihood that only depends on the common parameter. The two pseudo-likelihoods considered in this paper are the concentrated or ``profile'' likelihood and the infeasible ``target'' likelihood, given as $\ell_i^\frakp(\theta)\defined \elliT(\theta,\alphahatiT(\theta))$ and $\ell_i^\frakt(\theta)\defined \elliT(\theta,\alphai^\frakt(\theta))$, where $\alphahatiT(\theta)=\argmax_{\alpha_i\in\calJ}\elliT(\theta,\alpha_i)$ and $\alphai^\frakt(\theta)\defined\argmax_{\alpha_i\in\calJ}\E_0\elliT(\theta,\alpha_i)$. Their sample averages are $\ell^\frakp(\theta)=n^{-1}\sumi\ell_i^\frakp(\theta)$ and $\ell^\frakt(\theta)=n^{-1}\sumi\ell_i^\frakt(\theta)$, with maximizers $\thetahat^\frakp=\argmax_{\theta\in\Theta}\ell^\frakp(\theta)$ and $\thetahat^\frakt=\argmax_{\theta\in\Theta}\ell^\frakt(\theta)$. Superscripts $^\frakp$ and $^\frakt$ are also attached to the corresponding IM indicators, variance estimators, and test statistics introduced below. When a result applies to either pseudo-likelihood, we omit the superscript and write $\ell_i(\theta)$, $\thetahat$, and so on. Derivatives of a pseudo-likelihood with respect to $\theta$ are indexed numerically: $\elliOne(\theta)=\partial_\theta\ell_i(\theta)$ is a $p$-vector and $\elliTwo(\theta)=\partial_{\theta\theta'}^2\ell_i(\theta)$ is a $p\times p$ matrix.

The target likelihood is score unbiased and information unbiased at $\theta_0$ \citep{pace-salvan-JSPI}. While the target likelihood itself is not available in practical applications, it serves here as a placeholder for feasible likelihoods that are fixed-$T$ consistent, or first-order equivalent to the target as $T\to\infty$, under additional structure---for example conditional likelihoods when a sufficient statistic for $\alpha_i$ exists, and integrated likelihoods based on a correctly specified or robust prior for the fixed effects \citep{arellano-bonhomme}. Results for $\ell^\frakt$ therefore cover these cases.

Unlike the target likelihood, the profile likelihood suffers from score- and information bias of order $\Op(T^{-1})$ due to the presence of the unobserved fixed effects \citep{schumann-info}. In particular, the score bias drives the incidental-parameter bias of the fixed-effects MLE $\thetahat^\frakp$ \citep{neyman-scott,hahn-newey-nonlinear-panels}. Because of these first-order biases, the profile likelihood does not behave like a ``genuine'' likelihood even under correct specification \citep{lee-phillips2015modelselection}, and inference based on the MLE is misleading under rectangular asymptotics. A large literature develops analytical, jackknife, and related corrections that remove the leading first-order bias of the fixed-effects MLE and the associated score- and information biases; see the surveys of \citet{lancaster-survey-joe}, \citet{arellano-hahn-world-congress}, and \citet{fernandez2018fixed}, and the developments in \citet{hahn-kuersteiner-nonlinear-panels}, \citet{fernandez-val-probit}, \citet{dhaene-jochmans-restud},  \citet{diciccio-martin-stern-young}, and \citet{schumann2020joe}. Despite the fixed-effects MLE entering the panel IM statistic, we show in the next section that such corrections are asymptotically unnecessary for the PIM test, as the leading incidental-parameter bias term vanishes under rectangular asymptotics.

\section{The IM test in panel data models}
In this section, we establish the asymptotic properties of the PIM test. We demonstrate two main results. First, unlike in the pure cross-sectional case, the PIM statistic does not require third derivatives of the underlying pseudo-likelihood. Second, under parabolic asymptotics, i.e., $\rho=\lim n/T^2\in[0,\infty)$,  incidental-parameter bias affecting the profile PIM vanishes when $\rho=0$ and stabilizes when $\rho\in(0,\infty)$. This contrasts with Wald-type inference for parameters, which is biased under rectangular asymptotics.

To define the PIM statistic we adapt the IM test of \citet{white1982maximum} to panel data. Recall that $\elliOne(\theta)$ and $\elliTwo(\theta)$ are the score and Hessian of the $T$-averaged pseudo-likelihood $\ell_i(\theta)$. The corresponding information-matrix indicator for individual $i$ is based on
\[
T\elliOne(\theta)\elliOne(\theta)'+\elliTwo(\theta),
\] 
where the factor $T$ appears because $\E_0[\elliOne(\theta_0)\elliOne(\theta_0)']=O(T^{-1})$ while $\E_0[\elliTwo(\theta_0)]=O(1)$, so both terms are of the same order. As in \citet{white1982maximum}, the test may use a subset of the $p(p+1)/2$ distinct entries of this matrix. Let $D_i(\theta)\in\R^q$ collect the chosen components, where $1\le q\le p(p+1)/2$, and set
\[
D(\theta)=\frac{1}{n}\sumi D_i(\theta).
\]
The following lemma shows that $\var_0(\sqrt{n}D(\theta_0))$ remains of order one as $n$ and $T$ grow. The subsequent asymptotic properties of the PIM rest on this result.
\begin{lem}\label{lemma-variance}
Let Assumptions~\ref{ass-sampling} and~\ref{ass-moments} hold. Then,
\begin{equation*}
\var_0(\sqrt{n}D(\theta_0))=\Op(1).
\label{eq:variance}
\end{equation*}
\end{lem}
In the cross-sectional IM test of \citet{white1982maximum}, the averaged indicator has variance of order $1/n$, so the statistic is scaled by $\sqrt{n}$. The same scaling is the right one in the panel. Under the target likelihood the individual indicators are mean zero, but each contains products of scores from different periods. The variance of those products does not vanish as $T\to\infty$, even without unobserved individual effects, so $\var_0(D_i(\theta_0))$ remains of order one and the sample indicator is scaled by $\sqrt{n}$. 

The PIM test is constructed based on the indicator  evaluated at the pseudo-likelihood maximizer, e.g.\ $D^\frakp(\thetahat^\frakp)$ for the profile case. As is well-known, the profile estimator $\thetahat^\frakp$ suffers from an incidental-parameter bias of order $\Op(T^{-1})$ in the presence of fixed effects, which is driven by the score-bias of the same order of the profile likelihood (\citealp{schumann-info}). As a consequence, the asymptotic distribution of $\sqrt{nT}(\thetahat^\frakp-\theta_0)$ is biased under rectangular asymptotics (see, e.g., \citealp{hahn-newey-nonlinear-panels,hahn-kuersteiner-nonlinear-panels}), and hypothesis tests based on $\thetahat^\frakp$ are severely misleading in short panels. The score- and information bias of the profile likelihood also affect the bias of the profile indicator function when evaluated at the MLE, since $\E_0[D^\frakp(\thetahat^\frakp)]=\Op(T^{-1})$. However, because of Lemma \ref{lemma-variance}, the indicator is scaled by $\sqrt{n}$. As is shown in the following Lemma, this implies that $\sqrt{n}D^\frakp(\thetahat^\frakp)$ remains asymptotically unbiased if $n/T^2\to 0$, whereas the incidental-parameter bias of the scaled profile indicator stabilizes if $n/T^2\to\rho\in (0,\infty)$. Moreover, it is shown that the target likelihood indicator is asymptotically unbiased without the need to impose relative rates on $n$ and $T$. To state the Lemma, let $V\defined\lim_{n,T\to\infty}\var_0(\sqrt{n}D^\frakt(\theta_0))$, which exists and is nonsingular (see Assumption~\ref{ass-variance}).
\begin{lem}\label{lemma-scaled-indicator}
Let Assumptions~\ref{ass-sampling}--\ref{ass-estimators}, and~\ref{ass-variance} hold.
\begin{enumerate}
\item 
\[
\sqrt{n}D^\frakt(\thetahat^\frakt)=\sqrt{n}D^\frakt(\theta_0)+\Op(n^{-1/2})+\Op(T^{-1/2}) 
\]
and 
\[
\sqrt{n}D^\frakt(\theta_0)\dto \No(0,V).
\]
\item 
If $n,T\to\infty$, 
\[
\sqrt{n}D^\frakp(\thetahat^\frakp)=
\sqrt{n}D^\frakt(\theta_0)+\Op(n^{-1/2})+\Op(T^{-1/2})+\Op(\frac{\sqrt{n}}{T}).
\]
Moreover, if Assumption~\ref{ass-n-T-rate}\textup{(ii)} holds, i.e.\ $n/T^2\to\rho\in[0,\infty)$, then
\[
\sqrt{n}D^\frakp(\thetahat^\frakp)\dto \No(\sqrt{\rho}\,\kappa(\theta_0),V),
\] 
where $\kappa(\theta_0)$ is a non-random limit containing the first-order score and information bias of the profile likelihood (see the proof).
\end{enumerate}
\end{lem}
\begin{rem}
    Although the target estimator can be consistent for \(\theta_0\) at fixed panel length, the first part of Lemma \ref{lemma-scaled-indicator} still requires \(T\to\infty\). To see this, consider the first term in a Taylor expansion  \(\sqrt{n}\,D_1^\frakt(\theta_0)(\thetahat^\frakt-\theta_0)\). By Lemma~\ref{lemma-rates}, \(\E_0[D_{i,1}^\frakt(\theta_0)]=\Op(1)\) and \(\var_0(D_{i,1}^\frakt(\theta_0))=\Op(T)\), so this remainder is \(\Op(n^{-1/2})+\Op(T^{-1/2})\). The \(\Op(T^{-1/2})\) piece, which forces \(T\to\infty\), comes from the nonzero mean of \(D_1^\frakt(\theta_0)\). Bartlett identities keep \(D_i^\frakt(\theta_0)\) mean zero with \(\Op(1)\) variance only at the true parameter, where the score is mean zero. Since they do not transfer to \(\theta\)-derivatives of the indicator,  higher-order likelihood derivatives retain nonzero means, and \(D_{i,1}^\frakt\) therefore does not inherit the rates of \(D_i^\frakt\) itself.
    \end{rem}

\begin{rem}
    The asymptotic bias term $\sqrt{\rho}\kappa(\theta_0)$ that appears in the asymptotic distribution of the profile indicator when $\rho>0$ is a function of both the score- and information bias of the profile likelihood. The score bias enters when expanding the indicator around the true parameter, as the first-order bias of the fixed effects MLE is driven by the first-order score bias of the profile likelihood. The information bias enters when replacing the profile indicator at $\theta_0$ with the target indicator at the true parameter.
\end{rem}

As shown in Lemma \ref{lemma-variance}, the order of the variance does not depend on the panel length. In contrast, both the bias and the variance of $\thetahat^\frakp$ and $\thetahat^\frakt$ vanish as $n,T\to\infty$. Combined with the plug-in rates in Lemma \ref{lemma-scaled-indicator},  this fact implies that the estimator of the variance of $D(\thetahat)$ does not require third derivatives of the pseudo-likelihood $\ell$, even in the presence of unobserved fixed effects:
\begin{lem}\label{lemma-variance-estimator}
Let Assumptions~\ref{ass-sampling} -- \ref{ass-estimators}, and~\ref{ass-variance} hold. Define
\[
    \Vhat(\theta)\defined \frac{1}{n}\sumi D_i(\theta)D_i(\theta)'.
\]
Then,
\[
\Vhat(\thetahat)=V+\op(1)
\]
for $\ell\in\{\ell^\frakp,\ell^\frakt\}$.
\end{lem}
Given that the estimator of the variance only requires first- and second derivatives of the pseudo-likelihood, we define the PIM test statistics based on the target and profile likelihoods as 
\begin{equation}
\calW^\frakt\defined n D^\frakt(\thetahat^\frakt)' (\widehat{V}^\frakt)^{-1} D^\frakt(\thetahat^\frakt),
\qquad
\calW^\frakp\defined n D^\frakp(\thetahat^\frakp)' (\widehat{V}^\frakp)^{-1} D^\frakp(\thetahat^\frakp).
\label{eq:def-Whatp}
\end{equation}

Lemmas \ref{lemma-scaled-indicator} and \ref{lemma-variance-estimator} directly lead to the main result of this section, which derives the asymptotic distribution of the PIM test for the target and the profile likelihood.
\begin{thm}\label{main-theorem}
Let Assumptions~\ref{ass-sampling}--\ref{ass-estimators} and~\ref{ass-variance} hold.
\begin{enumerate}
	\item As $n,T\to\infty$,
	\begin{equation}
	\Whattarget \dto \chi^2_q.
\label{eq:main-thm-target}
\end{equation}
\item If $n,T\to\infty$ such that $n/T^2\to\rho\in[0,\infty)$, then
\begin{equation}
\Whatprofile \dto \chi^2_q(\lambda_\rho),
\qquad
\lambda_\rho=\rho\,\kappa(\theta_0)'V^{-1}\kappa(\theta_0).
\label{eq:main-thm-profile}
\end{equation}
\end{enumerate}
\end{thm}
Theorem~\ref{main-theorem} highlights two practical consequences.
The more surprising one concerns incidental parameters. Under rectangular
asymptotics the profile MLE $\thetahat^\frakp$ is biased at order $O(1/T)$,
and Wald tests for $\theta$ based on it are unreliable. In contrast, 
when $n/T^2\to 0$, $\Whatprofile$ has a central $\chi^2_q$ limit even
though estimation uses the same biased $\thetahat^\frakp$. When
$n/T^2\to\rho\in(0,\infty)$, the distortion is confined to the explicit
noncentrality $\lambda_\rho$ in \eqref{eq:main-thm-profile}, which contrasts with the explosive bias of the Wald test statistic. A specification
check can therefore be first-order valid in panels that are still too short
for conventional asymptotic inference on $\theta$.
The second consequence is computational. Because plug-in estimation error is
asymptotically negligible for $D(\thetahat)$, the feasible statistic uses only
the sample second-moment matrix of the individual indicators
$D_i(\thetahat)$. Third derivatives of the log-likelihood, which have long
impeded routine use of the cross-sectional IM test
\citep{chesher1983information,lancaster1984covariance,orme1990small}, are
not required. For the target (and target-equivalent) PIM, the null limit is
central $\chi^2_q$ without any restriction on $n/T^2$.
\begin{rem}
The IM test is the leading example of a likelihood specification test based on Bartlett identities. Related procedures include conditional moment diagnostics that exploit the same identities \citep{newey1985,tauchen1985} and tests of higher-order Bartlett restrictions. The rate argument developed here is not special to the second identity, since whenever the indicator is an average across individuals of within-individual likelihood moments whose $\sqrt{n}$-scaled variance remains of order one as $T\to\infty$, plug-in of $\thetahat$ becomes negligible and incidental-parameter bias enters only through $\sqrt{n}/T$. The same qualitative conclusions---redundancy of the next order of derivatives in the variance estimator, and first-order robustness of the profile statistic under rectangular asymptotics---should therefore extend to other Bartlett-identity tests in panels with fixed effects. A full treatment is left for future work.
\end{rem}
\begin{rem}
    The results in Theorem \ref{main-theorem} crucially depend on the definition of the individual indicator. 
Write $\ell_{it,1}(\theta)$ and $\ell_{it,2}(\theta)$ for the period-$t$ contributions to the score and Hessian of $\ell_i(\theta)$, so that $\elliOne(\theta)=T^{-1}\sumt\ell_{it,1}(\theta)$ and $\elliTwo(\theta)=T^{-1}\sumt\ell_{it,2}(\theta)$. Then,
\[
D_i(\theta)
=
T^{-1}\sumt\bigl(\ell_{it,1}(\theta)\ell_{it,1}(\theta)'+\ell_{it,2}(\theta)\bigr)
+
T^{-1}\sum_{t\neq s}\ell_{it,1}(\theta)\ell_{is,1}(\theta)',
\]
where the second (cross-period) sum is why $\var_0(D_i(\theta_0))$ remains bounded as $T\to\infty$, and why the sample indicator must be scaled by $\sqrt{n}$. Because those products are mean-zero at $\theta_0$ under independence across time, an apparently natural alternative is to retain only the period-level terms
\[
\widetilde{D}_i(\theta)
=
\frac{1}{T}\sumt\bigl(\ell_{it,1}(\theta)\ell_{it,1}(\theta)'+\ell_{it,2}(\theta)\bigr),
\qquad
\widetilde{D}=n^{-1}\sumi\widetilde{D}_i.
\]
For this choice, $\var_0(\widetilde{D}(\theta_0))=O((nT)^{-1})$, so the corresponding statistic must be scaled by $\sqrt{nT}$. The plug-in error replacing $\theta_0$ by  $\thetahat$ is then of the same order as the fluctuation of the indicator, and, as in the pure cross-sectional case, the variance estimator requires third derivatives of $\ell$. In addition,  the $O(T^{-1})$ incidental-parameter bias of the profile likelihood indicator is multiplied by $\sqrt{nT}$ rather than $\sqrt{n}$, resulting in an asymptotic bias under rectangular asymptotics.
\end{rem}

\section{Parametric bootstrap}\label{sec:bootstrap}
While Theorem~\ref{main-theorem} shows that the PIM is reliable even when $n$ is substantially larger than $T$, size distortions can remain when asymptotic inference for the profile statistic uses central $\chi^2_q$ critical values in the stabilizing regime $\rho\in(0,\infty)$. Residual finite-sample distortions of that feasible statistic may also persist in short panels. In these cases we recommend a parametric bootstrap of the profile PIM. \citet{horowitz1994bootstrap} motivates a parametric bootstrap for the cross-sectional information-matrix test. The construction follows \citet{higgins2024bootstrap}. Formal bootstrap validity is not inherited from \citet{higgins2024bootstrap}, because their result applies to the MLE and the likelihood-ratio statistic under rectangular asymptotics, and their delta method (Theorem~2) does not apply directly to the PIM statistic. We therefore establish a separate first-order bootstrap validity result for the profile statistic under neighborhood-uniform versions of the $\E_0$ primitives in Section~\ref{sec:ass-bootstrap}.

Following \citet{higgins2024bootstrap}, bootstrap outcomes are drawn recursively from the fitted profile transition densities at $(\thetahat^\frakp,\{\alphahat_i\})$, holding covariates and initial conditions fixed. On each bootstrap panel one recomputes  the profile bootstrap PIM statistic $\calW^{\frakp*}$. Write $P^*$ for the bootstrap probability measure conditional on the original sample, and $P_0$ for the sampling probability under the true conditional density with parameter $(\theta_0',\alpha_{10},...,\alpha_{n0})'$. We write $\dtostar$ for weak convergence under $P^*$, understood throughout as holding in probability under the sampling measure $P_0$. Explicit transition densities and implementation details are given in Appendix~\ref{sec:mc-implementation}.

\begin{thm}\label{thm:bootstrap}
Let Assumptions~\ref{ass-sampling}--\ref{ass-estimators},~\ref{ass-variance}, and~\ref{ass-bootstrap-uniform}--\ref{ass-bootstrap-bias} hold, and suppose $n/T^2\to\rho\in[0,\infty)$. Conditionally on the sample, under the parametric bootstrap for the profile pseudo-likelihood,
\[
\calW^{\frakp*}\dtostar \chi^2_q(\lambda_\rho)
\quad\text{in probability},
\]
with $\lambda_\rho=\rho\,\kappa(\theta_0)'V^{-1}\kappa(\theta_0)$ as in Theorem~\ref{main-theorem}. Consequently,
\[
\sup_{x\in\R}\bigl|P^*(\calW^{\frakp*}\le x)-P_0(\Whatprofile\le x)\bigr|\to_p 0,
\]
so that bootstrap critical values are consistent for the quantiles of the limiting law of the profile PIM statistic.
\end{thm}
Theorem~\ref{thm:bootstrap} is proved under the neighborhood-uniform bootstrap conditions in Section~\ref{sec:ass-bootstrap}. When $\rho>0$, the profile bootstrap replicates the noncentral limit $\chi^2_q(\lambda_\rho)$ rather than a central $\chi^2_q$ law. Consequently, the bootstrap yields critical values that track the sampling distribution of $\Whatprofile$. While the target PIM does not suffer from size distortions as $n,T\to\infty$ with $n/T^2\to\rho>0$, we find that a similar parametric bootstrap procedure outlined in Appendix~\ref{sec:mc-implementation} improves the size of the target PIM in small panel data sets. This is in line with the findings of \cite{horowitz1994bootstrap}, who finds that the parametric bootstrap improves the empirical size of the cross-sectional IM test for correctly specified likelihoods.
\section{Simulations}\label{sec:simulations}
In this section, we analyze the small-sample properties of the PIM test in static logit, static probit, and dynamic probit under the designs in Section~\ref{sec:sim-setup}. Neyman--Scott simulation results are deferred to Subsection~\ref{sec:add-sim-ns}.

\subsection{Simulation Setup}\label{sec:sim-setup}
\noindent
We conduct Monte Carlo experiments for three panel binary-choice models: static logit, static probit, and dynamic probit. In each design, we report results for a balanced panel with $n\in\{125,250,500,1000\}$ and $T\in\{5,10,15,20,25,30,40,50\}$. 

\noindent 
Let $m_{it}(\theta)=\theta_1 X_{it,1}+\theta_2 X_{it,2}$. In the correctly-specified design, the outcomes are generated as
\begin{equation}\label{eq:assumed-model-simulation}
\Pr(\Yit=1|\calXiT,\alphaio;\theta_0)=F(m_{it}(\theta_0)+\alphaio),
\end{equation}
where $F$ corresponds to the standard logistic CDF in static logit and to the standard normal CDF in static probit, respectively. The covariates $X_{it,1}$ and $X_{it,2}$ are standard normal with $\mathrm{Corr}(X_{it,1},X_{it,2})=0.5$. The individual effects satisfy $\alphaio\sim\No(\bar{X}_{i,1},1)$ with $\bar{X}_{i,1}=T^{-1}\sumt X_{it,1}$, and we set $\theta_0=(\theta_{0,1},\theta_{0,2})'=(1,-1)'$. Estimation always assumes model \eqref{eq:assumed-model-simulation}. We further consider two data generating processes that lead to a misspecified likelihood. In the heteroskedastic design, each unit has a latent scale $\sigma_i=\exp(0.75\,\xi_i)$ with $\xi_i\stackrel{\mathrm{iid}}{\sim}\No(0,1)$, and outcomes follow
\begin{equation*}
\Pr(\Yit=1\mid\calXiT,\alphaio,\sigma_i;\theta_0)
=F\!\left(\frac{m_{it}(\theta_0)+\alphaio}{\sigma_i}\right),
\end{equation*}
whereas in the omitted-interaction design \citep[Section~3.2, eq.~(7)]{horowitz1994bootstrap},
\begin{equation*}
\Pr(\Yit=1\mid\calXiT,\alphaio;\theta_0)
=F\!\left(m_{it}(\theta_0)+1.5\,X_{it,1}X_{it,2}+\alphaio\right).
\end{equation*}
In dynamic probit, $m_{it}(\theta)$ is replaced by
\[
m_{it}^{\mathrm{dyn}}(\theta,\gamma)
=\theta_{1} X_{it,1}+\theta_{2} X_{it,2}+\gamma Y_{i,t-1},
\]
with $\gamma_0=0.5$. 

In each of the 5000 Monte Carlo iterations, we compute the PIM test based on the target and the profile likelihood as defined in \eqref{eq:def-Whatp}. For the asymptotic PIM, we compare the statistic with the $95\%$ quantile of $\chi^2_q$, where $q=K(K+1)/2$ is the number of distinct entries of the information-matrix indicator ($K=2$ and $q=3$ in static logit and probit, $K=3$ and $q=6$ in dynamic probit). The Wald test uses $\chi^2_K$. The bootstrap is implemented as in Appendix~\ref{sec:mc-implementation} and uses $B=99$ bootstrap draws. We further report the rejection rates of the Wald test, which is asymptotically biased under rectangular asymptotics. The comparison of the Wald test and the PIM test highlights the practical relevance of the asymptotic regime since the PIM test is expected to perform well under rectangular asymptotics, whereas the Wald test is expected to suffer from severe size distortions. The rejection rates under the misspecified designs are reported in Appendix~\ref{sec:add-sim-results}.

\subsection{Simulation Results}\label{sec:sim-results}
\noindent
Tables~\ref{tab:logitrejectionrates}--\ref{tab:dynamicprobitrejectionrates} report rejection rates under the correctly specified DGP. Entries give the fraction of rejections at nominal size $5\%$ for the asymptotic IM test (IM), the bootstrap IM test (Boot), and the Wald test (Wald).

Several patterns recur across the three designs. The infeasible target likelihood behaves like a correctly specified cross-sectional IM test, since under the null its limiting law is central $\chi^2$ without requiring restrictions on the relative growth rates of $n$ and $T$. Asymptotic critical values therefore work well once $n$ is large enough. For instance, at $n=1000$, the target PIM rejection rates are close to the nominal level at every $T$ in the static designs. In the dynamic probit, the rates remain mildly elevated, which is consistent with the remainder terms in Lemma \ref{lemma-scaled-indicator}. In short panels the target test can still be liberal (up to $18.5\%$ in static probit and $28.6\%$ in dynamic probit at $T=5$ and $n=125$). However, the parametric bootstrap restores size close to $5\%$ throughout (e.g.\ $5.6\%$, $5.2\%$, and $6.6\%$ at that grid point in static logit, static probit, and dynamic probit), as observed for pure cross-sections in \citet{horowitz1994bootstrap}. In static logit, the conditional PIM tracks the target at every $(n,T)$ in Table~\ref{tab:logitrejectionrates}, consistent with the role of $\ell^\frakt$ as a placeholder for feasible pseudo-likelihoods in Section~\ref{sec:pseudo-lik}.

Turning to the profile PIM, asymptotic $\chi^2$ critical values can be severely oversized when $\rho>0$ is plausible, as in static logit with $T=5$, where rejection rises from $10.4\%$ at $n=125$ to $100\%$ at $n=1000$. Size distortions vanish much faster as $T$ increases than for Wald inference on $\theta$ based on the same profile estimator and asymptotic $\chi^2$ critical values. Holding $n=1000$ fixed in static logit, asymptotic profile rejection falls from $100\%$ at $T=5$ to $5.8\%$ at $T=50$, whereas Wald rejection remains $53\%$ at $T=30$ and $34\%$ at $T=50$. Moderately long panels can therefore support a reliable specification test before conventional asymptotic inference on $\theta$ becomes trustworthy.

When $T$ is still too short for asymptotic profile critical values, the parametric bootstrap removes most of the incidental-parameter size distortion (static logit, $T=5$, $n=500$: $98.9\%$ to $8.7\%$). At $T=5$ it does not fully restore the nominal level when $n$ is large. For instance, in static probit, profile bootstrap rejection rises with $n$ from $8.3\%$ at $n=125$ to $20.6\%$ at $n=1000$, and in dynamic probit it is $11.8\%$ at $n=1000$. By $T=10$ the same bootstrap is close to $5\%$ in all three designs, while asymptotic profile rejection at $n=1000$ is still $83.9\%$ in static logit.

\noindent
Power against misspecification is reported in Appendix~\ref{sec:add-sim-results}. Under the unit-scale heteroskedastic design of Subsection~\ref{sec:add-sim-hetero}, the PIM has strong power once $T$ is only moderately large. In static logit, bootstrap target rejection already reaches $98.2\%$ at $T=10$ and $n=500$, and $100\%$ at $T=10$ and $n=1000$; at the same grid points the Wald test based on $\thetahatp$ remains far weaker ($38.7\%$ and $60.6\%$). Even at $T=5$ and $n=1000$, bootstrap target rejection is $70.1\%$. Dynamic probit follows the same pattern once $T$ is moderate (profile Boot $92.7\%$ and target Boot $80.2\%$ at $T=10$ and $n=500$), though short-$T$ power is weaker than in the static designs. The omitted-interaction design of \citet{horowitz1994bootstrap} (Subsection~\ref{sec:add-sim-omitted}) is likewise powerful at moderate $T$: in static logit with $T=15$ and $n=1000$, bootstrap target rejection is $80.0\%$, and in static probit with $T=10$ and $n=500$ it is $72.4\%$. The wrong-link design (Subsection~\ref{sec:add-sim-link}) is harder to detect: for completed cells with $T\ge 10$, bootstrap rejection rates for the target and conditional PIM are moderate, typically between $10\%$ and $16\%$ at $n=1000$, reflecting that a link-function error is a relatively subtle deviation from the fitted logit model. Across these alternatives, the conditional likelihood again tracks the target closely in static logit, and correcting incidental-parameter size distortions through the bootstrap does not by itself create power where the correctly centered indicator has little signal.

\begin{landscape}
\begin{table}[p]
\centering
\caption{Rejection Rates for Static Logit Model}
\label{tab:logitrejectionrates}
\small
\setlength{\tabcolsep}{4.5pt}
\begin{tabular}{llcccccccccccccccc}
\toprule
 & & \multicolumn{3}{c}{$n=125$} & & \multicolumn{3}{c}{$n=250$} & & \multicolumn{3}{c}{$n=500$} & & \multicolumn{3}{c}{$n=1000$} \\
\cmidrule(lr){3-5} \cmidrule(lr){7-9} \cmidrule(lr){11-13} \cmidrule(lr){15-17}
$T$ & Method & IM & Boot & Wald & & IM & Boot & Wald & & IM & Boot & Wald & & IM & Boot & Wald \\
\midrule
5 & Profile & 0.1040 & 0.0876 & 0.4277 & & 0.5968 & 0.0896 & 0.7120 & & 0.9894 & 0.0872 & 0.9432 & & 1.0000 & 0.0967 & 0.9978 \\
 & Conditional & 0.0868 & 0.0586 & 0.0432 & & 0.0562 & 0.0564 & 0.0492 & & 0.0454 & 0.0606 & 0.0512 & & 0.0358 & 0.0627 & 0.0477 \\
 & Target & 0.0966 & 0.0556 & 0.0450 & & 0.0646 & 0.0546 & 0.0498 & & 0.0564 & 0.0594 & 0.0488 & & 0.0429 & 0.0597 & 0.0491 \\
\addlinespace
10 & Profile & 0.0436 & 0.0668 & 0.2436 & & 0.1086 & 0.0594 & 0.4274 & & 0.3816 & 0.0578 & 0.7202 & & 0.8394 & 0.0494 & 0.9542 \\
 & Conditional & 0.0930 & 0.0574 & 0.0524 & & 0.0742 & 0.0608 & 0.0490 & & 0.0520 & 0.0588 & 0.0506 & & 0.0436 & 0.0536 & 0.0516 \\
 & Target & 0.0968 & 0.0574 & 0.0530 & & 0.0796 & 0.0654 & 0.0464 & & 0.0540 & 0.0568 & 0.0512 & & 0.0498 & 0.0592 & 0.0518 \\
\addlinespace
15 & Profile & 0.0458 & 0.0658 & 0.1648 & & 0.0516 & 0.0574 & 0.2910 & & 0.1468 & 0.0616 & 0.5308 & & 0.4082 & 0.0540 & 0.8330 \\
 & Conditional & 0.0952 & 0.0566 & 0.0508 & & 0.0798 & 0.0618 & 0.0548 & & 0.0574 & 0.0602 & 0.0526 & & 0.0500 & 0.0580 & 0.0524 \\
 & Target & 0.0986 & 0.0612 & 0.0500 & & 0.0804 & 0.0642 & 0.0550 & & 0.0570 & 0.0602 & 0.0526 & & 0.0472 & 0.0566 & 0.0502 \\
\addlinespace
20 & Profile & 0.0544 & 0.0630 & 0.1310 & & 0.0460 & 0.0574 & 0.2368 & & 0.0764 & 0.0534 & 0.4092 & & 0.2080 & 0.0558 & 0.7140 \\
 & Conditional & 0.1024 & 0.0604 & 0.0464 & & 0.0674 & 0.0550 & 0.0512 & & 0.0500 & 0.0502 & 0.0524 & & 0.0548 & 0.0622 & 0.0512 \\
 & Target & 0.1032 & 0.0600 & 0.0456 & & 0.0674 & 0.0544 & 0.0502 & & 0.0486 & 0.0500 & 0.0528 & & 0.0544 & 0.0614 & 0.0490 \\
\addlinespace
25 & Profile & 0.0592 & 0.0590 & 0.1178 & & 0.0482 & 0.0626 & 0.1984 & & 0.0592 & 0.0558 & 0.3630 & & 0.1326 & 0.0642 & 0.6142 \\
 & Conditional & 0.0970 & 0.0590 & 0.0480 & & 0.0772 & 0.0616 & 0.0564 & & 0.0582 & 0.0546 & 0.0478 & & 0.0556 & 0.0608 & 0.0512 \\
 & Target & 0.0952 & 0.0600 & 0.0468 & & 0.0770 & 0.0608 & 0.0564 & & 0.0594 & 0.0534 & 0.0478 & & 0.0546 & 0.0628 & 0.0518 \\
\addlinespace
30 & Profile & 0.0686 & 0.0638 & 0.1142 & & 0.0514 & 0.0636 & 0.1626 & & 0.0528 & 0.0616 & 0.2870 & & 0.0914 & 0.0546 & 0.5256 \\
 & Conditional & 0.1050 & 0.0610 & 0.0534 & & 0.0784 & 0.0626 & 0.0502 & & 0.0552 & 0.0558 & 0.0536 & & 0.0486 & 0.0556 & 0.0518 \\
 & Target & 0.1038 & 0.0618 & 0.0546 & & 0.0760 & 0.0610 & 0.0486 & & 0.0560 & 0.0560 & 0.0534 & & 0.0480 & 0.0552 & 0.0520 \\
\addlinespace
40 & Profile & 0.0712 & 0.0564 & 0.0932 & & 0.0574 & 0.0630 & 0.1396 & & 0.0492 & 0.0600 & 0.2160 & & 0.0598 & 0.0542 & 0.4208 \\
 & Conditional & 0.0992 & 0.0626 & 0.0492 & & 0.0796 & 0.0632 & 0.0518 & & 0.0620 & 0.0570 & 0.0408 & & 0.0514 & 0.0582 & 0.0522 \\
 & Target & 0.0996 & 0.0578 & 0.0494 & & 0.0780 & 0.0648 & 0.0514 & & 0.0596 & 0.0568 & 0.0424 & & 0.0514 & 0.0568 & 0.0512 \\
\addlinespace
50 & Profile & 0.0768 & 0.0582 & 0.0780 & & 0.0530 & 0.0576 & 0.1118 & & 0.0454 & 0.0598 & 0.1954 & & 0.0582 & 0.0622 & 0.3424 \\
 & Conditional & 0.1000 & 0.0606 & 0.0494 & & 0.0756 & 0.0580 & 0.0518 & & 0.0618 & 0.0578 & 0.0482 & & 0.0562 & 0.0618 & 0.0506 \\
 & Target & 0.1004 & 0.0624 & 0.0486 & & 0.0772 & 0.0600 & 0.0508 & & 0.0614 & 0.0572 & 0.0476 & & 0.0580 & 0.0602 & 0.0510 \\
\bottomrule
\end{tabular}
\end{table}
\end{landscape}

\begin{landscape}
\begin{table}[p]
\centering
\caption{Rejection Rates for Static Probit Model}
\label{tab:probitrejectionrates}
\small
\setlength{\tabcolsep}{4.5pt}
\begin{tabular}{llcccccccccccccccc}
\toprule
 & & \multicolumn{3}{c}{$n=125$} & & \multicolumn{3}{c}{$n=250$} & & \multicolumn{3}{c}{$n=500$} & & \multicolumn{3}{c}{$n=1000$} \\
\cmidrule(lr){3-5} \cmidrule(lr){7-9} \cmidrule(lr){11-13} \cmidrule(lr){15-17}
$T$ & Method & IM & Boot & Wald & & IM & Boot & Wald & & IM & Boot & Wald & & IM & Boot & Wald \\
\midrule
5 & Profile & 0.0224 & 0.0826 & 0.7476 & & 0.2210 & 0.1232 & 0.9640 & & 0.9058 & 0.1640 & 0.9998 & & 1.0000 & 0.2062 & 1.0000 \\
 & Target & 0.1852 & 0.0522 & 0.0456 & & 0.1252 & 0.0550 & 0.0474 & & 0.0894 & 0.0618 & 0.0532 & & 0.0560 & 0.0550 & 0.0456 \\
\addlinespace
10 & Profile & 0.0272 & 0.0616 & 0.4930 & & 0.0738 & 0.0630 & 0.7852 & & 0.4148 & 0.0522 & 0.9748 & & 0.9386 & 0.0542 & 0.9998 \\
 & Target & 0.1426 & 0.0572 & 0.0504 & & 0.0948 & 0.0516 & 0.0474 & & 0.0654 & 0.0538 & 0.0534 & & 0.0462 & 0.0518 & 0.0490 \\
\addlinespace
15 & Profile & 0.0378 & 0.0658 & 0.3390 & & 0.0452 & 0.0608 & 0.6076 & & 0.1624 & 0.0586 & 0.8934 & & 0.5738 & 0.0490 & 0.9946 \\
 & Target & 0.1278 & 0.0604 & 0.0514 & & 0.0858 & 0.0588 & 0.0464 & & 0.0632 & 0.0530 & 0.0550 & & 0.0490 & 0.0576 & 0.0530 \\
\addlinespace
20 & Profile & 0.0470 & 0.0604 & 0.2602 & & 0.0352 & 0.0584 & 0.4840 & & 0.0838 & 0.0526 & 0.7758 & & 0.3030 & 0.0508 & 0.9802 \\
 & Target & 0.1176 & 0.0566 & 0.0504 & & 0.0860 & 0.0572 & 0.0540 & & 0.0612 & 0.0556 & 0.0484 & & 0.0552 & 0.0576 & 0.0436 \\
\addlinespace
25 & Profile & 0.0526 & 0.0590 & 0.2212 & & 0.0400 & 0.0610 & 0.3938 & & 0.0614 & 0.0546 & 0.6904 & & 0.1900 & 0.0586 & 0.9300 \\
 & Target & 0.1102 & 0.0606 & 0.0540 & & 0.0884 & 0.0584 & 0.0516 & & 0.0618 & 0.0586 & 0.0498 & & 0.0596 & 0.0638 & 0.0476 \\
\addlinespace
30 & Profile & 0.0568 & 0.0576 & 0.1922 & & 0.0394 & 0.0584 & 0.3272 & & 0.0520 & 0.0572 & 0.5950 & & 0.1304 & 0.0586 & 0.8910 \\
 & Target & 0.1172 & 0.0632 & 0.0542 & & 0.0826 & 0.0584 & 0.0480 & & 0.0658 & 0.0572 & 0.0514 & & 0.0604 & 0.0648 & 0.0520 \\
\addlinespace
40 & Profile & 0.0614 & 0.0536 & 0.1556 & & 0.0412 & 0.0546 & 0.2504 & & 0.0420 & 0.0540 & 0.4636 & & 0.0848 & 0.0598 & 0.7600 \\
 & Target & 0.1062 & 0.0540 & 0.0566 & & 0.0816 & 0.0552 & 0.0504 & & 0.0636 & 0.0588 & 0.0468 & & 0.0536 & 0.0578 & 0.0492 \\
\addlinespace
50 & Profile & 0.0766 & 0.0612 & 0.1230 & & 0.0488 & 0.0580 & 0.2050 & & 0.0464 & 0.0606 & 0.3800 & & 0.0654 & 0.0592 & 0.6724 \\
 & Target & 0.1136 & 0.0634 & 0.0516 & & 0.0842 & 0.0608 & 0.0538 & & 0.0678 & 0.0638 & 0.0492 & & 0.0534 & 0.0566 & 0.0470 \\
\bottomrule
\end{tabular}
\end{table}
\end{landscape}

\input{simulations/panel_im_test_dyn_correct_review_table.tex}

\section{Conclusion}\label{sec:conclusion}

This paper extends the Information Matrix test to panel data  models with fixed effects. Explicitly accounting for fixed effects allows us to separate functional misspecification from time-invariant unobserved heterogeneity. We show that  the indicator fluctuates at rate $\sqrt{n}$ while estimation error vanishes as both $n$ and $T$ increase. As a consequence, testing the likelihood specification is much more robust to incidental parameter bias than tests about the common parameter, as the panel information matrix test statistic is asymptotically unbiased under rectangular asymptotics whereas Wald tests remain biased and thus unreliable. Our Monte Carlo evidence confirms these theoretical predictions and indicates that the PIM is reliable in moderately long panels, even when Wald tests remain badly sized. Finally, our statistic can be computed from first- and second-order likelihood derivatives only, removing some of the computational burden that made the information matrix test impractical in cross-sectional designs.

For future research, we note that the rate argument developed here is not special to the second Bartlett identity. Hence, other likelihood-based specification diagnostics that aggregate within-individual moment conditions at rate $\sqrt{n}$ should inherit the same incidental-parameter robustness, potentially unifying a class of panel specification tests. Moreover, the computational simplicity of the PIM and its robustness to incidental parameters make it a promising tool for assessing the credibility of structural panel models, especially in cases where likelihoods are simulated or otherwise costly to differentiate \citep{rust1987,aguirregabiria-mira-2007,dearing-blevins-2025}.

\newpage

\appendix

\section{Assumptions}\label{sec-ass}
The following conventions are used in the assumptions and proofs. Subscripts indicate the level of aggregation: symbols without a subscript refer to sample averages across $i$ (e.g., $\ell(\theta)=n^{-1}\sumi\ell_i(\theta)$ and $D(\theta)=n^{-1}\sumi D_i(\theta)$); a subscript $i$ refers to individual $i$; and subscripts $i$ and $t$ together refer to individual $i$ in period $t$. Derivatives of the joint log-likelihood $(\theta,\alpha_i)\mapsto\elliT(\theta,\alpha_i)$ with respect to $\theta$ or $\alpha_i$ are indicated by letters in the index (e.g., $\elliTtheta(\theta,\alpha_i)=\partial_\theta\elliT(\theta,\alpha_i)$). Derivatives of a pseudo-likelihood $\ell_i(\theta)$ with respect to $\theta$ continue the numeric convention of Section~\ref{sec:pseudo-lik}: $\ell_{i,k}$ is the $k$-th derivative, with sample analogues $\ell_k=n^{-1}\sumi\ell_{i,k}$. Expectations and centered versions are $\lambda_i(\theta)=\E_0\ell_i(\theta)$ and $l_i(\theta)=\ell_i(\theta)-\lambda_i(\theta)$, and likewise at the sample level with $\lambda(\theta)$ and $l(\theta)$. Derivatives of the IM indicator $D_i$ follow the same numeric convention: $D_{i,k}=\partial^k_\theta D_i$ and $D_k=n^{-1}\sumi D_{i,k}$. We frequently omit the argument when evaluating at $\theta=\theta_0$.

At the period level, $\ell_{it,j}(\theta)$ denotes the period-$t$ contribution to the $j$-th derivative of $\ell_i(\theta)$, so that $\ell_{i,j}(\theta)=T^{-1}\sum_{t=1}^T\ell_{it,j}(\theta)$. Centered period contributions are $l_{it,j}(\theta)=\ell_{it,j}(\theta)-\E_0[\ell_{it,j}(\theta)]$ and $l_{i,j}(\theta)=T^{-1}\sum_{t=1}^T l_{it,j}(\theta)$. The same pattern extends to higher orders: $\ell_{it,k}(\theta)$ is the period-$t$ contribution to the $k$-th derivative, $\lambda_{it,k}(\theta)=\E_0[\ell_{it,k}(\theta)]$, and $\ell_{i,k}(\theta)=T^{-1}\sum_{t=1}^T\ell_{it,k}(\theta)$. Mixed partials of the joint log-likelihood are written $\ell_{i,k,m}(\theta,\alpha_i)$ for order $k$ in $\theta$ and order $m$ in $\alpha_i$, with $\lambda_{i,k,m}=\E_0\ell_{i,k,m}$. Superscripts $^c$ and $^o$ are reserved for corrected and oracle variants when needed; they are not used in the main results.

This section collects the assumptions for the asymptotic theory. Recall $\Iio\defined (Y_{i,1-L},\ldots,Y_{i0})$ and write $
\calG_i\defined\sigma\bigl((\Xit)_{t\in\N},\,\Iio,\,\alphaio\bigr)$. For each $i\in\N$ and $m\in\N$, define the conditional $\alpha$-mixing coefficient
\[
\mfa_i(m)\defined
\sup_{t\in\N}\;
\sup_{\substack{A\in\sigma(Y_{i1},\ldots,Y_{it})\\
B\in\sigma(Y_{i,t+m},Y_{i,t+m+1},\ldots)}}
\bigl|\Pr(A\cap B\mid\calG_i)-\Pr(A\mid\calG_i)\Pr(B\mid\calG_i)\bigr|
\]
(cf.\ \citealp{prakasa-rao-conditional}, Definition~4), with the convention $\mfa_i(m)\defined 1$ for $m\le 0$.

\begin{ass}[Sampling and dependence]\label{ass-sampling}
\textup{(i)}~The random variables $((\Yit,\Xit)_{t \in \N}, \Iio, \alphaio)$ are i.i.d.~for $i\in\N$\textup{;}
\textup{(ii)}~For each $i\in\N$, the distribution of $((\Yit,\Xit)_{t \in \N}, \Iio, \alphaio)$
is strictly stationary\textup{;}
\textup{(iii)}~Geometric mixing. The sequence $(\Yit)_{t\in\N}$ is conditionally $\mfa_i$-mixing given $\calG_i$, and there exists $a\in(0,1)$ such that
$
\sup_{i\in\N}\mfa_i(m)\overset{\wpone}{\lesssim} a^{m}
\qquad\text{for all }m\in\N.
$
\textup{;}
\textup{(iv)}~There exists an integer $L_*\ge 1$,
independent of $i$, $n$, $T$, and of $j\in\{1,\ldots,4\}$, such that each
centered period contribution $l_{it,j}$ (and each centered mixed-joint period
contribution $l_{it,k,m}$ used below) is, given $\calG_i$, a measurable
function of $(Y_{i,t-L_*+1},\ldots,Y_{it})$.
\end{ass}

Part~\textup{(i)} requires independence across individuals and a common distribution, which is standard in microeconometric applications. The identical-distribution part can be relaxed to allow for cross-sectional heterogeneity at the cost of more complicated notation.
Part~\textup{(ii)} rules out time effects $\gamma_t$ and other calendar-time parameters in the process $(\Yit,\Xit)_{t \in \N}$. Allowing an infinite number of such time effects in addition to the fixed effects would substantially increase the technical complexity of the proofs
(cf., e.g., \citealp{fernandez-val-weidner}) while obscuring the main contribution of the paper.
Part~\textup{(iii)} is the usual geometric conditional $\alpha$-mixing restriction on the primitive outcomes used by
\citet{hahn-kuersteiner-nonlinear-panels} and
\citet{fernandez-val-weidner}.
Part~\textup{(iv)} holds in the static and finite-lag dynamic models considered here since period likelihood contributions (and their $\theta$-derivatives) depend on at most finitely many lags of $Y_{it}$, so one may take $L_*$ of the order of the lag length already encoded in $\Iio$.
Together with the moment envelope in Assumption~\ref{ass-moments}\textup{(i)}, parts~\textup{(iii)}--\textup{(iv)}
imply that the centered period contributions are themselves geometrically conditionally mixing in the sense used by
\citet{hahn-kuersteiner-nonlinear-panels},
and yield cumulant summability of those contributions
(Lemma~\ref{lemma-cumulants}) which is needed in Lemma~\ref{lemma-mixing}. Details of the mixing-to-cumulant argument are collected in
Appendix~\ref{sec-intermediate}.

\begin{ass}[Moments and derivative bounds]\label{ass-moments}
Omit the argument $\theta_0$ when evaluating at $\theta_0$.
\textup{(i)}~Ninth-moment envelope. There exists a $\calG_i$-measurable function
$d_i \ge 0$ of $((\Xit)_{t\in\N},\Iio,\alphaio)$ such that, for all $j\in\{1,\ldots,4\}$,
$\sup_{t \in \N}\E_0\bigl[|l_{it,j}|^9\bigr] \overset{\wpone}{\le} d_i$
with $\E[d_i^{4}]<\infty$. The same envelope is imposed on the centered period contributions $l_{it,k,m}$ to the mixed joint partials for $k\in\{0,\ldots,4\}$ and $m\in\{1,\ldots,4\}$.
\textup{(ii)}~Uniform $L^{16}$ envelopes. There exists a function
$D_i \ge 0$ of $((\Xit)_{t\in\N},\Iio,\alphaio)$ such that, for $k\in\{1,\ldots,7\}$,
\begin{equation}\label{eq:uniform-bound}
\E_0\Bigl[\sup_{\theta\in\Theta}|\ell_{it,k}(\theta)|^{16}\Bigr] \overset{\wpone}{\le} D_i
\qquad\text{and}\qquad
\sup_{\theta\in\Theta}|\lambda_{it,k}(\theta)|^{16} \overset{\wpone}{\le} D_i,
\end{equation}
with $\E[D_i^{4}]<\infty$.
For the joint loglikelihood $\elliT(\theta,\alpha_i)$ and integers $k\in\{0,\ldots,10\}$, $m\in\{1,\ldots,10\}$, write $\ell_{i,k,m}(\theta,\alpha_i)$ for the mixed partial of order $k$ in $\theta$ and order $m$ in $\alpha_i$, and $\lambda_{i,k,m}(\theta,\alpha_i)\defined \E_0\ell_{i,k,m}(\theta,\alpha_i)$. Then
\begin{equation}\label{eq:uniform-bound-mixed}
\E_0\Bigl[\sup_{\theta\in\Theta,\,\alpha_i\in\calJ}|\ell_{i,k,m}(\theta,\alpha_i)|^{16}\Bigr] \overset{\wpone}{\le} D_i
\qquad\text{and}\qquad
\sup_{\theta\in\Theta,\,\alpha_i\in\calJ}|\lambda_{i,k,m}(\theta,\alpha_i)|^{16} \overset{\wpone}{\le} D_i.
\end{equation}
\textup{(iii)}~Neighborhood nondegeneracy. There exist a neighborhood $N\subset\Theta$ of $\theta_0$ and a constant $c>0$ such that
$\inf_{\theta\in N}|\lambda_{2}(\theta)|\ge c$,
$\inf_{\theta\in N}|\ell_{2}(\theta)|\ge c$, and
$\inf_{\theta\in N}|\lambda_{i,0,2}(\theta,\alpha_i^\frakt(\theta))|\ge c$.
\end{ass}

Part~\textup{(i)} supplies the integrable envelope $d_i$ used in product-moment arguments and in Lemma~\ref{lemma-cumulants} (the order strictly above~$8$ ensures a positive geometric mixing exponent for eighth-order cumulants).
Part~\textup{(ii)} is parallel in form, as a covariate/fixed-effect envelope $D_i$ bounds the $\E_0$-moment of $\sup_{\theta}|\ell_{it,k}|$ (and of the mixed partials), while the corresponding $\lambda$ bounds hold almost surely given $\calG_i$ because $\lambda_{it,k}$ and $\lambda_{i,k,m}$ are already $\E_0$-measurable.
These envelopes imply the usual $\Op(1)$ bounds by Markov's inequality and provide $\calG_i$-measurable dominates for remainder numerators. The seventh pure-$\theta$ order supports dominated differentiability.
Part~\textup{(iii)} keeps the $\theta$-Hessian of the pseudo-likelihood and the $\alpha_i$-Hessian of the joint likelihood bounded away from zero near $\theta_0$.
Compactness of $\Theta$ and $\calJ$ is taken from Assumption~\ref{ass-estimators}. Parts~\textup{(ii)}--\textup{(iii)} are stated over these compact sets for notational convenience, though the expansions only use them on a neighborhood of $(\theta_0,\alphaio)$.

The next high-level assumption allows us to deal with the remainder terms in various expansions.
It is in the spirit of \citet[Assumptions~2.3 and C.9]{schumann-info} and of high-level remainder conditions used elsewhere in the panel literature (e.g., \citealp{hahn-newey-nonlinear-panels}).

\begin{ass}\label{ass-remainder}
Taking the derivative or the expectation of a remainder term does not alter its rate
\textup{(}as $T\to\infty$\textup{)}.
\end{ass}
The remainders that arise in the proofs are of the form \(R_i(\theta)=T^{-k}r_i(\theta)\), where \(r_i(\theta)\) is a sum of ratios whose numerators are products of (centered) likelihood derivatives and whose denominators are Hessian (or information) blocks.
Assumption~\ref{ass-moments} supplies the ingredients for showing Assumption \ref{ass-remainder} under more primitive conditions, but at the cost of substantial algebra that does not yield further insight. To make this more explicit, note that the envelope $D_i$ in part~\textup{(ii)} yields $\calG_i$-measurable $L^{16}$ dominates for the ratio numerators (and, given $\E[D_i^{4}]<\infty$, unconditional integrability), part~\textup{(iii)} keeps denominators bounded away from zero near $\theta_0$, and the extra pure-$\theta$ order in part~\textup{(ii)} supports dominated differentiability of $r_i$ on compact $\Theta$ (Assumption~\ref{ass-estimators}). Sharp rates for products of centered terms under $\E_0$ are obtained from Assumptions~\ref{ass-sampling} and~\ref{ass-moments}\textup{(i)} via Lemmas~\ref{lemma-cumulants} and~\ref{lemma-mixing}, not from the $L^{16}$ envelopes alone. Stationarity in Assumption~\ref{ass-sampling}\textup{(ii)} ensures that the argument of \(r_i\) does not drift with \(n,T\).
It thus remains to verify case by case that differentiation and passage of \(\E_0\) through each such ratio leave the \(T\)-rate unchanged, which is algebraically tedious without offering deeper insight. Consequently, we record  Assumption~\ref{ass-remainder} as a high-level assumption, which is common in the literature (see, for instance, \citealp{hahn-newey-nonlinear-panels}).

The following assumption specifies identification conditions. Moreover, it imposes high-level conditions on the consistency and rate of convergence of $\thetahat$, which can be justified under similar regularity conditions as those imposed here.
\begin{ass}\label{ass-estimators}
The common parameter $\theta_0$ lies in the interior of a compact and convex parameter space $\Theta$. For every $i$, $\alphaio$ is a random variable that, with probability one, lies in the interior of a compact interval $\calJ\subset\R$. For a generic $i$, the pair $(\theta_0,\alphaio)$ is identified as a well-separated maximizer of the individual joint likelihood criterion, i.e., for every $\varepsilon>0$,
\begin{equation}\label{eq:joint-id}
\Bigl(
\E_0\bigl[\elliT(\theta_0,\alphaio)\bigr]
-
\sup_{\substack{(\theta,\alpha_i)\in\Theta\times\calJ\\
\|(\theta,\alpha_i)-(\theta_0,\alphaio)\|>\varepsilon}}
\E_0\bigl[\elliT(\theta,\alpha_i)\bigr]
\Bigr)
>0,
\end{equation}
where $\|\cdot\|$ denotes the Euclidean norm. The estimator $\thetahat$ is consistent for $\theta_0$ as $n,T\to\infty$.
Furthermore, as $n,T\to\infty$,
\begin{equation}\label{eq:rates-estimators}
\sqrt{nT}(\thetahatp-\theta_0)=\Op(1)+\Op(\sqrt{\frac{n}{T}}) \mbox{ and }
\sqrt{nT}(\thetahatstar-\theta_0)=\Op(1).
\end{equation}
\end{ass}

Condition~\eqref{eq:joint-id} is the individual joint identification assumption of \citet[Assumption~4]{kim-sun2016} or  \citealp{higgins2024bootstrap}, specialized to $\E_0$.
Because the sampling design is i.i.d.\ across~$i$, an $\inf_i$ is unnecessary, so it is enough to state the separation for a generic individual.
Profiling out $\alpha_i$ at $\alphai^\frakt(\theta)$ recovers well-separatedness of $\theta_0$ for the target likelihood, so~\eqref{eq:joint-id} is strictly stronger than target-only identification.
The stronger form is the primitive used by \citet[Theorem~1]{kim-sun2016} for uniform consistency of the fixed-effect estimators, $\max_i|\alphahat_i-\alphaio|=\op(1)$. Under Assumption~\ref{ass-n-T-rate}\textup{(ii)} their covering-plus-union bound continues to vanish when $n/T^2\to\rho\in[0,\infty)$.
Consistency of $\thetahat$ and the rates of convergence of $\sqrt{nT}(\thetahat-\theta_0)$ in \eqref{eq:rates-estimators} for the profile and target likelihood can be derived under our conditions: for profile likelihood estimators in static index panels, see \citet{hahn-newey-nonlinear-panels}; in dynamic single-index panels with predetermined regressors, see \citet{fernandez-val-probit}; for the target likelihood, see \citet{schumann-info} and \citet{schumann2020joe}.
Since deriving these rates is not the focus of this paper, they are stated as Assumption~\ref{ass-estimators} here.

The following assumption restricts the rates at which $n$ and $T$ can grow.

\begin{ass}\label{ass-n-T-rate}
\textup{(i)}~$n$ grows at least as fast as $T$, i.e., $\lim_{n,T\to \infty} n/T \in (0, \infty]$\textup{;}
\textup{(ii)}~$n/T^2\to\rho$, where $\rho\in[0,\infty)$.
\end{ass}
Assumption (i) models micro-economic panels where $n$ is larger than $T$.  This is mainly used to simplify remainder terms. Assumption (ii) indexes ``parabolic'' asymptotics, as the profile IM bias vanishes when $\rho=0$ and stabilizes when $\rho\in(0,\infty)$.

\begin{ass}\label{ass-variance}
    The limiting variance $V\defined \lim_{n,T}\var_0(\sqrt{n}D^\frakt(\theta_0))$ exists and is nonsingular.
\end{ass}
This assumption is the panel equivalent of Assumption~A10 in \citet{white1982maximum}, which requires the asymptotic variance of the IM indicators to be nonsingular. A notable difference is that \citet{white1982maximum} imposes this condition on the indicators evaluated at the QMLE, $D(\thetahat)$, whereas here $V$ is defined for the indicators evaluated at the true parameter $\theta_0$. This distinction reflects that, in the panel setting, the contribution of parameter estimation error is asymptotically negligible relative to the sampling variability of the indicators---in contrast to the cross-sectional case, where estimation uncertainty enters the limiting variance at the same order.

\subsection{Assumptions for the parametric bootstrap}\label{sec:ass-bootstrap}
In this subsection, we state neighborhood-uniform versions of the $\E_0$ sampling primitives under which Theorem~\ref{thm:bootstrap} follows. The bootstrap draws outcomes from the fitted profile transition densities at $(\thetahat^\frakp,\{\alphahat_i\}_{i=1}^n)$, where $\alphahat_i=\alphahatiT(\thetahat^\frakp)$, conditional on covariates and initial conditions. Write $P^*$ for the corresponding conditional bootstrap measure, and $\Opstar$ / $\opstar$ for stochastic orders under $P^*$. As in Section~\ref{sec:bootstrap}, $\dtostar$ denotes weak convergence under $P^*$, understood as holding in probability under $P_0$. Write $\tau=(\theta,\alpha_1,\ldots,\alpha_n)$ for a generic value of the joint parameter and $\tau_0=(\theta_0,\{\alphaio\}_{i=1}^n)$ for the true value. Let $\calU\defined\Theta\times\calJ^n$ denote the joint parameter space, and let $\calU_0\subset\calU$ be a subset whose interior contains $\tau_0$. Following \citet{higgins2024bootstrap}, define the $\delta$-neighborhood
\[
\calU_1\defined\bigl\{\tau\in\calU:d(\tau,\calU_0)<\delta\bigr\}
\]
for some $\delta>0$, where $d(\tau,\calU_0)\defined\inf\{\|\tau-\vartheta\|_2:\vartheta\in\calU_0\}$. For $\tau\in\calU$, let $\E_\tau$ denote integration with respect to the joint density
\[
\prod_{i=1}^n\prod_{t=1}^T f_{it}(\Yit\mid Y_{it-1},\ldots,Y_{it-L},\calXiT;\theta,\alpha_i),
\]
conditional on covariates and initial conditions. When likelihood derivatives are evaluated at the components of~$\tau$, centered period contributions are $l_{it,j}(\tau)\defined \ell_{it,j}(\tau)-\E_\tau[\ell_{it,j}(\tau)]$, and $\E_\tau$ replaces $\E_0$ throughout Appendix~\ref{sec-intermediate}.\footnote{This notation matches the $\mathbb{E}_\theta$-convention in \citet[Section~5]{higgins2024bootstrap}.} Under $P^*$, one has $\tau=(\thetahat^\frakp,\{\alphahat_i\}_{i=1}^n)$. In Remark~\ref{rem:bootstrap-primitives}, we comment on the uniform strengthening of the assumptions.

The next assumption imposes geometric mixing, derivative envelopes, and information-matrix nondegeneracy uniformly on a neighborhood of the true value, similar to Assumptions~2,~3, and~6 in \citet{higgins2024bootstrap}.
\begin{ass}[Uniform local regularity]\label{ass-bootstrap-uniform}
The analogues of Assumptions~\ref{ass-sampling}\textup{(iii)}, \ref{ass-moments}, and~\ref{ass-remainder} hold with $\E_0$ replaced by $\E_\tau$, \emph{uniformly} over $\tau\in\calU_1$ in the almost-sure / $\calG_i$-conditional sense of Appendix~\ref{sec-intermediate}. In particular, for each~$i$:
\begin{enumerate}
\item[(i)]~$(\Yit)_{t\in\N}$ is conditionally $\mfa_i$-mixing given $\calG_i$ under $\E_\tau$, with $\sup_{\tau\in\calU_1}\mfa_i(m)\overset{\wpone}{\lesssim} a^{m}$ for some $a\in(0,1)$\textup{;}
\item[(ii)]~the ninth- and sixteenth-moment envelopes of Assumption~\ref{ass-moments} hold with $\E_0$ replaced by $\E_\tau$, uniformly over $\tau\in\calU_1$\textup{;}
\item[(iii)]~the Hessian nondegeneracy bounds of Assumption~\ref{ass-moments}\textup{(iii)} hold uniformly over $\tau\in\calU_1$\textup{;}
\item[(iv)]~differentiation and passage of $\E_\tau$ through the remainder terms of Assumption~\ref{ass-remainder} preserve the stated $T$-rates, uniformly over $\tau\in\calU_1$.
\end{enumerate}
\end{ass}
The next assumption is the bootstrap counterpart of the profile rate in \eqref{eq:rates-estimators}.
\begin{ass}[Bootstrap estimator rates]\label{ass-bootstrap-estimators}
As $n,T\to\infty$,
\begin{equation}\label{eq:rates-bootstrap-estimators}
\sqrt{nT}(\thetahat^{\frakp*}-\thetahat^\frakp)=\Opstar(1)+\Opstar\Bigl(\sqrt{\frac{n}{T}}\Bigr),
\end{equation}
in probability under $P_0$.
\end{ass}
For a justification, see Remark~\ref{rem:bootstrap-primitives}. We record the rate as a high-level condition because the PIM is not a smooth functional of the MLE alone.

\begin{ass}[Bootstrap indicator transfer]\label{ass-bootstrap-indicator}
For every $\tau\in\calU_1$, the conclusions of Lemmas~\ref{lem:indicators} and~\ref{lemma-scaled-indicator} hold under $\E_\tau$ with the same $T$-rates as under $\E_0$ at $\tau=\tau_0$, and with constants that are uniform over $\tau\in\calU_1$. In particular, under $P^*=\E_{\hat\tau}$ at $\hat\tau=(\thetahat^\frakp,\{\alphahat_i\}_{i=1}^n)$, these expansions hold for the bootstrap profile indicator $D^{\frakp*}$ at $\thetahat^{\frakp*}$ and at $\thetahat^\frakp$. Here $\alphai^\frakt(\theta)$ is the target incidental parameter around which $\alphahatiT(\theta)$ is expanded inside that profile expansion.
\end{ass}
See Remark~\ref{rem:bootstrap-primitives}. Lemma~\ref{lemma-bootstrap-clt} uses this transfer together with Assumption~\ref{ass-bootstrap-uniform}. Lemma~\ref{lemma-bootstrap-variance} also uses Assumption~\ref{ass-bootstrap-estimators}.

\begin{ass}[Bootstrap bias continuity]\label{ass-bootstrap-bias}
For the profile likelihood, let $\kappa_T(\theta)\defined B_T^{\mathrm{score}}(\theta)+B_T^{\mathrm{info}}(\theta)$ denote the finite-$T$ score- and information-bias functional from the proof of Lemma~\ref{lemma-scaled-indicator}\textup{(ii)}, and let $\kappa_T^*(\thetahat^\frakp)\defined B_T^{\mathrm{score}*}+B_T^{\mathrm{info}*}$ be its bootstrap counterpart at the fitted profile estimator. Then
\[
\kappa_T^*(\thetahat^\frakp)-\kappa_T(\theta_0)=\op(1).
\]
\end{ass}
See Remark~\ref{rem:bootstrap-primitives}. Under Assumption~\ref{ass-n-T-rate}\textup{(ii)}, this yields $(\sqrt{n}/T)\kappa_T^*(\thetahat^\frakp)\to_p\sqrt{\rho}\,\kappa(\theta_0)$ and matches the noncentrality of Theorems~\ref{main-theorem} and~\ref{thm:bootstrap}.

\begin{rem}\label{rem:bootstrap-primitives}
The conditions above justify the profile bootstrap in Theorem~\ref{thm:bootstrap}. They are neighborhood-uniform strengthenings of the sampling primitives already used under~$\E_0$, rather than qualitatively new requirements. A full repetition of the intermediate lemmas under a generic law~$\E_\tau$ is omitted; the $T$-orders are the same once Assumption~\ref{ass-bootstrap-uniform} holds.

Assumption~\ref{ass-bootstrap-uniform} asks that geometric conditional mixing, the moment envelopes and Hessian nondegeneracy of Assumption~\ref{ass-moments}, and the remainder algebra of Assumption~\ref{ass-remainder} hold with~$\E_0$ replaced by~$\E_\tau$, uniformly for~$\tau\in\calU_1$.
This is the same  as Assumptions~\ref{ass-sampling}\textup{(iii)}, \ref{ass-moments}, and~\ref{ass-remainder}, localized to a $\delta$-neighborhood of~$\tau_0$ in the sense of \citet{higgins2024bootstrap}.
For smooth transition densities on compact parameter sets, these properties are typically continuous in~$(\theta,\alpha_i)$. Imposing them uniformly on~$\calU_1$ is then a mild strengthening of the pointwise conditions at~$\tau_0$.

Assumption~\ref{ass-bootstrap-indicator} records the transfer of Lemmas~\ref{lem:indicators} and~\ref{lemma-scaled-indicator} to~$\E_\tau$.
Those proofs expand $\alphahatiT(\theta)$ around $\alphai^\frakt(\theta)$ and obtain score- and information-bias formulae whose expectations are taken under the law that generates the outcomes.
Under~$\E_\tau$ the same algebra applies with expansion point equal to the common-parameter component of~$\tau$ and with~$\E_0$ replaced by~$\E_\tau$.
The leading bias terms $B_{T,1}^{\mathrm{score}}$, $B_{T,2}^{\mathrm{score}}$, and~$B_T^{\mathrm{info}}$ are rational functions of a fixed finite list of moments of likelihood derivatives.
Under Assumption~\ref{ass-bootstrap-uniform}, the mixing and product-moment arguments that deliver these $T$-orders under~$\E_0$ continue to hold under~$\E_\tau$ with constants independent of~$\tau\in\calU_1$.
In particular, the  score-bias expansion, the remainders in the $\Delta$-expansion, the derivative-moment rates of Lemma~\ref{lemma-rates}, the linearization of~$\thetahat$ in~\eqref{eq:thetahat-linearization}, and the remainder algebra of Assumption~\ref{ass-remainder} all transfer once their $\E_0$ primitives are stated uniformly on~$\calU_1$.
Under the parametric bootstrap, $P^*$ coincides with~$\E_{\hat\tau}$ at $\hat\tau=(\thetahat^\frakp,\{\alphahat_i\})$, so the expansions hold at this fitted~$\hat\tau$.
Membership of~$\hat\tau$ in~$\calU_1$ with probability approaching one uses the joint identification condition~\eqref{eq:joint-id} and \citet[Theorem~1]{kim-sun2016}, which also holds under parabolic asymptotics (as argued after Assumption~\ref{ass-estimators}), together with consistency of~$\thetahat^\frakp$ as in the neighborhood device of \citet{andrews2005bootstrap} used by \citet{higgins2024bootstrap}.
Lemmas~\ref{lemma-bootstrap-clt} and~\ref{lemma-bootstrap-variance} then deliver the centered Lyapunov CLT and variance consistency under~$P^*$. The fourth-moment bound uses the uniform envelopes in Assumption~\ref{ass-bootstrap-uniform} together with the indicator transfer in Assumption~\ref{ass-bootstrap-indicator}. Variance consistency also uses the profile rate in Assumption~\ref{ass-bootstrap-estimators}.

Assumption~\ref{ass-bootstrap-bias} ensures that the bootstrap and sampling noncentralities coincide when~$\rho>0$.
Because the moments that enter $\kappa_T$ are continuous in the parameter and obey a uniform LLN on a neighborhood of the truth, plugging in the consistent bootstrap fit $\hat\tau$ does not change the limiting bias, so $\kappa_T^*(\thetahat^\frakp)-\kappa_T(\theta_0)=\op(1)$.

Assumption~\ref{ass-bootstrap-estimators} is the bootstrap counterpart of the profile rate in~\eqref{eq:rates-estimators}.
\citet{higgins2024bootstrap} obtain the analogous rate from uniform weak convergence of the MLE on their neighborhood (via \citealp[Theorem~1]{kim-sun2016}) and Lemma~A.1 of \citet{andrews2005bootstrap}. While we state~\eqref{eq:rates-bootstrap-estimators} directly as a high-level condition (as under~$\E_0$), this rate can be derived from identification, uniform mixing and moment bounds, and Hessian nondegeneracy on~$\calU_1$, along the lines of \citet{hahn-newey-nonlinear-panels}, \citet{fernandez-val-probit}, and \citet{higgins2024bootstrap}.
\end{rem}

\section{Intermediate results}\label{sec-intermediate}
\label{sec-sampling-primitives}
This section derives the necessary rates for expectations of products of $k$ centered terms given the mixing conditions imposed in the assumptions.
The first result establishes \emph{conditional} cumulant summability of the centered period contributions under geometric conditional mixing. The second converts that bound into sharp product-moment rates for centered likelihood derivatives by a conditional moment--cumulant expansion, taking the unconditional expectation only at the last step.
These rates are used throughout the proofs in Appendix~\ref{sec-proofs}. Relative to a crude H\"older bound $O(T^{-k/2})$, the rate $O(T^{-\lceil k/2\rceil})$ is sharper when $k$ is odd.
Below, $A\lesssim B$ means that $A\le CB$ for a constant $C$ independent of $n$ and $T$. Write $\mathrm{cum}(\,\cdot\mid\calG_i)$ for joint cumulants under the conditional law given $\calG_i=\sigma((\Xit)_{t\in\N},\Iio,\alphaio)$, and omit the argument $\theta_0$ in $l_{it,j}\defined l_{it,j}(\theta_0)$.

\begin{lem}[Conditional cumulant summability]\label{lemma-cumulants}
Let Assumptions~\ref{ass-sampling} and~\ref{ass-moments}\textup{(i)} hold. Then, for every $r\in\{2,\ldots,8\}$ and every $j_1,\ldots,j_r\in\{1,\ldots,4\}$, there is a finite constant $C_r$ (independent of $n,T,i$) such that
\begin{equation}\label{eq:ass-mixing-cum-sum}
\sum_{t_1=1}^T\cdots\sum_{t_r=1}^T
\bigl|\mathrm{cum}\bigl(l_{it_1,j_1},\ldots,l_{it_r,j_r}\mid\calG_i\bigr)\bigr|
\overset{\wpone}{\le} C_r(1+d_i)\,T.
\end{equation}
The same bound holds for the centered period contributions $l_{it,k,m}$ to the mixed joint partials of $\elliT(\theta,\alpha_i)$ for $k\in\{0,\ldots,4\}$ and $m\in\{1,\ldots,4\}$.
\end{lem}

\begin{proof}[\textupandbold{Proof of Lemma~\ref{lemma-cumulants}}]
In the following, we consider cumulants under the conditional law given $\calG_i$. 
As before, we omit the argument $\theta_0$ and write $l_{it,j}\defined l_{it,j}(\theta_0)$.
By Assumption~\ref{ass-sampling}\textup{(iv)}, conditioning on $\calG_i$ makes each $l_{it,j}$ a measurable function of a fixed-length window of $(\Yit)$. By the finite-memory mapping theorem (\citealp{davidson-book}, Theorem~14.1), geometric conditional mixing of $(\Yit)$ therefore transfers to $(l_{it,j})$: the sequences of centered period contributions are geometrically conditionally mixing given $\calG_i$ with the same rate $a$. The same holds for the mixed-joint contributions $l_{it,k,m}$.
Assumption~\ref{ass-moments}\textup{(i)} and $|x|^p\le\1(|x|\le 1)+|x|^9$ give the envelope
\begin{equation}\label{eq:moment-envelope-primitive}
\sup_{t\in\N}\E_0\bigl[|l_{it,j}|^p\bigr]
\overset{\wpone}{\le}1+d_i
\qquad\text{for all }0<p\le 9.
\end{equation}
For a generic $i$, fix $r\in\{2,\ldots,8\}$,  $j_1,\ldots,j_r$, and write
\[
S_{i,r}\defined
\sum_{t_1=1}^T\cdots\sum_{t_r=1}^T
\bigl|\mathrm{cum}\bigl(l_{it_1,j_1},\ldots,l_{it_r,j_r}\mid\calG_i\bigr)\bigr|.
\]
For each tuple, let $t_{(1)}\le\cdots\le t_{(r)}$ be the ordered times, set $\delta_\ell\defined t_{(\ell+1)}-t_{(\ell)}$ and $h\defined t_{(r)}-t_{(1)}$, and write
\[
q\defined\bigl(\max_{\ell}\delta_\ell-L_*+1\bigr)_+
\]
for the effective outcome separation, taking into account the number of lags $L_*$ in the likelihood specification.
By Assumption~\ref{ass-sampling}\textup{(iv)}, each centered period contribution $l_{it,j}$ depends on only a short window of $(\Yit)$. Ordering the times and cutting at the largest gap therefore splits the $r$ contributions into an early group and a late group whose underlying outcome windows are separated by $q$ periods. Conditional $\alpha$-mixing of $(\Yit)$ in Assumption~\ref{ass-sampling}\textup{(iii)} then bounds the dependence between these two groups by $\mfa_i(q)$. Applying the joint-cumulant bound of \citet{prakasa-rao-cumulant-mixing} (Theorem~1.1 with $\lambda=0$) under the conditional law given $\calG_i$ with moment order $9>r$ therefore yields
\[
\bigl|\mathrm{cum}(l_{it_1,j_1},\ldots,l_{it_r,j_r}\mid\calG_i)\bigr|
\overset{\wpone}{\lesssim}
\mfa_i(q)^{\nu_r}
\prod_{m=1}^{r}\|l_{it_m,j_m}\|_{9,\E_0},
\qquad\nu_r\defined 1-\tfrac{r}{9}\in(0,1),
\]
where $\|Z\|_{9,\E_0}\defined\bigl(\E_0[|Z|^9]\bigr)^{1/9}$. By~\eqref{eq:moment-envelope-primitive}, each factor satisfies $\|l_{it_m,j_m}\|_{9,\E_0}\overset{\wpone}{\lesssim}(1+d_i)^{1/9}$, so the product of $r\le 8$ norms is absorbed into $(1+d_i)$ and
\[
\bigl|\mathrm{cum}(l_{it_1,j_1},\ldots,l_{it_r,j_r}\mid\calG_i)\bigr|
\overset{\wpone}{\lesssim}
(1+d_i)\,\mfa_i(q)^{\nu_r}.
\]
If $q\ge 1$, then $\mfa_i(q)\lesssim a^{q}$, so
\[
\mfa_i(q)^{\nu_r}
\lesssim a^{\nu_r\max_{\ell}\delta_\ell}
\le a^{\nu_r h/(r-1)}
=a^{\gamma h}
=\prod_{\ell=1}^{r-1}a^{\gamma\delta_\ell},
\qquad\gamma\defined\frac{\nu_r}{r-1},
\]
using $\max_{\ell}\delta_\ell\ge h/(r-1)$ and $a\in(0,1)$.
If $q=0$, then $\max_{\ell}\delta_\ell\le L_*-1$, so $h\le(r-1)(L_*-1)$ and
$1\le a^{-\gamma(r-1)(L_*-1)}a^{\gamma h}$; the moment bound $|\mathrm{cum}|\overset{\wpone}{\lesssim}(1+d_i)$
therefore yields $|\mathrm{cum}|\overset{\wpone}{\lesssim}(1+d_i)a^{\gamma h}$ after absorbing
$a^{-\gamma(r-1)(L_*-1)}$ into the $\lesssim$ constant.
Summing over all $T^r$ tuples costs at most an $r!$ factor relative to the sum over nondecreasing tuples since there are at most $r!$ reorderings of any tuple.
Moreover, for each fixed gap vector $(\delta_1,\ldots,\delta_{r-1})$, there are at most $T$ admissible starting times $t_{(1)}$, hence at most $T$ nondecreasing tuples.
Therefore
\begin{align*}
S_{i,r}
&\le r!\sum_{\delta_1=0}^{T-1}\cdots\sum_{\delta_{r-1}=0}^{T-1}
T\,(1+d_i)\prod_{\ell=1}^{r-1}a^{\gamma\delta_\ell}\\
&=r!\,(1+d_i)\,T\Bigl(\sum_{\delta=0}^{T-1}a^{\gamma\delta}\Bigr)^{r-1}
\overset{\wpone}{\lesssim}(1+d_i)\,T,
\end{align*}
since $\sum_{\delta=0}^\infty a^{\gamma\delta}<\infty$ and $r\le 8$ is fixed. Consequently, \eqref{eq:ass-mixing-cum-sum} follows. Analogous arguments show the corresponding result for  $l_{it,k,m}$.
\end{proof}

\begin{lem}\label{lemma-mixing}
Let Assumptions~\ref{ass-sampling} and~\ref{ass-moments} hold. Then, for each $k\in\{1,\ldots,8\}$ and indices $j_1,\ldots,j_k\in\{1,\ldots,4\}$,
\begin{align}
\E_0\prod_{m=1}^{k} l_{i,j_m}(\theta_0)
&=O\bigl(T^{-\lceil k/2\rceil}\bigr),
\label{eq:mixing-prod-individual}\\
\E_0\prod_{m=1}^{k} l_{j_m}(\theta_0)
&=O\bigl((nT)^{-\lceil k/2\rceil}\bigr).
\label{eq:mixing-prod-sample}
\end{align}
\end{lem}

In particular, \eqref{eq:mixing-prod-sample} with $k=2$ yields $\E_0 \lk^2(\theta_0)=O((nT)^{-1})$. Since $\E_0 \lk(\theta_0)=0$ by construction, this implies
\begin{equation}\label{lbarnTk-rate}
\lk(\theta_0) = \Op\bigl((nT)^{-1/2}\bigr), \qquad k\in\{1, \ldots, 4\}.
\end{equation}
When the outcomes are independent across $t$ under $\E_0$, \eqref{eq:mixing-prod-individual}--\eqref{eq:mixing-prod-sample} reduce to the usual Wick/pairing count under time-independence.
By Lemma~\ref{lemma-cumulants}, the same rates hold for products of the centered mixed-joint averages $l_{i,k,m}$ with $k\in\{0,\ldots,4\}$ and $m\in\{1,\ldots,4\}$.

\begin{proof}[\textupandbold{Proof of Lemma~\ref{lemma-mixing}}]
For a generic individual $i$, fix integers $k\in\{1,\ldots,8\}$ and $j_1,\ldots,j_k\in\{1,\ldots,4\}$.
Omit $\theta_0$. By definition, $\E_0[l_{it,j_m}]=0$ and
\(
l_{i,j_m}=T^{-1}\sum_{t=1}^T l_{it,j_m}.
\)
Let $\Pi_k$ denote the set of all partitions of $\{1,\ldots,k\}$. For $\pi\in\Pi_k$ and a block $B\in\pi$, write
\[
\kappa_B\defined\mathrm{cum}\Bigl(\sum_{t=1}^T l_{it,j_m}:m\in B\Bigm|\calG_i\Bigr)
\]
for the joint cumulant conditional on $\calG_i$.
Since $\prod_{m=1}^{k}l_{i,j_m}=T^{-k}\prod_{m=1}^{k}\bigl(\sum_{t=1}^T l_{it,j_m}\bigr)$, the Leonov--Shiryaev moment--cumulant formula \citep{leonov-shiryaev} yields
\begin{equation}\label{eq:moment-cumulant-prod}
\E_0\prod_{m=1}^{k} l_{i,j_m}
=T^{-k}\sum_{\pi\in\Pi_k}\prod_{B\in\pi}\kappa_B.
\end{equation}
We have to show that for every block $B$ with $r\defined|B|\in\{1,\ldots,8\}$,
\begin{equation}\label{eq:cum-S-order}
|\kappa_B|\overset{\wpone}{\lesssim}(1+d_i)\,T.
\end{equation}
For $r=1$, $\kappa_B=\E_0\bigl[\sum_{t=1}^T l_{it,j_m}\bigr]=0$ by centering, which is stronger than~\eqref{eq:cum-S-order}. For $r\ge 2$, multilinearity of cumulants gives
\[
\kappa_B
=\sum_{t_1=1}^T\cdots\sum_{t_r=1}^T
\mathrm{cum}\bigl(l_{it_1,j_{m_1}},\ldots,l_{it_r,j_{m_r}}\mid\calG_i\bigr),
\]
where $\{m_1,\ldots,m_r\}=B$. Hence, by the triangle inequality and Lemma~\ref{lemma-cumulants},
\[
|\kappa_B|
\le
\sum_{t_1=1}^T\cdots\sum_{t_r=1}^T
\bigl|\mathrm{cum}\bigl(l_{it_1,j_{m_1}},\ldots,l_{it_r,j_{m_r}}\mid\calG_i\bigr)\bigr|\overset{\wpone}{\le} C_r(1+d_i)\,T,
\]
which is~\eqref{eq:cum-S-order}. Since partitions with block size one are zero by centering, we only need to consider partitions with block size of at least two. Consequently, for a partition $\pi$ with $\min_{B\in\pi}|B|\ge 2$, one has $|\pi|\le\lfloor k/2\rfloor\le 4$, and
\[
\biggl|T^{-k}\prod_{B\in\pi}\kappa_B\biggr|
\overset{\wpone}{\lesssim} (1+d_i)^{|\pi|}\,T^{|\pi|-k}
\le (1+d_i)^{4}\,T^{-\lceil k/2\rceil},
\]
using $|\pi|-k\le\lfloor k/2\rfloor-k=-\lceil k/2\rceil$.
Summing over the finitely many partitions in $\Pi_k$ and using~\eqref{eq:moment-cumulant-prod},
\[
\Bigl|\E_0\prod_{m=1}^{k} l_{i,j_m}\Bigr|
\overset{\wpone}{\lesssim} (1+d_i)^{4}\,T^{-\lceil k/2\rceil}.
\]
Taking the unconditional expectation and applying $\E[d_i^{4}]<\infty$ from Assumption~\ref{ass-moments}\textup{(i)} yields~\eqref{eq:mixing-prod-individual}. For the sample averages across $i$, expand
\[
\E_0\prod_{m=1}^{k} l_{j_m}
=n^{-k}\sum_{i_1=1}^n\cdots\sum_{i_k=1}^n
\E_0\prod_{m=1}^{k} l_{i_m j_m}.
\]
By cross-sectional independence in Assumption~\ref{ass-sampling}\textup{(i)},
the expectation factors across distinct individuals.
Each tuple $(i_1,\ldots,i_k)$ induces an equality partition $\rho$ of $\{1,\ldots,k\}$
with $q$ blocks of sizes $r_1,\ldots,r_q$ (so $\sum_b r_b=k$), and
\[
\E_0\prod_{m=1}^{k} l_{i_m j_m}
=\prod_{b=1}^{q}
\E_0\prod_{m\in B_b} l_{i^{(b)},j_m},
\]
where $i^{(b)}$ is the common label on block $B_b$.
If some $r_b=1$, the corresponding factor vanishes by centering, so only partitions
with $\min_b r_b\ge 2$ contribute, which implies that  $q\le\lfloor k/2\rfloor$.
For such a partition, there are at most $n^q$ assignments of the $q$ blocks to
distinct individuals, and~\eqref{eq:mixing-prod-individual} gives
\[
\Bigl|\E_0\prod_{m\in B_b} l_{i^{(b)},j_m}\Bigr|
=O\bigl(T^{-\lceil r_b/2\rceil}\bigr)
\]
inside each block. The contribution of $\rho$ is therefore
\[
O\bigl(n^{q-k}\,T^{-\sum_b\lceil r_b/2\rceil}\bigr)=
O\bigl(n^{q-k}\,T^{-\sum_b\lceil r_b/2\rceil}\bigr).
\]
To see the last equality, note that since $q-k\le\lfloor k/2\rfloor-k=-\lceil k/2\rceil$ and
$\sum_b\lceil r_b/2\rceil\ge\lceil k/2\rceil$
(because $r_b\ge 2$ and $\sum_b r_b=k$). Summing over the finitely many partitions $\rho$ yields~\eqref{eq:mixing-prod-sample}.
\end{proof}

\newcommand{\ellitalpha}{\ell_{it\alpha}}
\newcommand{\ellitalphasq}{\ell_{it\alpha\alpha}}
\newcommand{\ellitalphacube}{\ell_{it\alpha\alpha\alpha}}

\section{Proofs}\label{sec-proofs}
This section collects the proofs of the paper's main results.
In many instances, we have to show that each component of a vector or matrix is of a specific stochastic order. Since the same proof can be applied component-wise, we will often assume $\dim(\theta_0)=1$ to simplify notation.
The product-moment rates used below are those of Lemmas~\ref{lemma-cumulants} and~\ref{lemma-mixing} in Appendix~\ref{sec-intermediate}.

\subsection{Main proofs}
We begin by showing Lemma \ref{lemma-variance}.
\begin{proof}[\textupandbold{Proof of Lemma \ref{lemma-variance}}]
Let $\dim(\theta_0)=1$. Recall that $D=n^{-1}\sumi D_i$. Since the data is iid across individuals by Assumption \ref{ass-sampling}, the desired result follows if $\var_0(D_i)=\Op(1)$. To see this,  note that since $\E_0[D_i]=\IB$, where $\IB=0$ if $\ell$ is the target likelihood and $\IB=\Op(T^{-1})$ if $\ell$ is the profile likelihood,
\[
\var_0(D_i)=\E_0[T^2\ell_{i,1}^4+ 2T\ell_{i,1}^2\ell_{i,2}+\ell_{i,2}^2]-\IB^2.
\]
Using that $\ell_{i,1}=l_{i,1}+\lambda_{i,1}$ and $\ell_{i,2}=l_{i,2}+\lambda_{i,2}$, we get
\begin{align*}
\var_0(D_i)=& T^2 \left( \E_0 [ l_{i,1}^4  ] + 4\lambda_{i,1} \E_0 [ l_{i,1}^3  ] + 6\lambda_{i,1}^2 \E_0 [ l_{i,1}^2  ] + \lambda_{i,1}^4 \right) \\
&\quad + 2T \left( \E_0 [ l_{i,1}^2 l_{i,2}  ] + \lambda_{i,2} \E_0 [ l_{i,1}^2  ] + 2\lambda_{i,1} \E_0 [ l_{i,1} l_{i,2}  ] + \lambda_{i,1}^2 \lambda_{i,2} \right) \\
&\quad +  \E_0 [ l_{i,2}^2  ] + \lambda_{i,2}^2 -\IB^2.
\end{align*}
By Lemma~\ref{lemma-mixing}, expectations of products of $k$ centered likelihood terms are of order $O(T^{-\lceil k/2\rceil})$. Moreover, recall that $\lambda_{i,1}^\frakt=0$ and $\lambda_{i,1}^\frakp=\Op(T^{-1})$, whereas $\lambda_{i,2}=\Op(1)$ by Assumption \ref{ass-moments}. Therefore, 
\begin{align*}
    \var_0(D_i) &= T^2\E_0 [ l_{i,1}^4] + 2T (\E_0 [ l_{i,1}^2 l_{i,2} ] + \lambda_{i,2} \E_0 [ l_{i,1}^2]) + \E_0 [ l_{i,2}^2 ] + \lambda_{i,2}^2 && \text{(target)} \\
    &\quad + T^2(4\lambda_{i,1} \E_0 [ l_{i,1}^3 ] + 6\lambda_{i,1}^2 \E_0 [ l_{i,1}^2 ] + \lambda_{i,1}^4) + 2T(2\lambda_{i,1} \E_0 [ l_{i,1} l_{i,2} ] + \lambda_{i,1}^2 \lambda_{i,2}) - \IB^2, && \text{(profile)}
\end{align*}
where the first line contains all terms of order $\Op(1)$ that appear in the variance of the target indicator, whereas the second line contains additional terms of order $\Op(T^{-1})$ or lower due to the incidental parameter estimation in the profile likelihood. Overall, $\var_0(D_i)=\Op(1)$, which finishes the proof.
\end{proof}

The proofs of Lemmas \ref{lemma-scaled-indicator} and \ref{lemma-variance-estimator} require stochastic orders for the indicator and its derivatives. To establish these, we need the following Lemmata.

\begin{lem}\label{lemma-rates}
For $k=0,...,4$, let $D_{i,k}(\theta)$ denote the $k$-th derivative of $D_i(\theta)$, with $D_{i,0}(\theta)=D_i(\theta)$. Under Assumptions~\ref{ass-sampling}, \ref{ass-moments}, and Lemma~\ref{lemma-mixing}, 
\begin{align}
    \E_0[D_{i,1}(\theta_0)] &= \Op(1) \mbox{ and }\var_0(D_{i,1}(\theta_0))=\Op(T) \label{eq:rate_first_deriv} \\
    \sup_\theta |D_{i,k}(\theta)| &= \Op(T) \mbox{ for }k\in\{0,1,2,3,4\}. \label{eq:rate_second_deriv-b}\\
    \E_0[D_{i,k}(\theta_0)] &= \Op(T) \mbox{ and }\var_0(D_{i,k}(\theta_0))=\Op(T) \mbox{ for }k\in\{2,3,4\}.  \label{eq:rate_kth_deriv} 
\end{align}
\end{lem}
\begin{proof}[\textupandbold{Proof of Lemma \ref{lemma-rates}}]
    Recall that $D_i(\theta) = T \elliOne(\theta)^2 + \elliTwo(\theta)$ and take the first four derivatives of $D_i(\theta)$ with respect to $\theta$, which yields
    \begin{align*}
    D_{i,1}(\theta) &= 2T \elliOne(\theta) \elliTwo(\theta) + \ell_{i,3}(\theta), \\
    D_{i,2}(\theta) &= 2T \elliTwo(\theta)^2 + 2T \elliOne(\theta) \elliThree(\theta) + \ell_{i,4}(\theta), \\
    D_{i,3}(\theta) &= 6T \elliTwo(\theta)\elliThree(\theta) + 2T \elliOne(\theta)\ell_{i,4}(\theta) + \ell_{i,5}(\theta), \\
    D_{i,4}(\theta) &= 6T \elliThree(\theta)^2 + 8T \elliTwo(\theta)\ell_{i,4}(\theta) + 2T \elliOne(\theta)\ell_{i,5}(\theta) + \ell_{i,6}(\theta).
    \end{align*}
    Recall that $\ell_{i,k}(\theta)$ is a time-series average of individual- and period-specific likelihood derivatives that are each uniformly $\Op(1)$ in $\theta$ by Assumption \ref{ass-moments}. Due to the presence of the factor $T$ in the leading term, it thus follows that  $\sup_\theta|D_{i,k}(\theta)|=\Op(T)$. Since $D_{i,1}(\theta)$ depends on the score of the pseudo-likelihood, its rate changes when evaluated at the true parameter. Again by Assumption \ref{ass-moments}, $\ell_{3} = \Op(1)$.  By centering the pseudo-likelihood derivatives around its expectation, $\ellitOne=\lambda_{i,1}+l_{i,1}$ and $\elliTwo = \lambda_{i,2} + l_{i,2}$. Hence,
    \[
    2T \elliOne \elliTwo=2T(\lambda_{i,1}+l_{i,1})(\lambda_{i,2} + l_{i,2})
    =2T\Bigl(\lambda_{i,1}\lambda_{i,2} + \lambda_{i,1}l_{i,2}+l_{i,1}\lambda_{i,2}+l_{i,1}l_{i,2}\Bigr).
    \]
    Recall that $\lambda_{i,1}^\frakt=0$ and $\lambda_{i,1}^\frakp=\Op(T^{-1})$, and that the centered terms have mean zero and variance of order $\Op(T^{-1})$ by Lemma~\ref{lemma-mixing}. Moreover, by \eqref{eq:uniform-bound}, $\lambda_{i,k}=\Op(1)$ for $k=1,...,6$. Hence,
    \[
    \E_0[2T \elliOne \elliTwo]=2T\E_0[l_{i,1}l_{i,2}]+ 2T\lambda_{i,1}\lambda_{i,2}=\Op(1).
    \]
    Again by Lemma~\ref{lemma-mixing} and \eqref{eq:uniform-bound}, 
    \[
    \var_0(2T \elliOne \elliTwo)=\var_0(2T l_{i,1}\lambda_{i,2})+\Op(1)=\Op(T).
    \]
    Consequently, $\E_0D_{i,1}(\theta_0)=\Op(1)$ and $\var_0 D_{i,1}(\theta_0)=\Op(T)$, which is~\eqref{eq:rate_first_deriv}.
    Sample averages inherit the same rates from Assumption~\ref{ass-sampling}\textup{(i)}. Finally, we show \eqref{eq:rate_kth_deriv} by using a centering argument. Notice that for $k=2,3,4$, $D_{i,k}(\theta_0)$ is of the form
    \[
    c_1T\ell_{i,v_1}\ell_{i,v_2}+ c_2T\ell_{i,v_3}\ell_{i,v_4}+c_3T\ell_{i,1}\ell_{i,v_5}+\ell_{i,v_6}
    \]
    where $v_1,...,v_6\in\{2,...,6\}$ and $c_1, c_2, c_3$ are constants. By centering the likelihood,
    \[
    T\ell_{i,v_1}\ell_{i,v_2}=Tl_{i,v_1}l_{i,v_2}+T\lambda_{i,v_1}l_{i,v_2}+T\lambda_{i,v_2}l_{i,v_1}+T\lambda_{i,v_1}\lambda_{i,v_2}. 
    \]
    Since the centered terms have mean zero, 
    \[
    \E_0[T\ell_{i,v_1}\ell_{i,v_2}]=T\E_0[l_{i,v_1}l_{i,v_2}]+T\lambda_{i,v_1}\lambda_{i,v_2}.
    \]
    By Lemma~\ref{lemma-mixing}, $\E_0[l_{i,v_1}l_{i,v_2}]=\Op(T^{-1})$, so that $T\E_0[l_{i,v_1}l_{i,v_2}]=\Op(1)$. Moreover, by \eqref{eq:uniform-bound}, $\lambda_{i,v_j}=\Op(1)$, and therefore $T\lambda_{i,v_1}\lambda_{i,v_2}=\Op(T)$. Hence, $\E_0[T\ell_{i,v_1}\ell_{i,v_2}]=\Op(T)$. For the variance, note that $T\lambda_{i,v_1}\lambda_{i,v_2}$ is non-random with respect to $\var_0$, so that
    \[
    \var_0(T\ell_{i,v_1}\ell_{i,v_2})=\var_0\bigl(Tl_{i,v_1}l_{i,v_2}+T\lambda_{i,v_1}l_{i,v_2}+T\lambda_{i,v_2}l_{i,v_1}\bigr).
    \]
    Again by Lemma~\ref{lemma-mixing} and \eqref{eq:uniform-bound}, $\var_0(T\lambda_{i,v_1}l_{i,v_2})=\Op(T)$ and $\var_0(Tl_{i,v_1}l_{i,v_2})=\Op(1)$. The same argument applies to the remaining product terms in $D_{i,k}(\theta_0)$, whereas $\ell_{i,v_6}=\Op(1)$ contributes only at a lower order.     Consequently, $\E_0[D_{i,k}(\theta_0)]=\Op(T)$ and $\var_0(D_{i,k}(\theta_0))=\Op(T)$ for $k\in\{2,3,4\}$, which is~\eqref{eq:rate_kth_deriv}.
\end{proof}
    
The next Lemma relates the individual indicator of the profile likelihood to the corresponding individual indicator of the target likelihood.
\begin{lem}\label{lem:indicators}
Let Assumptions~\ref{ass-sampling},~\ref{ass-moments}, and~\ref{ass-remainder} hold.
For $k=0,...,4$, let $\Delta_{ik}(\theta)$ denote the $k$-th derivative of $\Delta_i(\theta)$, where 
    \begin{equation}
        \Delta_i(\theta)\defined D_i^\frakp(\theta)-D_i^\frakt(\theta),
        \qquad
        \delta(\theta)\defined\E_0[\Delta_i(\theta)].
    \end{equation}
    Then, $\Delta_{ik}(\theta)=\Op(T^{-1/2})$,   $\delta(\theta)=\Op(T^{-1})$ and $\var_0(\Delta_{ik}(\theta))=\Op(T^{-1})$.

\end{lem}
\begin{proof}
    To prove this Lemma, we derive an explicit expression of $\Delta_i$. For this, we need an expansion of the profile likelihood around the target likelihood.
    Following \cite{Schumann-second}, and using the mixed-partial notation from Assumption~\ref{ass-moments}\textup{(ii)}, write $l_{i,k,m}(\theta,\alpha_i)=\sqrt{T}\bigl(\ell_{i,k,m}(\theta,\alpha_i)-\lambda_{i,k,m}(\theta,\alpha_i)\bigr)$ for the centered and $\sqrt{T}$-scaled versions.
    Using this notation, the individual profile likelihood can be expanded around the individual target likelihood as
\newcommand{\alphabar}{\bar{\alpha}}
\newcommand{\alphastar}{\alpha^\frakt}

\begin{equation}
\elliT(\theta,\alphahatiT(\theta))-\elliT(\theta,\alphaiT^\frakt(\theta))= B_{i}^{(1)}(\theta)+B_{i}^{(3/2)}(\theta)+B_{i}^{(2)}(\theta),
\label{eq:profile-IB-expansion}
\end{equation}
where $B_i^{(2)}(\theta)=B_{i}^{(2)}(\theta,\alphabar_i(\theta))$ and $\alphabar(\theta)$ lies between $\alphahat_i(\theta)$ and $\alphai^\frakt(\theta)$. Suppressing arguments for brevity,
\begin{align*}
B_{i}^{(1)}(\theta)
&\defined -\frac{l_{i,0,1}^{2}}{2T\lambda_{i,0,2}}, \quad B_{i}^{(3/2)}(\theta)
\defined \frac{l_{i,0,1}^{2}\,l_{i,0,2}}{2T^{3/2}\lambda_{i,0,2}^{2}}
-\frac{l_{i,0,1}^{3}\,\lambda_{i,0,3}}{6T^{3/2}\lambda_{i,0,2}^{3}},\\
B_{i}^{(2)}(\theta,\alpha_i)
&\defined -\frac{l_{i,0,1}^{2}\,l_{i,0,2}^{2}}{2T^{2}\lambda_{i,0,2}^{3}}
-\frac{l_{i,0,1}^{3}\,l_{i,0,3}}{6T^{2}\lambda_{i,0,2}^{3}}
+\frac{l_{i,0,1}^{3}\,l_{i,0,2}\,\lambda_{i,0,3}}{2T^{2}\lambda_{i,0,2}^{4}}
-\frac{l_{i,0,1}^{4}\,\lambda_{i,0,3}^{2}}{8T^{2}\lambda_{i,0,2}^{5}}
+\frac{l_{i,0,1}^{4}\,\lambda_{i,0,4}}{24T^{2}\lambda_{i,0,2}^{4}}.
\end{align*}
Inspecting these expressions and using Assumption~\ref{ass-moments}\textup{(ii)}--\textup{(iii)}, $B_i^{(1)}(\theta)=\Op(T^{-1})$ and $B_i^{(3/2)}(\theta)=\Op(T^{-3/2})$ uniformly in a neighborhood of $\theta_0$ and $B_i^{(2)}(\theta,\alpha_i)=\Op(T^{-2})$ uniformly in a neighborhood of $(\theta_0,\alphaio)$. Moreover, since the stochastic order of these terms is determined by the number of centered likelihood derivatives, which is not affected by taking derivatives with respect to $\theta$, differentiation with respect to $\theta$ does not affect the stochastic order of these terms. Now, recall that 
\[
D_i^\frakp(\theta)=T\ell_{i,1}^\frakp(\theta)^2+\ell_{i,2}^\frakp(\theta).
\]
Differentiating \eqref{eq:profile-IB-expansion} with respect to $\theta$ (once and twice) yields $\ell_{i,1}^\frakp=\ell_{i,1}^\frakt+\Gamma_{i,1}$ and $\ell_{i,2}^\frakp=\ell_{i,2}^\frakt+\Gamma_{i,2}$, where
\[
\Gamma_{i,1}(\theta)\defined B_{i,1}^{(1)}(\theta)+B_{i,1}^{(3/2)}(\theta)+B_{i,1}^{(2)}(\theta),
\qquad
\Gamma_{i,2}(\theta)\defined B_{i,2}^{(1)}(\theta)+B_{i,2}^{(3/2)}(\theta)+B_{i,2}^{(2)}(\theta),
\]
so that $\ell_{i,1}^\frakp=\ell_{i,1}^\frakt+\Gamma_{i,1}$ and $\ell_{i,2}^\frakp=\ell_{i,2}^\frakt+\Gamma_{i,2}$. Substituting into the profile indicator and expanding the square,
\begin{align*}
D_i^\frakp(\theta)
&=T\bigl(\ell_{i,1}^\frakt(\theta)+\Gamma_{i,1}(\theta)\bigr)^2+\ell_{i,2}^\frakt(\theta)+\Gamma_{i,2}(\theta)\\
&=T\ell_{i,1}^\frakt(\theta)^2+\ell_{i,2}^\frakt(\theta)+2T\ell_{i,1}^\frakt(\theta)\Gamma_{i,1}(\theta)+T\Gamma_{i,1}(\theta)^2+\Gamma_{i,2}(\theta)\\
&=D_i^\frakt(\theta)+2T\ell_{i,1}^\frakt(\theta)\Gamma_{i,1}(\theta)+T\Gamma_{i,1}(\theta)^2+\Gamma_{i,2}(\theta).
\end{align*}
Expanding the products in $\Gamma_{i,1}$ and $\Gamma_{i,2}$ and collecting the terms according to their stochastic order yields, after some algebra, 
\begin{align}\label{eq:Delta-expansion}
    \Delta_i(\theta)
    &=D_i^\frakp(\theta)-D_i^\frakt(\theta)\notag\\
    &=2T\ell_{i,1}^\frakt(\theta)B_{i,1}^{(1)}(\theta)
    +T\bigl(B_{i,1}^{(1)}(\theta)\bigr)^2
    +B_{i,2}^{(1)}(\theta)
    +2T\ell_{i,1}^\frakt(\theta)B_{i,1}^{(3/2)}(\theta)+R_i^{(3/2)}(\theta)+R_i^{(2)}(\theta),
    \end{align}
    where
    \begin{align*}
        R_i^{(3/2)}(\theta)&=2T\ell_{i,1}^\frakt(\theta)B_{i,1}^{(2)}(\theta)
        +2T B_{i,1}^{(1)}(\theta)B_{i,1}^{(3/2)}(\theta)
        +B_{i,2}^{(3/2)}(\theta) \mbox{ and }\\
        R_i^{(2)}(\theta)&=2T B_{i,1}^{(1)}(\theta)B_{i,1}^{(2)}(\theta)
        +T\bigl(B_{i,1}^{(3/2)}(\theta)\bigr)^2
        +2T B_{i,1}^{(3/2)}(\theta)B_{i,1}^{(2)}(\theta)
        +T\bigl(B_{i,1}^{(2)}(\theta)\bigr)^2
        +B_{i,2}^{(2)}(\theta).
    \end{align*}
By the preceding discussion of the stochastic order, $\E_0[R_i^{(2)}(\theta)]=\Op(T^{-2})$  and $\var_0(R_i^{(2)}(\theta))=\Op(T^{-4})$ in a neighborhood of $\theta_0$. Similarly, $R_i^{(3/2)}(\theta)=\Op(T^{-3/2})$. However, since this remainder term consists of terms that involve products of 5 centered likelihood terms, $\E_0 R_i^{(3/2)}(\theta)=\Op(T^{-2})$ in a neighborhood of $\theta_0$ by Lemma~\ref{lemma-mixing}. Further notice that $2T\ell_{i,1}^\frakt(\theta)B_{i,1}^{(1)}(\theta)=\Op(T^{-1/2})$, but since this term involves 3 centered likelihood terms, $2T\E_0\ell_{i,1}^\frakt(\theta)B_{i,1}^{(1)}(\theta)=\Op(T^{-1})$ by Lemma~\ref{lemma-mixing}. Therefore, $\Delta_i(\theta)=\Op(T^{-1/2})$ and $\delta(\theta)=\E_0[\Delta_i(\theta)]+\Op(T^{-2})$, where
\[
\delta(\theta)=\E_0\Bigl[2T\ell_{i,1}^\frakt(\theta)B_{i,1}^{(1)}(\theta)
+T\bigl(B_{i,1}^{(1)}(\theta)\bigr)^2
+B_{i,2}^{(1)}(\theta)
+2T\ell_{i,1}^\frakt(\theta)B_{i,1}^{(3/2)}(\theta)\Bigr].
\]
Taking the variance on both sides of \eqref{eq:Delta-expansion} and using Assumption~\ref{ass-moments}\textup{(ii)}--\textup{(iii)}, 
\[
    \var_0(\Delta_i(\theta))\leq 4T^2\E_0[(\ell_{i,1}^\frakt(\theta)B_{i,1}^{(1)}(\theta))^2]+\Op(T^{-2})=\Op(T^{-1}).
\]
Since differentiation does not change the number of centered likelihood terms involved in each term, this finishes the proof.

\end{proof}

Using Lemmas \ref{lemma-rates} and \ref{lem:indicators}, we now turn to the proof of Lemma \ref{lemma-scaled-indicator}.
\medskip
\begin{proof}[\textupandbold{Proof of Lemma \ref{lemma-scaled-indicator}}]
The argument uses Assumptions~\ref{ass-sampling}, \ref{ass-moments}, \ref{ass-remainder}, \ref{ass-estimators}, and~\ref{ass-variance}. The profile CLT under parabolic asymptotics requires Assumption~\ref{ass-n-T-rate}\textup{(ii)}. A second-order Taylor expansion of $D^\frakt$ around $\theta_0$ yields
\[
\sqrt{n}D^\frakt(\thetahat^\frakt)=\sqrt{n}D^\frakt(\theta_0)+\sqrt{n}D_1^\frakt(\theta_0)(\thetahat^\frakt-\theta_0)+\mathrm{R}_n^\frakt,
\]
where $\mathrm{R}_n^\frakt=\frac12\sqrt{n}\,D_2^\frakt(\thetabar)(\thetahat^\frakt-\theta_0)^2$ for an intermediate value $\thetabar$ between $\thetahat^\frakt$ and $\theta_0$. By \eqref{eq:rate_first_deriv}, $\E_0[D_{i,1}^\frakt(\theta_0)]=\Op(1)$ and $\var_0(D_{i,1}^\frakt(\theta_0))=\Op(T)$. Hence, under Assumption~\ref{ass-sampling},
\begin{equation}\label{eq:D1-mean-fluct}
D_1^\frakt(\theta_0)=\E_0[D_1^\frakt(\theta_0)]+\Op\bigl(T^{1/2}n^{-1/2}\bigr).
\end{equation}
Consequently, 
\[
\sqrt{n}D_1^\frakt(\theta_0)(\thetahat^\frakt-\theta_0)=\E_0[D_{i,1}^\frakt(\theta_0)]\sqrt{n}(\thetahat^\frakt-\theta_0) + \Op(\sqrt{T})(\thetahat^\frakt-\theta_0),
\]
so that by Assumption~\ref{ass-estimators}
\[
\sqrt{n}D_1^\frakt(\theta_0)(\thetahat^\frakt-\theta_0)=\Op(n^{-1/2})+\Op(T^{-1/2}).
\]
For the second remainder term, recall that $\sup_\theta|D_2^\frakt(\theta)|=\Op(T)$ by Lemma~\ref{lemma-rates}. Hence, Assumption~\ref{ass-estimators} implies that
\[
R_n=\Op(n^{-1/2}).
\]
Overall, we thus get
\begin{equation}\label{eq:rates-indicator-thetahat}
\sqrt{n}D^\frakt(\thetahat^\frakt)=\sqrt{n}D^\frakt(\theta_0)+\Op(n^{-1/2})+\Op(T^{-1/2}),
\end{equation}
as stated in the Lemma. Next, recall from Lemma \ref{lemma-variance} that $\var_0(D_i^\frakt)=\Op(1)$. Moreover, since the target likelihood is information unbiased, $\E_0[D_i^\frakt(\theta_0)]=0$. To apply a CLT for triangular arrays, we use the Lyapunov criterion by showing that $\E_0[D_i^\frakt(\theta_0)^4]=\Op(1)$. To see this, omit the argument $\theta_0$ and the target superscript $^\frakt$, and write $D_i=T\ell_{i,1}^2+\ell_{i,2}$. For the target likelihood, $\lambda_{i,1}=0$, so $\ell_{i,1}=l_{i,1}$ and $\ell_{i,2}=l_{i,2}+\lambda_{i,2}$. Hence,
$D_i=T l_{i,1}^2+l_{i,2}+\lambda_{i,2}$. 
Expanding the fourth power,
\begin{align*}
D_i^4
&=(T l_{i,1}^2+l_{i,2}+\lambda_{i,2})^4\\
&=T^4 l_{i,1}^8
+4T^3 l_{i,1}^6(l_{i,2}+\lambda_{i,2})
+6T^2 l_{i,1}^4(l_{i,2}+\lambda_{i,2})^2
+4T l_{i,1}^2(l_{i,2}+\lambda_{i,2})^3
+(l_{i,2}+\lambda_{i,2})^4.
\end{align*}
Taking expectations while noting that products of $k$ centered individual likelihood derivatives are of order $O(T^{-\lceil k/2\rceil})$ by Lemma~\ref{lemma-mixing} together with $\lambda_{i,2}=\Op(1)$ from Assumption~\ref{ass-moments}, the desired result follows. Next, we continue by showing asymptotic normality for the target indicator. By Assumption~\ref{ass-variance}, the limiting variance $V=\lim_{n,T\to\infty}\var_0(\sqrt{n}D^\frakt(\theta_0))$ exists and is nonsingular.
Writing $\sqrt{n}D^\frakt(\theta_0)=\sum_{i=1}^n \xi_{n,i}$ with
$\xi_{n,i}\defined n^{-1/2}D_i^\frakt(\theta_0)$,
 the triangular array $\{\xi_{n,i}\}_{i=1}^n$ is i.i.d.\ with mean zero and $\sigma_n^2\defined\var_0(\sum_{i=1}^n\xi_{n,i})=\var_0(D_1^\frakt(\theta_0))\to V\in(0,\infty)$ for each $(n,T)$. 
By the Lyapunov criterion with $\delta=2$, 
\[
\frac{1}{\sigma_n^{4}}\sum_{i=1}^n\E_0|\xi_{n,i}|^{4}
=
\frac{1}{\sigma_n^{4}}\,n\cdot n^{-2}\E_0\bigl[D_1^\frakt(\theta_0)^4\bigr]
=
\frac{\E_0\bigl[D_1^\frakt(\theta_0)^4\bigr]}{n\,\sigma_n^{4}}
=\Op(n^{-1})\to 0,
\]
using $\E_0[D_1^\frakt(\theta_0)^4]=\Op(1)$ and $\sigma_n^2\to V>0$. As the Lyapunov condition implies that the Lindeberg condition holds, $\sqrt{n}D^\frakt(\theta_0)\dto\No(0,V)$. 
Combined this result with \eqref{eq:rates-indicator-thetahat} and Slutsky's lemma, $\sqrt{n}D^\frakt(\thetahat^\frakt)\dto\No(0,V)$, which finishes part~(i).

\newcommand{\calD}{\mathcal{D}}

We now continue with part~(ii) for the profile likelihood. To make the role of the score bias of the profile likelihood explicit, define $\calD(\theta)=D_i(\theta)-\E_0[D(\theta)]$ and denote additional derivatives by the corresponding numbers in the index. By a second-order Taylor expansion,
\[
\sqrt{n}D^\frakp(\thetahat^\frakp)=\sqrt{n}D^\frakp(\theta_0)+\sqrt{n}\E_0[D_1(\theta_0)](\thetahat-\theta_0)+\sqrt{n}\calD_1(\theta_0)(\thetahat-\theta_0)+ \frac{1}{2}\sqrt{n}\,D_2^\frakp(\thetabar)(\thetahat^\frakp-\theta_0)^2
\]
 for an intermediate value $\thetabar$ between $\thetahat^\frakp$ and $\theta_0$.
A fourth-order Taylor expansion of $\sqrt{n}D^\frakp(\thetahat^\frakp)$ around $\theta_0$ yields
\begin{multline}
\sqrt{n}D^\frakp(\thetahat^\frakp)
=\sqrt{n}D^\frakp(\theta_0)
+\sqrt{n}D_1^\frakp(\theta_0)(\thetahat^\frakp-\theta_0)
+\frac{1}{2}\sqrt{n}D_2^\frakp(\theta_0)(\thetahat^\frakp-\theta_0)^2\\
+\frac{1}{6}\sqrt{n}D_3^\frakp(\theta_0)(\thetahat^\frakp-\theta_0)^3
+\frac{1}{24}\sqrt{n}D_4^\frakp(\thetabar)(\thetahat^\frakp-\theta_0)^4,
\end{multline}
where $\thetabar$ lies between $\thetahat^\frakp$ and $\theta_0$.
By Lemma~\ref{lemma-rates}, $D_3^\frakp(\theta_0)=\Op(T)$ and $D_4^\frakp(\thetabar)=\Op(T)$. Combined with Assumption~\ref{ass-estimators},
\[
\frac{1}{6}\sqrt{n}D_3^\frakp(\theta_0)(\thetahat^\frakp-\theta_0)^3
=\Op(n^{-1}T^{-1/2})+\Op(n^{-1/2}T^{-1})+\Op(\sqrt{n}\,T^{-2})
\]
and
\[
\frac{1}{24}\sqrt{n}D_4^\frakp(\thetabar)(\thetahat^\frakp-\theta_0)^4
=\Op(n^{-3/2}T^{-1})+\Op(n^{-1/2}T^{-2})+\Op(\sqrt{n}\,T^{-3}),
\]
so that all terms vanish if $n/T^2\to\rho\in[0,\infty)$. 
Next, we consider the first remainder term $\sqrt{n}D_1^\frakp(\theta_0)(\thetahat^\frakp-\theta_0)$, which is shown to contribute to asymptotic bias if $\rho>0$. To see this, we use equation~(F.1) of \citet{schumann-info}, which shows that
\begin{equation}\label{eq:thetahat-linearization}
\thetahat-\theta_0=-\lambda_2^{-1}\ell_1+R,
\end{equation}
where the remainder satisfies
\[
R=\begin{cases}
\Op\bigl((nT)^{-1}\bigr), & \ell=\ell^\frakt,\\
\Op\bigl((nT)^{-1}\bigr)+\Op\bigl((nT^3)^{-1/2}\bigr), & \ell=\ell^\frakp
\end{cases}
\]
and where we omit the argument when evaluating at $\theta_0$. Hence,
\[
\sqrt{n}D_1^\frakp(\thetahat^\frakp-\theta_0)=-\sqrt{n}D_1^\frakp\lambda_2^{-1}\ell_1^\frakp+\sqrt{n}D_1^\frakp R.
\]
By Lemma~\ref{lemma-rates}, $\E_0[D_1^\frakp]=\Op(1)$ and $\var_0(D_1^\frakp)=\Op(T/n)$, so $D_1^\frakp=\Op(1)+\Op(\sqrt{T/n})$. Combined with the profile rate for $R$ in \eqref{eq:thetahat-linearization},
\[
\sqrt{n}D_1^\frakp R=\Op(n^{-1/2}T^{-1})+\Op(T^{-3/2})+\Op(n^{-1}T^{-1/2})=\op(1).
\]
For the leading term, $\lambda_2^{-1}=\Op(1)$ by \eqref{eq:uniform-bound} is non-random with respect to $\var_0$. For the profile likelihood, $\ell_1^\frakp=l_1^\frakp+\lambda_1^\frakp$ with $l_1^\frakp=\Op((nT)^{-1/2})$ and $\lambda_1^\frakp=\Op(T^{-1})$ is the score bias of the profile likelihood. Writing $D_1^\frakp=\E_0[D_1^\frakp]+\calD_1^\frakp$ and centering the score,
\begin{align*}
-\sqrt{n}D_1^\frakp\lambda_2^{-1}\ell_1^\frakp
&=-\sqrt{n}\lambda_2^{-1}\bigl(\E_0[D_1^\frakp]+\calD_1^\frakp\bigr)\bigl(\lambda_1^\frakp+l_1^\frakp\bigr)\\
&=-\sqrt{n}\lambda_2^{-1}\E_0[D_1^\frakp]\,\lambda_1^\frakp
-\sqrt{n}\lambda_2^{-1}\E_0[D_1^\frakp]\,l_1^\frakp
-\sqrt{n}\lambda_2^{-1}\calD_1^\frakp\,\lambda_1^\frakp
-\sqrt{n}\lambda_2^{-1}\calD_1^\frakp\,l_1^\frakp.
\end{align*}
By Assumption~\ref{ass-sampling}, $\E_0[\calD_1^\frakp]=\E_0[l_1^\frakp]=0$. The leading contribution to the expectation is therefore $\sqrt{n}/TB_{T,1}^{\mathrm{score}}$, where
\[
B_{T,1}^{\mathrm{score}}\defined -T\lambda_2^{-1}\E_0[D_1^\frakp]\,\lambda_1^\frakp.
\]
By Lemma~\ref{lemma-rates}, $\E_0[D_1^\frakp]=\Op(1)$, and since $\lambda_1^\frakp=\Op(T^{-1})$ and $\lambda_2^{-1}=\Op(1)$,
\[
 B_{1}^{\mathrm{score}}\defined \plim_{n,T\to\infty}\frac{\sqrt{n}}{T}B_{T,1}^{\mathrm{score}}
 = \sqrt{\rho}\plim_{T\to\infty}B_{T,1}^{\mathrm{score}}.
\]
The second limit exists under the maintained assumptions. To see this, notice that by the first-order score-bias expansion for the profile likelihood \citep{schumann-info,hahn-newey-nonlinear-panels}, $\lambda_{i,1}^\frakp=T^{-1}\beta_i+o_{\mathrm{a.s.}}(T^{-1})$ with $\beta_i$ $\calG_i$-measurable. Averaging across $i$ and using Assumption~\ref{ass-sampling}\textup{(i)} yields $T\lambda_1^\frakp\to_p\beta\defined\E[\beta_i]$. Analogously, $\E_0[D_{i,1}^\frakp]\overset{\wpone}{\to}\mu_i$ by Lemma~\ref{lemma-rates} and time-stationarity in Assumption~\ref{ass-sampling}, so $\E_0[D_1^\frakp]\to_p\mu\defined\E[\mu_i]$, and $\lambda_2\to_p\bar\lambda_2\neq 0$. Continuous mapping then gives $B_{T,1}^{\mathrm{score}}\to_p-\bar\lambda_2^{-1}\mu\beta$ and $B_1^{\mathrm{score}}=-\sqrt{\rho}\,\bar\lambda_2^{-1}\mu\beta$.
For the variance, $\E_0[D_1^\frakp]$, $\lambda_1^\frakp$, and $\lambda_2^{-1}$ are non-random with respect to $\var_0$, $\var_0(\calD_1^\frakp)=\Op(T/n)$, and $\var_0(l_1^\frakp)=\Op((nT)^{-1})$. The dominating term in the variance is $-\sqrt{n}\lambda_2^{-1}\E_0[D_1^\frakp]\,l_1^\frakp$, with $\var_0\bigl(-\sqrt{n}\lambda_2^{-1}\E_0[D_1^\frakp]\,l_1^\frakp\bigr)=\Op(T^{-1})$.
The remaining terms have variances of lower order, so $\var_0(-\sqrt{n}D_1^\frakp\lambda_2^{-1}\ell_1^\frakp)=\Op(T^{-1})$.

Next, we consider the second remainder term $\frac12\sqrt{n}D_2^\frakp(\theta_0)(\thetahat^\frakp-\theta_0)^2$. Using \eqref{eq:thetahat-linearization},
\[
(\thetahat^\frakp-\theta_0)^2=\lambda_2^{-2}(\ell_1^\frakp)^2-2\lambda_2^{-1}\ell_1^\frakp R+R^2,
\]
and therefore
\[
\frac12\sqrt{n}D_2^\frakp(\thetahat^\frakp-\theta_0)^2=\frac12\lambda_2^{-2}\sqrt{n}D_2^\frakp(\ell_1^\frakp)^2+\widetilde{R}_n,
\]
where
\[
\widetilde{R}_n\defined\frac12\sqrt{n}D_2^\frakp\bigl(-2\lambda_2^{-1}\ell_1^\frakp R+R^2\bigr).
\]
By Lemma~\ref{lemma-rates}, $D_2^\frakp=\Op(T)$. Combined with the rates for $\ell_1^\frakp$ and $R$ already established,
\[
\widetilde{R}_n=\Op(n^{-1/2}T^{-1})+\Op(T^{-3/2})+\Op(n^{-1}T^{-1/2})=\op(1).
\]
Writing $D_2^\frakp=\E_0[D_2^\frakp]+\calD_2^\frakp$ and centering the scores as above, the same independence argument yields the leading mean contribution $\sqrt{n}/TB_{T,2}^{\mathrm{score}}$, where
\[
B_{T,2}^{\mathrm{score}}\defined \tfrac12 T\lambda_2^{-2}\E_0[D_2^\frakp]\,(\lambda_1^\frakp)^2.
\]
By Lemma~\ref{lemma-rates}, $\E_0[D_2^\frakp]=\Op(T)$, and since $\lambda_1^\frakp=\Op(T^{-1})$ and $\lambda_2^{-2}=\Op(1)$,
\[
B_2^{\mathrm{score}}\defined \plim_{n,T\to\infty}\frac{\sqrt{n}}{T}B_{T,2}^{\mathrm{score}}
= \sqrt{\rho}\plim_{T\to\infty}B_{T,2}^{\mathrm{score}}.
\]
Again, the second limit exists following analogous arguments. Hence $B_{T,2}^{\mathrm{score}}\to_p\tfrac12\bar\lambda_2^{-2}\mu_2\beta^2$ and $B_2^{\mathrm{score}}=\tfrac12\sqrt{\rho}\,\bar\lambda_2^{-2}\mu_2\beta^2$.
For the variance, the leading term equals $\frac12\lambda_2^{-2}\sqrt{n}D_2^\frakp(\ell_1^\frakp)^2$. Expanding as above,
\begin{align*}
\sqrt{n}D_2^\frakp(\ell_1^\frakp)^2
&=\sqrt{n}\,\E_0[D_2^\frakp]\,(\lambda_1^\frakp)^2
+2\sqrt{n}\,\E_0[D_2^\frakp]\,\lambda_1^\frakp\,l_1^\frakp
+\sqrt{n}\,\E_0[D_2^\frakp]\,(l_1^\frakp)^2\\
&\quad+\sqrt{n}\,\calD_2^\frakp\,(\lambda_1^\frakp)^2
+2\sqrt{n}\,\calD_2^\frakp\,\lambda_1^\frakp\,l_1^\frakp
+\sqrt{n}\,\calD_2^\frakp\,(l_1^\frakp)^2.
\end{align*}
The leading fluctuation is $2\sqrt{n}\,\E_0[D_2^\frakp]\,\lambda_1^\frakp\,l_1^\frakp$. Using $\E_0[D_2^\frakp]=\Op(T)$, $\var_0(\calD_2^\frakp)=\Op(T/n)$, and $\var_0(l_1^\frakp)=\Op((nT)^{-1})$ as before,
\[
\var_0\bigl(2\sqrt{n}\,\E_0[D_2^\frakp]\,\lambda_1^\frakp\,l_1^\frakp\bigr)=\Op(T^{-1}),
\]
and the remaining terms are of lower order. Hence,
\[
\var_0\Bigl(\tfrac12\lambda_2^{-2}\sqrt{n}D_2^\frakp(\ell_1^\frakp)^2\Bigr)=\Op(T^{-1}).
\]
The centered contributions in both score expansions therefore satisfy $\Op(T^{-1/2})=\op(1)$. Together with $\sqrt{n}D_1^\frakp R=\op(1)$ and $\widetilde{R}_n=\op(1)$,
\begin{equation}\label{eq:profile-indicator-approx}
\sqrt{n}D^\frakp(\thetahat^\frakp)=\sqrt{n}D^\frakp(\theta_0)+\frac{\sqrt{n}}{T}B_{T}^{\mathrm{score}}+\op(1),
\end{equation}
where $B_{T}^{\mathrm{score}}\defined B_{T,1}^{\mathrm{score}}+B_{T,2}^{\mathrm{score}}$ collects the terms due to the score bias that are asymptotically non-negligible if $\rho>0$. Next, we focus on the leading term. Recall that by Lemma~\ref{lem:indicators}, $D_i^\frakp(\theta_0)=D_i^\frakt(\theta_0)+\Delta_i(\theta_0)$ so that  $\E_0[\Delta(\theta_0)]=\delta_T(\theta_0)$ is the information bias of the profile likelihood. Thus,
\begin{equation}\label{eq:profile-at-theta0}
\sqrt{n}D^\frakp(\theta_0)
=\sqrt{n}D^\frakt(\theta_0)
+\frac{1}{\sqrt{n}}\sumi\bigl(\Delta_i(\theta_0)-\delta_T(\theta_0)\bigr)
+\sqrt{n}\,\delta_T(\theta_0).
\end{equation}
Since $\var_0(\Delta_i(\theta_0))=\Op(T^{-1})$ by Lemma~\ref{lem:indicators}, 
\[
\frac{1}{\sqrt{n}}\sumi\bigl(\Delta_i(\theta_0)-\E_0[\Delta_i(\theta_0)]\bigr)=\Op(T^{-1/2}).
\]
The remaining term is the information-bias contribution $\sqrt{n}\,\delta_T(\theta_0)$. Writing
\[
\sqrt{n}\,\delta_T(\theta_0)=\frac{\sqrt{n}}{T}B_T^{\mathrm{info}},
\qquad
B_T^{\mathrm{info}}\defined T\delta_T(\theta_0),
\]
Lemma~\ref{lem:indicators} yields $\delta_T(\theta_0)=\Op(T^{-1})$, and therefore
\[
B^{\mathrm{info}}\defined \plim_{n,T\to\infty}\frac{\sqrt{n}}{T}B_T^{\mathrm{info}}
= \sqrt{\rho}\plim_{T\to\infty}B_T^{\mathrm{info}},
\]
where the second limit exists under the maintained assumptions. To see this, notice that by Lemma~\ref{lem:indicators}, $B_T^{\mathrm{info}}$ equals $T$ times the expectation of the leading terms in \eqref{eq:Delta-expansion} up to $o(1)$. Each of these terms is a smooth function of within-individual time averages of strictly stationary, geometrically mixing period contributions, so $T\E_0[\Delta_i(\theta_0)]\overset{\wpone}{\to}\gamma_i$ with $\gamma_i$ $\calG_i$-measurable. Averaging across $i$ and using Assumptions~\ref{ass-sampling} and~\ref{ass-moments} then yields $B_T^{\mathrm{info}}\to_p\gamma\defined\E[\gamma_i]$ and $B^{\mathrm{info}}=\sqrt{\rho}\,\gamma$. The centered information-bias average in \eqref{eq:profile-at-theta0} is $\Op(T^{-1/2})=\op(1)$. Substituting \eqref{eq:profile-at-theta0} into \eqref{eq:profile-indicator-approx} therefore gives
\begin{equation}\label{eq:profile-indicator-rate}
\sqrt{n}D^\frakp(\thetahat^\frakp)
=\sqrt{n}D^\frakt(\theta_0)
+\frac{\sqrt{n}}{T}\bigl(B_{T}^{\mathrm{score}}+B_T^{\mathrm{info}}\bigr)
+\op(1).
\end{equation}
Since $\sqrt{n}D^\frakt(\theta_0)\dto \No(0,V)$ as shown in the proof for part (i), Slutsky's Lemma implies that
\[
    \sqrt{n}D^\frakp(\thetahat^\frakp)\dto \No\Bigl( \sqrt{\rho}\kappa,V\Bigr),
    \]
    with $\kappa\defined \plim_{T\to\infty} (B_{T}^{\mathrm{score}}+B_T^{\mathrm{info}})$, which finishes the proof.
\end{proof}

Using Lemma \ref{lemma-rates}, we now continue with the proof of Lemma \ref{lemma-variance-estimator}.
\medskip
\begin{proof}[\textupandbold{Proof of Lemma \ref{lemma-variance-estimator}}]
    Assume $\dim(\theta_0)=1$. Let $\ell$ denote a generic pseudo-likelihood (target or profile), with corresponding indicator $D$ and maximizer $\thetahat$, and recall $\Vhat(\theta)\defined n^{-1}\sumi D_i(\theta)^2$. Expanding $\Vhat(\thetahat)$ around the true parameter,
    \begin{multline}\label{eq:fourth-order-taylor}
    \Vhat(\thetahat)
    =\Vhat(\theta_0)
    +\Vhat_1(\theta_0)(\thetahat-\theta_0)
    +\frac{1}{2}\Vhat_2(\theta_0)(\thetahat-\theta_0)^2\\
    +\frac{1}{6}\Vhat_3(\theta_0)(\thetahat-\theta_0)^3
    +\frac{1}{24}\Vhat_4(\thetabar)(\thetahat-\theta_0)^4,
    \end{multline}
    where $\thetabar$ lies between $\thetahat$ and $\theta_0$, and $\Vhat_k$ denotes the $k$-th derivative of $\Vhat(\theta)$ for $k=1,\ldots,4$. Since $D_i(\theta_0)=\Op(1)$ and $D_{i,1}(\theta_0)=\Op(\sqrt{T})$ by Lemma~\ref{lemma-rates},
    \[
    \Vhat_1(\theta_0)=\frac{1}{n}\sumi 2D_{i,1}(\theta_0)D_i(\theta_0)=\Op(\sqrt{T}).
    \]
    Combined with Assumption~\ref{ass-estimators},
    \[
    \Vhat_1(\theta_0)(\thetahat-\theta_0)=
    \begin{cases}
    \Op(n^{-1/2}), &\quad \ell=\ell^\frakt,\\
    \Op(n^{-1/2})+\Op(T^{-1/2}), &\quad \ell=\ell^\frakp.
    \end{cases}
    \]
    Next, evaluating the second derivative at $\theta_0$ (so that $D_{i,1}(\theta_0)=\Op(\sqrt{T})$ rather than the uniform $\Op(T)$ bound),
    \[
    \Vhat_2(\theta_0)=\frac{1}{n}\sumi\bigl[2D_{i,1}(\theta_0)^2+2D_{i,2}(\theta_0)D_i(\theta_0)\bigr]=\Op(T),
    \]
    and therefore
    \[
    \Vhat_2(\theta_0)(\thetahat-\theta_0)^2=
    \begin{cases}
    \Op(n^{-1}), &\quad \ell=\ell^\frakt,\\
    \Op(n^{-1})+\Op(T^{-1})+\Op\bigl((nT)^{-1/2}\bigr), &\quad \ell=\ell^\frakp
    \end{cases}
    \]
    is $\op(1)$. Differentiating once more,
    \[
    \Vhat_3(\theta_0)=\frac{1}{n}\sumi\bigl[6D_{i,1}(\theta_0)D_{i,2}(\theta_0)+2D_{i,3}(\theta_0)D_i(\theta_0)\bigr]=\Op(T^{3/2}),
    \]
    since $D_{i,2}(\theta_0)=\Op(T)$ and $D_{i,3}(\theta_0)=\Op(T)$ by Lemma~\ref{lemma-rates}. Hence
    \[
    \Vhat_3(\theta_0)(\thetahat-\theta_0)^3=
    \begin{cases}
    \Op(n^{-3/2}), &\quad \ell=\ell^\frakt,\\
    \Op(n^{-3/2})+\Op(n^{-1}T^{-1/2})+\Op(n^{-1/2}T^{-1})+\Op(T^{-3/2}), &\quad \ell=\ell^\frakp,
    \end{cases}
    \]
    so the third remainder in \eqref{eq:fourth-order-taylor} is $\op(1)$. Likewise,
    \[
    \Vhat_4(\theta)=\frac{1}{n}\sumi\bigl[6D_{i,2}(\theta)^2+8D_{i,1}(\theta)D_{i,3}(\theta)+2D_{i,4}(\theta)D_i(\theta)\bigr].
    \]
    By Lemma~\ref{lemma-rates}, $\sup_\theta|D_{i,k}(\theta)|=\Op(T)$ for $k=0,\ldots,4$, so $\Vhat_4(\thetabar)=\Op(T^2)$. Combined with Assumption~\ref{ass-estimators},
    \[
    \Vhat_4(\thetabar)(\thetahat-\theta_0)^4=
    \begin{cases}
    \Op(n^{-2}), &\quad \ell=\ell^\frakt,\\
    \Op(n^{-2})+\Op(n^{-3/2}T^{-1/2})+\Op(n^{-1}T^{-1})+\Op(n^{-1/2}T^{-3/2})+\Op(T^{-2}), &\quad \ell=\ell^\frakp,
    \end{cases}
    \]
    hence the fourth remainder is also $\op(1)$. Altogether, $\Vhat(\thetahat)=\Vhat(\theta_0)+\op(1)$.

It remains to show $\Vhat(\theta_0)=V+\op(1)$. For the target likelihood, $\E_0[D_i^\frakt(\theta_0)]=0$, so $\E_0[\Vhat^\frakt(\theta_0)]=V_T$ with $V_T\defined\var_0(D_i^\frakt(\theta_0))$. Assumption~\ref{ass-variance} yields $V_T\to V$. Independence across individuals and $\E_0[D_i^\frakt(\theta_0)^4]=\Op(1)$ (as in the Lyapunov argument in the proof of Lemma~\ref{lemma-scaled-indicator}) give $\Vhat^\frakt(\theta_0)=V_T+\Op(n^{-1/2})=V+\op(1)$, and therefore $\Vhat^\frakt(\thetahat^\frakt)=V+\op(1)$.

For the profile likelihood, Lemma~\ref{lem:indicators} yields $D_i^\frakp(\theta_0)=D_i^\frakt(\theta_0)+\Delta_i(\theta_0)$ with $\E_0[\Delta_i(\theta_0)^2]=\Op(T^{-1})$. Expanding the square,
\[
\Vhat^\frakp(\theta_0)=\Vhat^\frakt(\theta_0)+\frac{1}{n}\sumi\bigl(2D_i^\frakt(\theta_0)\Delta_i(\theta_0)+\Delta_i(\theta_0)^2\bigr).
\]
The expectation of the summand is $\op(1)$ by the Cauchy--Schwarz inequality, so the average is $\op(1)$ by Markov's inequality under Assumption~\ref{ass-sampling}\textup{(i)}. Combined with $\Vhat^\frakt(\theta_0)=V+\op(1)$ and $\Vhat^\frakp(\thetahat^\frakp)=\Vhat^\frakp(\theta_0)+\op(1)$, we obtain $\Vhat^\frakp(\thetahat^\frakp)=V+\op(1)$.
\end{proof}

\begin{proof}[\textupandbold{Proof of Theorem \ref{main-theorem}}]
Recall from \eqref{eq:def-Whatp} that
\[
\Whattarget
=\bigl(\sqrt{n}D^\frakt(\thetahat^\frakt)\bigr)'(\Vhatstar)^{-1}\bigl(\sqrt{n}D^\frakt(\thetahat^\frakt)\bigr),
\qquad
\Whatprofile
=\bigl(\sqrt{n}D^\frakp(\thetahat^\frakp)\bigr)'(\Vhatp)^{-1}\bigl(\sqrt{n}D^\frakp(\thetahat^\frakp)\bigr),
\]
where $\Vhatstar\defined\Vhat^\frakt(\thetahat^\frakt)$ and $\Vhatp\defined\Vhat^\frakp(\thetahat^\frakp)$ with $\Vhat(\theta)$ as in Lemma~\ref{lemma-variance-estimator}. We begin the proof with part (i) for the target likelihood. By Lemma~\ref{lemma-scaled-indicator}\textup{(i)}, 
$\sqrt{n}D^\frakt(\thetahat^\frakt)\dto\No(0,V)$.  Moreover, by Lemma~\ref{lemma-variance-estimator}, $\Vhatstar=V+\op(1)$, and therefore $\Vhatstar\to_p V$. Since $V$ is nonsingular by Assumption~\ref{ass-variance}, the continuous mapping theorem yields $(\Vhatstar)^{-1}\to_p V^{-1}$ (in particular, $\Vhatstar$ is invertible with probability approaching one). Writing $Z_n^\frakt\defined\sqrt{n}D^\frakt(\thetahat^\frakt)$ and $A_n^\frakt\defined(\Vhatstar)^{-1}$, we thus have $Z_n^\frakt\dto Z\sim\No(0,V)$ and $A_n^\frakt\to_p V^{-1}$. Continuous mapping for the quadratic form therefore gives
\[
\Whattarget=Z_n^{\frakt\prime}A_n^\frakt Z_n^\frakt\dto Z'V^{-1}Z\sim\chi^2_q,
\]
which finishes part~(i). Next, we continue with part (ii) for the profile likelihood. Assume that $n/T^2\to\rho\in[0,\infty)$, i.e.\ Assumption~\ref{ass-n-T-rate}\textup{(ii)} holds. By Lemma~\ref{lemma-scaled-indicator}\textup{(ii)},
\[
\sqrt{n}D^\frakp(\thetahat^\frakp)\dto\No\bigl(\sqrt{\rho}\,\kappa(\theta_0),V\bigr).
\]
By Lemma~\ref{lemma-variance-estimator}, $\Vhatp=V+\op(1)$, so $\Vhatp\to_p V$ as $n,T\to\infty$, and the same continuous-mapping argument as above yields $(\Vhatp)^{-1}\to_p V^{-1}$. Writing $Z_n^\frakp\defined\sqrt{n}D^\frakp(\thetahat^\frakp)$ and $A_n^\frakp\defined(\Vhatp)^{-1}$, continuous mapping for the quadratic form gives
\[
\Whatprofile=Z_n^{\frakp\prime}A_n^{\frakp}Z_n^{\frakp}
\dto (Z+\mu)'V^{-1}(Z+\mu),
\]
where $Z\sim\No(0,V)$ and $\mu\defined\sqrt{\rho}\,\kappa(\theta_0)$. The right-hand side is $\chi^2_q(\lambda_\rho)$ with noncentrality
\[
\lambda_\rho=\mu'V^{-1}\mu=\rho\,\kappa(\theta_0)'V^{-1}\kappa(\theta_0),
\]
as claimed. When $\rho=0$, one has $\mu=0$ and $\lambda_\rho=0$, so the profile limit coincides with the target limit. This finishes part~(ii).
\end{proof}

In the following, we state some intermediate results for the proof of Theorem~\ref{thm:bootstrap}.
\begin{lem}\label{lemma-bootstrap-clt}
Suppose that Assumptions~\ref{ass-bootstrap-uniform},~\ref{ass-bootstrap-indicator}, and~\ref{ass-variance} hold. Then $\E^*\bigl[|D_i^{\frakp*}(\thetahat^\frakp)|^4\bigr]=\Opstar(1)$ in probability under $P_0$,
\[
\widehat{V}_n\defined\frac{1}{n}\sum_{i=1}^n\var^*\!\bigl(D_i^{\frakp*}(\thetahat^\frakp)\bigr)\to_p V,
\]
and
\[
\sqrt{n}\bigl(D^{\frakp*}(\thetahat^\frakp)-\E^*[D^{\frakp*}(\thetahat^\frakp)]\bigr)
\dtostar\No(0,\widehat{V}_n)
\quad\text{in probability under }P_0.
\]
\end{lem}

\begin{proof}[\textupandbold{Proof of Lemma~\ref{lemma-bootstrap-clt}}]
Write $\xi_{n,i}\defined n^{-1/2}\bigl(D_i^{\frakp*}(\thetahat^\frakp)-\E^*[D_i^{\frakp*}(\thetahat^\frakp)]\bigr)$. Conditionally on the original sample, $\{\xi_{n,i}\}_{i=1}^n$ are independent with mean zero and $\sigma_n^{*2}\defined\sum_{i=1}^n\var^*(\xi_{n,i})=\widehat{V}_n$. For the fourth-moment bound, Assumption~\ref{ass-bootstrap-indicator} transfers $D_i^{\frakp*}(\thetahat^\frakp)=D_i^{\frakt*}(\thetahat^\frakp)+\Delta_i^*$ from Lemma~\ref{lem:indicators}, with $\E^*[(\Delta_i^*)^2]=\Opstar(T^{-1})$. Under $P^*$ the target likelihood is information unbiased at $\thetahat^\frakp$, so the product-moment argument in the proof of Lemma~\ref{lemma-scaled-indicator}\textup{(i)}, with $\E_0$ replaced by $\E^*$ and with the envelopes of Assumption~\ref{ass-bootstrap-uniform}, yields $\E^*\bigl[|D_i^{\frakt*}(\thetahat^\frakp)|^4\bigr]=\Opstar(1)$. The cross terms in $\E^*\bigl[|D_i^{\frakt*}(\thetahat^\frakp)+\Delta_i^*|^4\bigr]$ are $\opstar(1)$ by the Cauchy--Schwarz inequality and the moment envelopes in Assumption~\ref{ass-bootstrap-uniform}, so $\E^*\bigl[|D_i^{\frakp*}(\thetahat^\frakp)|^4\bigr]=\Opstar(1)$. The same decomposition yields $\var^*\bigl(D_i^{\frakp*}(\thetahat^\frakp)\bigr)=\var^*\bigl(D_i^{\frakt*}(\thetahat^\frakp)\bigr)+\opstar(1)$. The target conditional variance converges to $V$ by Assumption~\ref{ass-variance}, transferred under Assumption~\ref{ass-bootstrap-uniform} because $\hat\tau\to_p\tau_0$. Therefore $\widehat{V}_n\to_p V$. By the Lyapunov criterion with $\delta=2$,
\[
\frac{1}{(\sigma_n^{*2})^{2}}\sum_{i=1}^n\E^*|\xi_{n,i}|^{4}
=
\frac{\E^*\bigl[|D_1^{\frakp*}(\thetahat^\frakp)|^4\bigr]}{n\,\sigma_n^{*4}}
=\Opstar(n^{-1})\to_p 0,
\]
using $\widehat{V}_n\to_p V>0$. Hence the Lindeberg condition holds under $P^*$, and the claimed conditional CLT follows.
\end{proof}

\begin{lem}\label{lemma-bootstrap-variance}
Suppose that Assumptions~\ref{ass-bootstrap-uniform},~\ref{ass-bootstrap-indicator},~\ref{ass-bootstrap-estimators}, and~\ref{ass-variance} hold, and that $n/T^2\to\rho\in[0,\infty)$. Then $\widehat{V}^{\frakp*}=V+\opstar(1)$ in probability under $P_0$.
\end{lem}

\begin{proof}[\textupandbold{Proof of Lemma~\ref{lemma-bootstrap-variance}}]
The argument parallels the proof of Lemma~\ref{lemma-variance-estimator}, with the fourth-order expansion taken around the fitted profile estimator $\thetahat^\frakp$. By Assumption~\ref{ass-bootstrap-uniform}, the derivative-moment bounds used there hold under $P^*$ at $\thetahat^\frakp$. In particular, $\widehat{V}_1^{\frakp*}(\thetahat^\frakp)=\Opstar(\sqrt{T})$. Assumption~\ref{ass-bootstrap-estimators} gives
\[
\thetahat^{\frakp*}-\thetahat^\frakp=\Opstar\bigl((nT)^{-1/2}\bigr)+\Opstar(T^{-1}),
\]
and therefore
\[
\widehat{V}_1^{\frakp*}(\thetahat^\frakp)(\thetahat^{\frakp*}-\thetahat^\frakp)
=\Opstar(n^{-1/2})+\Opstar(T^{-1/2})=\opstar(1).
\]
The quadratic and higher products are $\opstar(1)$ by the same calculation as for $\ell=\ell^\frakp$ in that proof. Hence $\widehat{V}^{\frakp*}(\thetahat^{\frakp*})=\widehat{V}^{\frakp*}(\thetahat^\frakp)+\opstar(1)$. Consistency of $\widehat{V}^{\frakp*}(\thetahat^\frakp)$ for $V$ then follows from $\widehat{V}_n\to_p V$ and the fourth-moment bound of Lemma~\ref{lemma-bootstrap-clt}, as in the proof of Lemma~\ref{lemma-variance-estimator}.
\end{proof}

\begin{proof}[\textupandbold{Proof of Theorem~\ref{thm:bootstrap}}]
The argument is the bootstrap analogue of the proof of Theorem~\ref{main-theorem}. Lemmas~\ref{lem:indicators} and~\ref{lemma-scaled-indicator} are transferred under Assumptions~\ref{ass-bootstrap-uniform} and~\ref{ass-bootstrap-indicator} (Remark~\ref{rem:bootstrap-primitives}), and Lemmas~\ref{lemma-bootstrap-clt} and~\ref{lemma-bootstrap-variance} supply the Lyapunov and variance steps under~$P^*$. Under $n/T^2\to\rho\in[0,\infty)$, Assumption~\ref{ass-bootstrap-indicator} transfers the expansion of Lemma~\ref{lemma-scaled-indicator}\textup{(ii)} under $P^*$, with $\thetahat^\frakp$ in the role of $\theta_0$ and $\thetahat^{\frakp*}$ in the role of $\thetahat^\frakp$. The information bias is the mean of the profile indicator at the common parameter that generates the outcomes, and the score bias arises from replacing that parameter by the maximizer. Thus,
\begin{align}
\sqrt{n}D^{\frakp*}(\thetahat^{\frakp*})
&=\sqrt{n}D^{\frakp*}(\thetahat^\frakp)
+\frac{\sqrt{n}}{T}B_T^{\mathrm{score}*}(\thetahat^\frakp)
+\opstar(1),\label{eq:bootstrap-indicator-expansion}\\
\sqrt{n}\,\E^*[D^{\frakp*}(\thetahat^\frakp)]
&=\frac{\sqrt{n}}{T}B_T^{\mathrm{info}*}(\thetahat^\frakp)
+\opstar(1).\nonumber
\end{align}
Writing $\kappa_T^*(\thetahat^\frakp)\defined B_T^{\mathrm{score}*}(\thetahat^\frakp)+B_T^{\mathrm{info}*}(\thetahat^\frakp)$ and centering the first display,
\[
\sqrt{n}D^{\frakp*}(\thetahat^{\frakp*})
=\sqrt{n}\bigl(D^{\frakp*}(\thetahat^\frakp)-\E^*[D^{\frakp*}(\thetahat^\frakp)]\bigr)
+\frac{\sqrt{n}}{T}\kappa_T^*(\thetahat^\frakp)
+\opstar(1).
\]
Lemma~\ref{lemma-bootstrap-clt} yields
\[
\sqrt{n}\bigl(D^{\frakp*}(\thetahat^\frakp)-\E^*[D^{\frakp*}(\thetahat^\frakp)]\bigr)
\dtostar\No(0,\widehat{V}_n)
\quad\text{in probability},
\]
with $\widehat{V}_n\to_p V$. Assumption~\ref{ass-bootstrap-bias} gives $\kappa_T^*(\thetahat^\frakp)-\kappa_T(\theta_0)=\op(1)$. Combined with $(\sqrt{n}/T)\kappa_T(\theta_0)\to_p\sqrt{\rho}\,\kappa(\theta_0)$ from the proof of Lemma~\ref{lemma-scaled-indicator}\textup{(ii)},
\[
\frac{\sqrt{n}}{T}\kappa_T^*(\thetahat^\frakp)\to_p\sqrt{\rho}\,\kappa(\theta_0).
\]
Lemma~\ref{lemma-bootstrap-variance} gives $\widehat{V}^{\frakp*}=V+\opstar(1)$, so $(\widehat{V}^{\frakp*})^{-1}\to_p V^{-1}$ by continuous mapping and nonsingularity of $V$ in Assumption~\ref{ass-variance}. Applying Slutsky's theorem under $P^*$,
\[
\sqrt{n}D^{\frakp*}(\thetahat^{\frakp*})\dtostar\No\bigl(\sqrt{\rho}\,\kappa(\theta_0),V\bigr)
\quad\text{in probability}.
\]
Writing $Z_n^{\frakp*}\defined\sqrt{n}D^{\frakp*}(\thetahat^{\frakp*})$ and $A_n^{\frakp*}\defined(\widehat{V}^{\frakp*})^{-1}$, continuous mapping for the quadratic form yields
\[
\calW^{\frakp*}=Z_n^{\frakp*\prime}A_n^{\frakp*}Z_n^{\frakp*}
\dtostar\chi^2_q(\lambda_\rho)
\quad\text{in probability},
\]
with $\lambda_\rho=\rho\,\kappa(\theta_0)'V^{-1}\kappa(\theta_0)$. The cdf of $\chi^2_q(\lambda_\rho)$ is continuous. Since $\Whatprofile\dto\chi^2_q(\lambda_\rho)$ by Theorem~\ref{main-theorem}, Polya's theorem yields
\[
\sup_{x\in\R}\bigl|P^*(\calW^{\frakp*}\le x)-P_0(\Whatprofile\le x)\bigr|\to_p 0.
\]
When $\rho=0$, the mean shift vanishes and $\lambda_\rho=0$, so the profile bootstrap limit is central $\chi^2_q$. This finishes the proof of Theorem~\ref{thm:bootstrap}.
\end{proof}

\section{Monte Carlo implementation details}\label{sec:mc-implementation}
This appendix records the parametric bootstrap draws and related Monte Carlo conventions used in Section~\ref{sec:simulations} and Appendix~\ref{sec:add-sim-results}.

\paragraph*{Parametric bootstrap draws.}
Bootstrap panels are generated recursively from the fitted transition density of the pseudo-likelihood under consideration. For the profile likelihood, with $\alphahat_i\defined\alphahatiT(\thetahat^\frakp)$,
\begin{equation}\label{eq:bootstrap-density}
Y_{it}^*\;\sim\; f_{it}\bigl(\,\cdot\,\bigm|Y_{i,t-1}^*,\ldots,Y_{i,t-L}^*,\calXiT;\,\thetahat^\frakp,\alphahat_i\bigr),
\end{equation}
independently across $i$ and recursively in $t$, with $Y_{i,s}^*=Y_{is}$ for $s\le 0$ (covariates and initial conditions held fixed). Equivalently, the bootstrap joint density of $Y_i^*$ given $(\calXiT,\calY_{i0})$ is
\[
\prod_{t=1}^T f_{it}\bigl(Y_{it}^*\bigm|Y_{i,t-1}^*,\ldots,Y_{i,t-L}^*,\calXiT;\,\thetahat^\frakp,\alphahat_i\bigr).
\]
The target bootstrap replaces $(\thetahat^\frakp,\alphahat_i)$ by $(\thetahat^\frakt,\alphai^\frakt(\thetahat^\frakt))$ (or the oracle fixed effect in designs where the target is implemented as an oracle likelihood).

In the binary panel designs (static logit and probit, dynamic probit, and the heteroskedastic design below), an individual is uninformative for the profile likelihood if $\sum_t y_{it}\in\{0,T\}$.
The profile IM test is computed on the informative units only. Its bootstrap critical values are obtained from a parametric bootstrap that redraws paths for units that become uninformative until the original informative set is restored, so that each bootstrap draw retains the same number of individuals as the original informative panel.
By contrast, the target IM test always uses all $n$ individuals---units with $\sum_t y_{it}\in\{0,T\}$ remain informative for the target incidental-parameter moment---and its bootstrap is a plain parametric draw on the full panel (no redraw and no reduction to the informative subset).
This distinction does not arise in the continuous Neyman--Scott design of Subsection~\ref{sec:add-sim-ns}.

\section{Additional Simulation Results}\label{sec:add-sim-results}
\noindent
This appendix reports supplementary Monte Carlo results: the unit-scale heteroskedastic designs for static logit, static probit, and dynamic probit (Subsection~\ref{sec:add-sim-hetero}), the omitted-interaction design (Subsection~\ref{sec:add-sim-omitted}), and a wrong-link static logit design (Subsection~\ref{sec:add-sim-link}). Neyman--Scott Monte Carlo evidence is reported with the corresponding closed-form illustration in Subsection~\ref{sec:add-sim-ns}.

\subsection{Heteroskedasticity}\label{sec:add-sim-hetero}
\noindent
Tables~\ref{tab:logitheteroskedasticrejectionrates}--\ref{tab:dynamicprobitheteroskedasticrejectionrates} report rejection rates for the unit-scale heteroskedastic design of Subsection~\ref{sec:sim-setup}. For each grid point $(n,T)$ with $n\in\{125,250,500,1000\}$ and $T\in\{5,10,15,20,25,30,40,50\}$, we draw a balanced panel with covariates $X_{it,1}$ and $X_{it,2}$ that are standard normal and satisfy $\mathrm{Corr}(X_{it,1},X_{it,2})=0.5$, and fixed effects $\alpha_i\sim\No(\bar X_{i,1},1)$ with $\bar X_{i,1}=T^{-1}\sum_t X_{it,1}$. In the static designs, let $m_{it}=\theta_1 X_{it,1}+\theta_2 X_{it,2}$ with $(\theta_1,\theta_2)=(1,-1)$. In the dynamic probit design, the index is $m_{it}^{\mathrm{dyn}}=\theta_1 X_{it,1}+\theta_2 X_{it,2}+\gamma Y_{i,t-1}$ with $(\theta_1,\theta_2,\gamma)=(1,-1,0.5)$, as under the correctly specified dynamic DGP. Each unit draws a latent scale $\sigma_i=\exp(0.75\,\xi_i)$ with $\xi_i\stackrel{\mathrm{iid}}{\sim}\No(0,1)$, and outcomes are generated from
\[
  \Pr(\Yit=1\mid X_{it,1},X_{it,2},\alpha_i,\sigma_i)
=F\!\left(\frac{m_{it}+\alpha_i}{\sigma_i}\right),
\]
where $F$ is the standard logistic CDF in the logit design and the standard normal CDF in the static and dynamic probit designs, while estimation still uses the additive homoskedastic index in the fitted covariates (including the lag in the dynamic case).

In the static designs (Tables~\ref{tab:logitheteroskedasticrejectionrates} and~\ref{tab:probitheteroskedasticsigmarejectionrates}), unit-scale heteroskedasticity is readily detected once $T$ is moderate. Bootstrap rejection for the target (and, in logit, the conditional) PIM rises quickly in both $n$ and $T$: at $T=10$ and $n=500$, target Boot is already $98.2\%$ in logit and $97.9\%$ in probit, and essentially $100\%$ at $n=1000$. Even at $T=5$ and $n=1000$, target Boot exceeds $70\%$ in both designs. The profile bootstrap follows the same pattern once $T$ is no longer very short, while asymptotic profile critical values remain severely oversized for large $n$ and small $T$. The contrast with Wald tests for $\theta$ is sharp at moderate $T$ (e.g.\ at $T=20$ and $n=500$ in logit, profile and target Boot both equal $100\%$, against Wald rejection below $25\%$). In logit, the conditional likelihood tracks the target throughout.

Dynamic probit (Table~\ref{tab:dynamicprobitheteroskedasticrejectionrates}) follows the same qualitative pattern, but power builds more slowly at the shortest panels. At $T=5$ and $n=1000$, target Boot is only $35.9\%$ and profile Boot $20.0\%$. Once $T$ is moderate, bootstrap rejection rises sharply (profile Boot $92.7\%$ and target Boot $80.2\%$ at $T=10$ and $n=500$; both above $99\%$ at $n=1000$). As under the null, asymptotic profile critical values remain unreliable when $n$ is large and $T$ is small, so the bootstrap is the relevant sized check. Wald rejection is high throughout much of this grid, including at short $T$ where the Wald test is oversized under correct specification, so it does not provide a clean comparison.

\input{simulations/panel_im_test_logit_hetero_review_table.tex}
\input{simulations/panel_im_test_probit_hetero_sigma_review_table.tex}
\input{simulations/panel_im_test_dyn_hetero_review_table.tex}

\subsection{Omitted interaction}\label{sec:add-sim-omitted}
\noindent
A second misspecified design follows the incorrect-mean alternative of \citet[Section~3.2, eq.~(7)]{horowitz1994bootstrap}, i.e.,  the true index includes the additive interaction $1.5\,X_{it,1}X_{it,2}$ while estimation continues to use the additive index in \eqref{eq:assumed-model-simulation}. Covariates and fixed effects are generated as in Subsection~\ref{sec:add-sim-hetero}, and outcomes follow
\[
\Pr(\Yit=1\mid X_{it,1},X_{it,2},\alpha_i)=F\!\left(\theta_1 X_{it,1}+\theta_2 X_{it,2}+1.5\,X_{it,1}X_{it,2}+\alpha_i\right),
\]
with $(\theta_1,\theta_2)=(1,-1)$ and $F$ the logistic or standard normal CDF according to the static logit or probit design. The dynamic probit omitted-interaction design uses the same interaction in the index together with the lagged dependent variable, as in the correctly specified dynamic design of Section~\ref{sec:sim-setup}. Thus both covariates enter linearly in the DGP, and misspecification arises from omitting the interaction term from the fitted model.
Tables~\ref{tab:logitomittedrejectionrates}--\ref{tab:dynamicprobitomittedrejectionrates} report rejection rates for the static logit, static probit, and dynamic probit designs under this omitted-interaction alternative.

Across Tables~\ref{tab:logitomittedrejectionrates}--\ref{tab:dynamicprobitomittedrejectionrates},
bootstrap power against the omitted interaction rises in both $n$ and $T$, and the
profile asymptotic IM remains severely oversized for large $n$ and small $T$, so the
bootstrap is again the relevant sized check. Relative to the unit-scale
heteroskedastic design, this alternative is harder to detect at short panels:
in static logit at $T=10$ and $n=500$, target Boot is only $25.6\%$ (profile Boot
$31.4\%$), and even at $T=15$ and $n=1000$ it reaches $80.0\%$. Static probit is
noticeably more powerful at the same grid points (target Boot $72.4\%$ and profile
Boot $84.1\%$ at $T=10$ and $n=500$; both near $100\%$ at $n=1000$). In logit, the
conditional likelihood continues to track the target.
Dynamic probit is intermediate: short-$T$ power remains modest (target Boot
$38.5\%$ and profile Boot $32.2\%$ at $T=5$ and $n=1000$), but once $T$ is moderate
rejection rises quickly (profile Boot $91.8\%$ and target Boot $94.5\%$ at $T=10$
and $n=1000$). Wald rejection for $\theta$ is high throughout much of these grids
because the omitted interaction biases the common-parameter estimator, so it does
not provide a clean comparison with a correctly sized specification check.
\input{simulations/panel_im_test_logit_omitted_review_table.tex}
\input{simulations/panel_im_test_probit_omitted_review_table.tex}
\input{simulations/panel_im_test_dyn_omitted_review_table.tex}

\subsection{Wrong link function}\label{sec:add-sim-link}
\noindent
A third misspecified design keeps the additive index and covariates of Subsection~\ref{sec:add-sim-hetero} but generates binary outcomes from the \emph{probit} link while estimation continues to use fixed-effects \emph{logit}. Specifically, letting $m_{it}=\theta_1 X_{it,1}+\theta_2 X_{it,2}$ with $(\theta_1,\theta_2)=(1,-1)$,
\[
\Pr(\Yit=1\mid X_{it,1},X_{it,2},\alpha_i)=\Phi(m_{it}+\alpha_i),
\]
where $\Phi$ is the standard normal CDF, whereas the fitted model remains logistic as in \eqref{eq:assumed-model-simulation}. Table~\ref{tab:logitlinkrejectionrates} reports rejection rates for $n\in\{125,250,500,1000\}$ and $T\in\{8,10,15,20,25,30,40,50\}$ (M=5000, B=99; cells with at least 4990 good replications). We omit $T=5$ because, under this wrong-link design, estimation is numerically unstable at that panel length since the profile fixed-effect fits and profile bootstrap redraws fail at high rates, so reliable Monte Carlo rejection frequencies are not available. For the retained grid, power for the target and conditional PIM is moderate (roughly 10--16\% rejection). This is consistent with the well-known fact that average partial effects under logit and probit are typically close \citep[Ch.~15]{wooldridge2010} since the link misspecification is subtle. The profile asymptotic IM can be severely oversized when $n$ is large and $T$ is moderate unless bootstrap critical values are used. This is consistent with the incidental-parameter bias of the profile PIM under parabolic asymptotics.
\input{simulations/panel_im_test_logit_link_review_table.tex}

\section{Example: Neyman--Scott}\label{sec:neyman-scott-short}
\noindent
This appendix illustrates Theorem~\ref{main-theorem} in the heterogeneous means model of \citet{neyman-scott}. Subsection~\ref{sec:neyman-scott-theory} derives the limiting distributions in closed form. Subsection~\ref{sec:add-sim-ns} reports the corresponding Monte Carlo size evidence.

\subsection{Limiting distributions}\label{sec:neyman-scott-theory}
\noindent
Symbolic checks were verified using the SymPy library in Python \citep{sympy}. Let $\Yit=\alphaio+\epsilon_{it}$ with $\epsilon_{it}\stackrel{\mathrm{iid}}{\sim}\No(0,\theta_0)$ across $i$ and $t$. To simplify the algebra, assume that $\theta_0\in\Theta=[\underline{\theta},\overline{\theta}]$, where $\underline{\theta}>0$ and $\overline{\theta}<\infty$. Ignoring constants, the log-likelihood is
\begin{equation*}
\elliT(\theta,\alpha_i)=-\frac{1}{2}\log(\theta)-\frac{1}{2T\theta}\sumt(\Yit-\alpha_i)^2.
\end{equation*}
As can easily be verified by the first-order condition, $\alphai^\frakt(\theta)=\alphaio$ for every $\theta$, so the target likelihood
\begin{equation}\label{eq:ns-short-target}
\ell_i^\frakt(\theta)=-\frac{1}{2}\log(\theta)-\frac{1}{2T\theta}\sumt(\Yit-\alphaio)^2
\end{equation}
coincides with the oracle (known-$\alphaio$) likelihood. Similarly,  $\alphahatiT(\theta)=\Yibar$, yielding the profile likelihood
\begin{equation*}
\ell_i^\frakp(\theta)=-\frac{1}{2}\log(\theta)-\frac{1}{2T\theta}\sumt(\Yit-\Yibar)^2.
\end{equation*}
Write $Q_i(\theta)=\theta^{-1}\sumt(\Yit-\alphaio)^2$ (so $Q_i(\theta_0)\sim\chi^2_T$) and
\[
\eta_i=\theta_0^{-1}\sumt(\Yit-\Yibar)^2\sim\chi^2_{T-1},\qquad
\xi_i=T\theta_0^{-1}(\Yibar-\alphaio)^2\sim\chi^2_1,
\]
where $\eta_i$ and $\xi_i$ are independent by Basu's theorem. We begin by analyzing the properties of the target PIM test. Differentiating \eqref{eq:ns-short-target} yields
\begin{equation*}
\begin{aligned}
\ell_{i,1}^\frakt(\theta)
&=-\frac{1}{2\theta}+\frac{Q_i(\theta)}{2T\theta}
=\frac{1}{2T\theta}\bigl(Q_i(\theta)-T\bigr)
=\frac{\theta_0}{2T\theta^2}\bigl(Q_i(\theta_0)-T\bigr)
+\frac{1}{2\theta}\Bigl(\frac{\theta_0}{\theta}-1\Bigr),\\
\ell_{i,2}^\frakt(\theta)
&=-\frac{2\ell_{i,1}^\frakt(\theta)}{\theta}-\frac{1}{2\theta^2}.
\end{aligned}
\end{equation*}
Hence,
\begin{equation*}
D_i^\frakt(\theta)=T\bigl(\ell_{i,1}^\frakt(\theta)\bigr)^2-\frac{2\ell_{i,1}^\frakt(\theta)}{\theta}-\frac{1}{2\theta^2}.
\end{equation*}
Evaluating moments with $Q_i(\theta_0)\sim\chi^2_T$ yields $\E_0\bigl[D_i^\frakt(\theta_0)\bigr]=0$, confirming the information unbiasedness of the target likelihood. Similarly,
\begin{equation*}
\var_0\bigl[D_i^\frakt(\theta_0)\bigr]=\frac{1}{2\theta_0^4}+\frac{1}{T\theta_0^4}=O(1),
\end{equation*}
confirming Lemma \ref{lemma-variance}. Independence across $i$ therefore implies $\var_0\bigl(\sqrt{n}\,D^\frakt(\theta_0)\bigr)\to\frac{1}{2\theta_0^4}$ as  $n,T\to\infty$, meaning that the indicator must be scaled by $\sqrt{n}$ and not by the square root of the sample size. Next, a second-order expansion yields
\[
\sqrt{n}\,D^\frakt(\thetahat^\frakt)=\sqrt{n}\,D^\frakt(\theta_0)+D_1^\frakt(\theta_0)\sqrt{n}\,(\thetahat^\frakt-\theta_0)+R_n,
\]
where $R_n=\frac12\sqrt{n}\,D_2^\frakt(\thetabar)(\thetahat^\frakt-\theta_0)^2$ for an intermediate value $\thetabar$ between $\thetahat^\frakt$ and $\theta_0$.
The oracle MLE admits the exact representation
\begin{equation*}
\thetahat^\frakt-\theta_0=\frac{2\theta_0^2}{n}\sumi\ell_{i,1}^\frakt(\theta_0),
\end{equation*}
so $\E_0\thetahat^\frakt=\theta_0$ and $\sqrt{nT}(\thetahat^\frakt-\theta_0)=\Op(1)$. Differentiating $D_i^\frakt$ gives
\begin{equation*}
\begin{aligned}
D_{i,1}^\frakt(\theta)
&=2T\ell_{i,1}^\frakt(\theta)\ell_{i,2}^\frakt(\theta)+\ell_{i,3}^\frakt(\theta)
=-\frac{4T\bigl(\ell_{i,1}^\frakt(\theta)\bigr)^2}{\theta}+\frac{(6-T)\ell_{i,1}^\frakt(\theta)}{\theta^2}+\frac{2}{\theta^3},\\
D_{i,2}^\frakt(\theta)
&=2T\bigl(\ell_{i,2}^\frakt(\theta)\bigr)^2+2T\ell_{i,1}^\frakt(\theta)\ell_{i,3}^\frakt(\theta)+\ell_{i,4}^\frakt(\theta)
=\frac{20T}{\theta^2}\bigl(\ell_{i,1}^\frakt(\theta)\bigr)^2+\frac{8T-24}{\theta^3}\ell_{i,1}^\frakt(\theta)+\frac{T-18}{2\theta^4}.
\end{aligned}
\end{equation*}
For the first derivative, score moments  yield $\E_0[D_{i,1}^\frakt(\theta_0)]=0$ and $\var_0(D_{i,1}^\frakt(\theta_0))=O(T)$.\footnote{The mean-zero property of $D_{i,1}^\frakt(\theta_0)$ is special to the Neyman--Scott target likelihood; in general $\E_0[D_{i,1}(\theta_0)]$ need not vanish (Lemma~\ref{lemma-rates}). However, since the leading contribution to $D_1^\frakt(\theta_0)$ is the $O(T)$ variance term, the asymptotic properties are not affected.} Consequently, the sample average satisfies $D_1^\frakt(\theta_0)=\Op(\sqrt{T/n})$. Combined with $\thetahat^\frakt-\theta_0=\Op((nT)^{-1/2})$, we get that the first remainder term is negligible, i.e.,
$D_1^\frakt(\theta_0)\sqrt{n}\,(\thetahat^\frakt-\theta_0)=\Op(n^{-1/2})=o_p(1)$.
Since $\Theta=[\underline{\theta},\overline{\theta}]$ with $\underline{\theta}>0$, the score display implies $\sup_{\theta\in\Theta}|\ell_{i,1}^\frakt(\theta)|=\Op(1)$ since the deterministic term is $O(1)$ uniformly on $\Theta$, while the stochastic term is uniformly $\Op(T^{-1/2})$. Therefore,  $\sup_{\theta\in\Theta}|D_{i,2}^\frakt(\theta)|=\Op(T)$, and hence $\sup_{\theta\in\Theta}|D_2^\frakt(\theta)|=\Op(T)$. Similarly, $D_{i,3}^\frakt(\theta)$ and $D_{i,4}^\frakt(\theta)$ consist of a stochastic part that is $\Op(\sqrt{T})$ at $\theta_0$ and $\Op(T)$ uniformly on $\Theta$, and a deterministic part that is $O(T)$ uniformly on $\Theta$, which confirms Lemma \ref{lemma-rates} in the Neyman-Scott model. Together with $\thetahat^\frakt-\theta_0=\Op((nT)^{-1/2})$, the remainder is $R_n=\Op\!\bigl(\sqrt{n}\cdot T\cdot (nT)^{-1}\bigr)=\Op(n^{-1/2})=o_p(1)$. Consequently, $\sqrt{n}\,D^\frakt(\thetahat^\frakt)=\sqrt{n}\,D^\frakt(\theta_0)+o_p(1),$ 
which delivers the plug-in equivalence of Lemma~\ref{lemma-scaled-indicator}\textup{(i)}.
The consistency of the variance estimator $\widehat{V}^\frakt(\thetahatstar)=n^{-1}\sumi\bigl(D_i^\frakt(\thetahatstar)\bigr)^2$ can be verified by a fourth-order Taylor expansion as in the proof of Lemma \ref{lemma-variance-estimator}, which yields  $\widehat{V}^\frakt(\thetahat^\frakt)=\widehat{V}^\frakt(\theta_0)+\op(1)$. At $\theta_0$, information unbiasedness yields $\E_0[\widehat{V}^\frakt(\theta_0)]=\var_0\bigl(D_i^\frakt(\theta_0)\bigr)=1/(2\theta_0^4)+1/(T\theta_0^4)$, and $\E_0\bigl[(D_i^\frakt(\theta_0))^4\bigr]=\Op(1)$, so independence across $i$ gives $\widehat{V}^\frakt(\theta_0)=1/(2\theta_0^4)+\op(1)$. Therefore, $\widehat{V}^\frakt(\thetahat^\frakt)\pto1/(2\theta_0^4)$, and an application of the continuous mapping theorem gives $\calW^\frakt\dto\chi^2_1$ as $n,T\to\infty$, as in Theorem~\ref{main-theorem}\textup{(i)}.

Next, we consider the profile PIM. 
Relative to the target score, the profile score differs only through the incidental-parameter term $b_{i,1}(\theta)=-\theta_0\xi_i/(2T\theta^2)$:
\begin{equation*}
\ell_{i,1}^\frakp(\theta)=\ell_{i,1}^\frakt(\theta)+b_{i,1}(\theta),
\end{equation*}
with $b_{i,2}(\theta)=\theta_0\xi_i/(T\theta^3)$ likewise entering the Hessian, $\ell_{i,2}^\frakp(\theta)=\ell_{i,2}^\frakt(\theta)+b_{i,2}(\theta)$. Substituting into the indicator and expanding the square gives
\begin{equation*}
D_i^\frakp(\theta)=D_i^\frakt(\theta)+\IB_i^\frakp(\theta),
\qquad
\IB_i^\frakp(\theta)=Tb_{i,1}(\theta)^2+2T\ell_{i,1}^\frakt(\theta)b_{i,1}(\theta)+b_{i,2}(\theta),
\end{equation*}
so the profile indicator is the target indicator plus an information-bias term driven by fixed-effect estimation. Using the independence between $\eta_i$ and $\xi_i$ and $\chi^2$ moments,
\begin{equation*}
\E_0\bigl[\IB_i^\frakp(\theta_0)\bigr]=\frac{3}{4T\theta_0^2}=O(T^{-1}),
\qquad
\E_0\bigl[\ell_{i,1}^\frakp(\theta_0)\bigr]=-\frac{1}{2T\theta_0}=O(T^{-1}),
\end{equation*}
confirming that both the information- and the score bias of the profile likelihood are of order $O(T^{-1})$. The fixed-effects MLE satisfies the same score representation as the target case,
\begin{equation*}
\thetahat^\frakp-\theta_0=\frac{2\theta_0^2}{n}\sumi\ell_{i,1}^\frakp(\theta_0)
=\frac{2\theta_0^2}{n}\sumi\bigl(\ell_{i,1}^\frakt(\theta_0)+b_{i,1}(\theta_0)\bigr),
\end{equation*}
so $\E_0[\thetahat^\frakp-\theta_0]=-\theta_0/T+O((nT)^{-1})$ and $\sqrt{n}(\thetahat^\frakp-\theta_0)=\Op(\sqrt{n}/T)+\Op(T^{-1/2})$. A second-order expansion yields
\[
\sqrt{n}\,D^\frakp(\thetahat^\frakp)=\sqrt{n}\,D^\frakp(\theta_0)+D_1^\frakp(\theta_0)\sqrt{n}\,(\thetahat^\frakp-\theta_0)+R_n^\frakp,
\]
where $R_n^\frakp=\tfrac12\sqrt{n}\,D_2^\frakp(\thetabar)(\thetahat^\frakp-\theta_0)^2$ for an intermediate value $\thetabar$. Decomposing the leading term,
\[
\sqrt{n}\,D^\frakp(\theta_0)=\sqrt{n}\,D^\frakt(\theta_0)+\frac{1}{\sqrt{n}}\sumi\IB_i^\frakp(\theta_0),
\]
so that the information bias in the asymptotic distribution of the profile indicator is
\begin{equation*}
\sqrt{n}\,\E_0\bigl[D^\frakp(\theta_0)\bigr]=\sqrt{n}\,\E_0\bigl[\IB_i^\frakp(\theta_0)\bigr]=\frac{3\sqrt{n}}{4T\theta_0^2}=O\!\left(\frac{\sqrt{n}}{T}\right).
\end{equation*}
The first remainder term is $D_{i,1}^\frakp(\theta_0)=D_{i,1}^\frakt(\theta_0)+\IB_{i,1}^\frakp(\theta_0)$, where differentiating $\IB_i^\frakp$ introduces products of $b_{i,1}$, $b_{i,2}$, and target derivatives. Unlike the target case, $\E_0[D_{i,1}^\frakp(\theta_0)]=1/(2\theta_0^3)-2/(T\theta_0^3)=O(1)$ does not vanish, so score bias enters through
\[
\E_0\bigl[D_1^\frakp(\theta_0)\sqrt{n}\,(\thetahat^\frakp-\theta_0)\bigr]
=\E_0\bigl[D_{i,1}^\frakp(\theta_0)\bigr]\,\sqrt{n}\,\E_0\bigl[\thetahat^\frakp-\theta_0\bigr]
+\sqrt{n}\,\cov_0\bigl(D_1^\frakp(\theta_0),\thetahat^\frakp-\theta_0\bigr).
\]
The covariance term is negligible for the leading mean bias. Using $\thetahat^\frakp-\theta_0=(2\theta_0^2/n)\sumi\ell_{i,1}^\frakp(\theta_0)$ and independence across~$i$,
\[
\cov_0\bigl(D_1^\frakp(\theta_0),\thetahat^\frakp-\theta_0\bigr)
=\frac{2\theta_0^2}{n}\,\cov_0\bigl(D_{i,1}^\frakp(\theta_0),\ell_{i,1}^\frakp(\theta_0)\bigr)=O(n^{-1}),
\]
because cross-index terms factor and only the same-$i$ overlap remains. Hence,
\[
\E_0\bigl[D_1^\frakp(\theta_0)\sqrt{n}\,(\thetahat^\frakp-\theta_0)\bigr]
=-\frac{\sqrt{n}}{2T\theta_0^2}+O\!\left(\frac{\sqrt{n}}{T^2}\right),
\]
using $\E_0[D_{i,1}^\frakt(\theta_0)]=0$. The target contribution $D_1^\frakt(\theta_0)\sqrt{n}(\thetahat^\frakp-\theta_0)$ and the centered fluctuation of the linear term are $o_p(1)$, as in the target proof. The quadratic remainder is of the same order as these biases. Differentiating the profile indicator and taking the $\chi^2_{T-1}$ moment gives
\[
\E_0\bigl[D_{i,2}^\frakp(\theta_0)\bigr]=\frac{T^2-6T+14}{2T\theta_0^4}=\frac{T}{2\theta_0^4}+O(1),
\]
and therefore $\E_0[D_2^\frakp(\theta_0)]=T/(2\theta_0^4)+O(1)$. Since $\thetabar$ lies between $\thetahat^\frakp$ and $\theta_0$, replacing $D_2^\frakp(\thetabar)$ by $D_2^\frakp(\theta_0)$ changes the expectation of $R_n^\frakp$ only at order $O(\sqrt{n}/T^2)$. The contribution of $\var_0(\thetahat^\frakp)$ and of the covariance between $D_2^\frakp(\theta_0)$ and $\thetahat^\frakp-\theta_0$ is $O(n^{-1/2})$, so
\[
\E_0[R_n^\frakp]
=\frac12\sqrt{n}\,\E_0[D_2^\frakp(\theta_0)]\bigl(\E_0[\thetahat^\frakp-\theta_0]\bigr)^2+O(n^{-1/2})+O\!\left(\frac{\sqrt{n}}{T^2}\right)
=\frac{\sqrt{n}}{4T\theta_0^2}+O(n^{-1/2})+O\!\left(\frac{\sqrt{n}}{T^2}\right).
\]
The centered part of $R_n^\frakp$ is $o_p(1)$. Combining the information bias, the linear score bias, and this quadratic score bias,
\begin{equation*}
\sqrt{n}\,D^\frakp(\thetahat^\frakp)
=\sqrt{n}\,D^\frakt(\theta_0)+\frac{\sqrt{n}}{T}\left(\frac{3}{4\theta_0^2}-\frac{1}{2\theta_0^2}+\frac{1}{4\theta_0^2}\right)+o_p(1)
=\sqrt{n}\,D^\frakt(\theta_0)+\frac{\sqrt{n}}{2T\theta_0^2}+o_p(1).
\end{equation*}
Under Assumption~\ref{ass-n-T-rate}\textup{(ii)}, $n/T^2\to\rho\in[0,\infty)$, the mean shift converges to $\sqrt{\rho}/(2\theta_0^2)$ and vanishes when $\rho=0$. Overall,
\begin{equation*}
\sqrt{n}\,D^\frakp(\thetahat^\frakp)=\sqrt{n}\,D^\frakt(\theta_0)+O\!\left(\frac{\sqrt{n}}{T}\right)+o_p(1),
\end{equation*}
confirming Lemma~\ref{lemma-scaled-indicator}\textup{(ii)}. For the variance, note that $\var_0(D_i^\frakp(\theta_0))=1/(2\theta_0^4)+O(T^{-1})$, since the profile information-bias pieces change the individual variance only at order $T^{-1}$, while the $O(\sqrt{n}/T)$ mean shift above is deterministic at first order. Hence $\var_0(\sqrt{n}\,D^\frakp(\thetahat^\frakp))\to V\defined1/(2\theta_0^4)$, the same limit as for the target indicator. Slutsky's theorem therefore gives
\[
\sqrt{n}\,D^\frakp(\thetahat^\frakp)\dto \No\!\Bigl(\sqrt{\rho}\,\kappa(\theta_0),\,V\Bigr),
\qquad
\kappa(\theta_0)=\plim_{T\to\infty}\left(\frac{3}{4\theta_0^2}-\frac{1}{2\theta_0^2}+\frac{1}{4\theta_0^2}\right)=\frac{1}{2\theta_0^2},
\]
where $\kappa(\theta_0)$ collects the information- and score-bias limits derived above. With a consistent variance estimator $\widehat{V}^\frakp(\thetahat^\frakp)\pto V$ (as for $\widehat{V}^\frakt$ above), the continuous mapping theorem yields $\calW^\frakp=(\sqrt{n}\,D^\frakp(\thetahat^\frakp))^2/\widehat{V}^\frakp(\thetahat^\frakp)\dto\chi^2_1(\lambda_\rho)$ when $\rho\in(0,\infty)$ and $\calW^\frakp\dto\chi^2_1$ when $\rho=0$. When $\rho\in(0,\infty)$, 
\begin{equation*}
\sqrt{n}\,D^\frakp(\thetahat^\frakp)\dto \No\!\Bigl(\tfrac{\sqrt{\rho}}{2\theta_0^2},\,\tfrac{1}{2\theta_0^4}\Bigr),
\qquad
\calW^\frakp\dto \chi^2_1(\lambda_\rho)=\chi^2_1(\rho/2),
\end{equation*}
where $\mu\defined\sqrt{\rho}\,\kappa(\theta_0)$ and 
\[
\lambda_\rho=\mu'V^{-1}\mu=\rho\,\kappa(\theta_0)^2 V^{-1}
=\rho\left(\frac{1}{2\theta_0^2}\right)^2(2\theta_0^4)=\frac{\rho}{2}.
\]

\subsection{Monte Carlo}\label{sec:add-sim-ns}
\noindent
We simulate the same heterogeneous-means model. For each grid point $(n,T)$ with $n\in\{125,250,500,1000\}$ and $T\in\{5,10,15,20,25,30,40,50\}$, $n$ individuals are observed over $T$ periods. Individual effects satisfy $\alpha_i\sim\No(0,1)$, idiosyncratic errors $\varepsilon_{it}$ are independent $\No(0,\theta_0)$, and outcomes are generated as
\[
\Yit=\alpha_i+\varepsilon_{it},
\]
with $\theta_0=1$. The common variance $\theta$ is estimated from the panel using the target (oracle) and profile likelihoods. Results are based on $5{,}000$ Monte Carlo replications; bootstrap critical values use $99$ parametric bootstrap draws.
Table~\ref{tab:ns_rejection_rates} reports rejection rates for the asymptotic IM test (IM), the bootstrap IM test (Boot), and the Wald test (Wald) based on the profile and target likelihoods.

Table~\ref{tab:ns_rejection_rates} matches the closed-form predictions above. The target PIM is correctly sized throughout the grid under both asymptotic and bootstrap critical values, as is Wald inference based on the oracle estimator. For the profile PIM, asymptotic rejection moves with $\rho$ as predicted by the noncentrality $\lambda_\rho=\rho/2$. The sharpest check is $T=5$ and $n=1000$, where $\rho=40$: $\chi^2_1(20)$ rejects with probability about $99\%$, and the table gives $98.3\%$. At larger $T$ the empirical rates lie below that limiting probability, because the variance of the indicator has not yet settled at $V$, but they fall toward the nominal level as $\rho$ shrinks ($6.1\%$ at $T=50$ and $n=1000$). The parametric bootstrap restores size near the nominal $5\%$ level throughout the grid, since it matches $\chi^2_1(\lambda_\rho)$ rather than a central $\chi^2_1$. By contrast, the profile Wald test remains severely oversized even at $T=50$ and $n=1000$ ($89.2\%$), because $\thetahat^\frakp$ itself carries the incidental-parameter bias. In short, the PIM indicator is first-order robust under rectangular asymptotics, and the residual parabolic bias of the profile statistic is removable by the bootstrap.
\input{simulations/panel_im_test_neyman_scott_review_table.tex}

\section{Illustrative examples}\label{sec:illustrative-example}
\noindent
This section illustrates the \emph{exploratory} use of the profile PIM test developed in this paper. 
In each application we follow the same workflow: (i) detect misspecification with the omnibus profile IM statistic, (ii) diagnose directions of failure with blockwise IM components and targeted LR/LM tests, (iii) propose a parsimonious repair, and (iv) re-test the repaired model with both asymptotic and bootstrap inference. The first two examples use the full estimation sample for exploration. The third example adds a random train/hold-out split by plant to separate exploratory specification search from out-of-sample validation.

\subsection{Dynamic LFP probit (PSID)}
Using the dynamic FE probit specification of \citet{fernandez-val-probit} on the PSID panel of informative women ($n=664$, $T=9$), we start from

\begin{equation}\label{eq:fv-lfp-spec}
\Pr(P_{it}=1\mid P_{i,t-1},X_{it},\alpha_i)
=\Phi\bigl(\theta_P P_{i,t-1}+X_{it}'\theta_X+\alpha_i\bigr).
\end{equation}

The omnibus profile PIM rejects strongly under both asymptotic and bootstrap inference (Table~\ref{tab:fv-lfp-pim}). Diagnostics then localize the failure mainly to fertility, age, and time-related moments (Table~\ref{tab:fv-lfp-diagnostics}), while targeted tests rank the kids0--2$\times$lag interaction first (LR $p=0.006$), ahead of a weaker AR(2) signal and with heteroskedasticity also active. We therefore take that interaction as the first targeted repair because it is both economically interpretable and statistically most supported. The repaired model improves fit (higher profile likelihood), but the omnibus PIM still rejects (Table~\ref{tab:fv-lfp-repair}), indicating that the interaction addresses only one component of misspecification and that additional structure (for example richer time/age effects or scale heterogeneity) is still needed.
\input{applied_example/FV_lfp_pim_results.tex}

\input{applied_example/FV_lfp_diagnostics.tex}

\input{applied_example/FV_lfp_repair_pim.tex}

\subsection{Static union probit (NLSY)}
In the NLSY union panel of informative men ($n=246$, $T=8$), the baseline FE probit with marital status and year effects is also rejected by the omnibus PIM under both asymptotic and bootstrap inference (Table~\ref{tab:union-pim}). Diagnostics and targeted tests indicate omitted industry structure as the main source of misspecification (Table~\ref{tab:union-diagnostics}), as manufacturing and wholesale/retail have very strong LR evidence under both asymptotic and bootstrap critical values. Heteroskedasticity is weaker and mixed across benchmarks (asymptotic LM $p=0.078$; bootstrap $p=0.030$). Guided by the industry ranking, we run a short exploratory sequence and retain the smallest specification that eliminates rejection: hours worked plus manufacturing and construction indicators. Mean-fit LR tests in the same table support that enrichment (manufacturing LR $=18.74$; hours plus the three-industry bundle LR $=31.42$), and the repaired model passes the omnibus PIM under both asymptotic and bootstrap inference (Table~\ref{tab:union-repair}). Here the IM test therefore plays both an exploratory and constructive role.

\input{applied_example/union_pim_results.tex}

\input{applied_example/union_diagnostics.tex}

\input{applied_example/union_repair_pim.tex}

\subsection{Dynamic export probit (Chile ENIA)}
We apply the same workflow to export participation in the Chilean manufacturing panel underlying \citet{pavcnik-2002} and \citet{deloecker-2007}. Plants with time-varying export status are retained from the balanced 1995--2005 ENIA extract ($n=399$ movers; $T=10$ after lagging). A random 70\% subset of plants forms the \emph{exploration} sample ($n=279$). The remaining 30\% are held out for validation. For a conventional dynamic export specification---lagged export status, log employment, and year effects---the central-$\chi^2$ benchmark rejects, but the parametric bootstrap does not ($p=0.15$; Table~\ref{tab:enia-pim}), suggesting that asymptotic profile-PIM inference is anti-conservative in this short panel. Blockwise diagnostics show the same pattern: lag, size, and time reject under the asymptotic benchmark, while bootstrap $p$-values for these blocks are larger (Table~\ref{tab:enia-diagnostics}). Targeted diagnostics nevertheless indicate skill intensity as a plausible enrichment (LR $p=0.029$). A pre-announced specification search then adds foreign ownership and skill intensity; the smallest non-rejecting candidate (S5) passes both asymptotic and bootstrap omnibus tests on the exploration sample (Table~\ref{tab:enia-repair}), and the repaired specification also fails to reject on the hold-out plants under either benchmark (Table~\ref{tab:enia-holdout}), although only 32 of 120 hold-out plants contribute finite IM indicators. Time-invariant location controls are excluded from the search because they are collinear with plant fixed effects.

\input{applied_example/enia_pim_results.tex}

\input{applied_example/enia_diagnostics.tex}

\input{applied_example/enia_repair_pim.tex}

\input{applied_example/enia_holdout_pim.tex}

\bibliographystyle{agsm}
\bibliography{mybib}

\end{document}

%% file: simulations/panel_im_test_dyn_correct_review_table.tex
\begin{landscape}
\begin{table}[p]
\centering
\caption{Rejection Rates for the Dynamic Probit Model}
\label{tab:dynamicprobitrejectionrates}
\small
\setlength{\tabcolsep}{4.5pt}
\begin{tabular}{llcccccccccccccccc}
\toprule
 & & \multicolumn{3}{c}{$n=125$} & & \multicolumn{3}{c}{$n=250$} & & \multicolumn{3}{c}{$n=500$} & & \multicolumn{3}{c}{$n=1000$} \\
\cmidrule(lr){3-5} \cmidrule(lr){7-9} \cmidrule(lr){11-13} \cmidrule(lr){15-17}
$T$ & Method & IM & Boot & Wald & & IM & Boot & Wald & & IM & Boot & Wald & & IM & Boot & Wald \\
\midrule
5 & Profile & 0.0472 & 0.0470 & 0.9794 & & 0.0468 & 0.0848 & 1.0000 & & 0.4094 & 0.1092 & 1.0000 & & 0.9940 & 0.1182 & 1.0000 \\
 & Target & 0.2864 & 0.0662 & 0.0478 & & 0.1924 & 0.0552 & 0.0490 & & 0.1256 & 0.0526 & 0.0466 & & 0.0812 & 0.0550 & 0.0430 \\
\addlinespace
10 & Profile & 0.0668 & 0.0680 & 0.8582 & & 0.0500 & 0.0706 & 0.9962 & & 0.1496 & 0.0638 & 1.0000 & & 0.7298 & 0.0446 & 1.0000 \\
 & Target & 0.2130 & 0.0602 & 0.0476 & & 0.1428 & 0.0620 & 0.0426 & & 0.0990 & 0.0582 & 0.0532 & & 0.0674 & 0.0566 & 0.0470 \\
\addlinespace
15 & Profile & 0.0760 & 0.0620 & 0.7000 & & 0.0566 & 0.0686 & 0.9538 & & 0.0844 & 0.0632 & 0.9994 & & 0.3276 & 0.0564 & 1.0000 \\
 & Target & 0.1802 & 0.0572 & 0.0436 & & 0.1312 & 0.0620 & 0.0476 & & 0.0856 & 0.0572 & 0.0528 & & 0.0620 & 0.0576 & 0.0470 \\
\addlinespace
20 & Profile & 0.0902 & 0.0626 & 0.5702 & & 0.0526 & 0.0620 & 0.8872 & & 0.0642 & 0.0598 & 0.9954 & & 0.1702 & 0.0590 & 1.0000 \\
 & Target & 0.1788 & 0.0608 & 0.0510 & & 0.1216 & 0.0574 & 0.0486 & & 0.0874 & 0.0570 & 0.0496 & & 0.0600 & 0.0586 & 0.0488 \\
\addlinespace
25 & Profile & 0.0884 & 0.0550 & 0.4708 & & 0.0588 & 0.0606 & 0.7924 & & 0.0536 & 0.0594 & 0.9838 & & 0.1124 & 0.0586 & 1.0000 \\
 & Target & 0.1688 & 0.0532 & 0.0490 & & 0.1190 & 0.0628 & 0.0486 & & 0.0842 & 0.0572 & 0.0550 & & 0.0710 & 0.0666 & 0.0498 \\
\addlinespace
30 & Profile & 0.0970 & 0.0532 & 0.4152 & & 0.0712 & 0.0654 & 0.7038 & & 0.0514 & 0.0580 & 0.9558 & & 0.0890 & 0.0612 & 0.9998 \\
 & Target & 0.1600 & 0.0526 & 0.0548 & & 0.1260 & 0.0686 & 0.0546 & & 0.0822 & 0.0600 & 0.0488 & & 0.0648 & 0.0596 & 0.0442 \\
\addlinespace
40 & Profile & 0.1044 & 0.0546 & 0.3158 & & 0.0668 & 0.0558 & 0.5666 & & 0.0520 & 0.0596 & 0.8778 & & 0.0664 & 0.0628 & 0.9956 \\
 & Target & 0.1558 & 0.0538 & 0.0492 & & 0.1064 & 0.0586 & 0.0546 & & 0.0832 & 0.0612 & 0.0488 & & 0.0704 & 0.0678 & 0.0418 \\
\addlinespace
50 & Profile & 0.1176 & 0.0598 & 0.2494 & & 0.0756 & 0.0598 & 0.4728 & & 0.0462 & 0.0528 & 0.7996 & & 0.0594 & 0.0578 & 0.9810 \\
 & Target & 0.1500 & 0.0576 & 0.0500 & & 0.1152 & 0.0610 & 0.0576 & & 0.0774 & 0.0552 & 0.0504 & & 0.0604 & 0.0556 & 0.0496 \\
\bottomrule
\end{tabular}
\end{table}
\end{landscape}

%% file: simulations/panel_im_test_logit_hetero_review_table.tex
\begin{landscape}
\begin{table}[p]
\centering
\caption{Rejection Rates for Static Logit Model under Unit-Scale Heteroskedasticity}
\label{tab:logitheteroskedasticrejectionrates}
\small
\setlength{\tabcolsep}{4.5pt}
\begin{tabular}{llcccccccccccccccc}
\toprule
 & & \multicolumn{3}{c}{$n=125$} & & \multicolumn{3}{c}{$n=250$} & & \multicolumn{3}{c}{$n=500$} & & \multicolumn{3}{c}{$n=1000$} \\
\cmidrule(lr){3-5} \cmidrule(lr){7-9} \cmidrule(lr){11-13} \cmidrule(lr){15-17}
$T$ & Method & IM & Boot & Wald & & IM & Boot & Wald & & IM & Boot & Wald & & IM & Boot & Wald \\
\midrule
5 & Profile & 0.1830 & 0.1484 & 0.2988 & & 0.8450 & 0.2200 & 0.4861 & & 1.0000 & 0.2900 & 0.7402 & & 1.0000 & 0.4346 & 0.9512 \\
 & Conditional & 0.0636 & 0.0390 & 0.0914 & & 0.0618 & 0.0608 & 0.1122 & & 0.0834 & 0.1106 & 0.1472 & & 0.2496 & 0.3742 & 0.2276 \\
 & Target & 0.0576 & 0.0292 & 0.0844 & & 0.0770 & 0.0614 & 0.0906 & & 0.1858 & 0.2036 & 0.1154 & & 0.6352 & 0.7006 & 0.1620 \\
\addlinespace
10 & Profile & 0.4270 & 0.5164 & 0.1772 & & 0.9510 & 0.8758 & 0.2518 & & 1.0000 & 0.9902 & 0.3870 & & 1.0000 & 0.9998 & 0.6062 \\
 & Conditional & 0.1744 & 0.0884 & 0.1126 & & 0.5406 & 0.4468 & 0.1336 & & 0.9610 & 0.9514 & 0.1852 & & 1.0000 & 1.0000 & 0.2606 \\
 & Target & 0.2300 & 0.1092 & 0.1128 & & 0.6722 & 0.5574 & 0.1200 & & 0.9876 & 0.9818 & 0.1624 & & 1.0000 & 1.0000 & 0.2130 \\
\addlinespace
15 & Profile & 0.7858 & 0.8164 & 0.1562 & & 0.9966 & 0.9966 & 0.1858 & & 1.0000 & 1.0000 & 0.2438 & & 1.0000 & 1.0000 & 0.3600 \\
 & Conditional & 0.5838 & 0.3450 & 0.1292 & & 0.9756 & 0.9394 & 0.1672 & & 0.9998 & 1.0000 & 0.2132 & & 1.0000 & 1.0000 & 0.3230 \\
 & Target & 0.6436 & 0.3954 & 0.1248 & & 0.9836 & 0.9606 & 0.1598 & & 1.0000 & 0.9998 & 0.1932 & & 1.0000 & 1.0000 & 0.2838 \\
\addlinespace
20 & Profile & 0.9502 & 0.9460 & 0.1542 & & 0.9998 & 0.9998 & 0.1570 & & 1.0000 & 1.0000 & 0.1992 & & 1.0000 & 1.0000 & 0.2484 \\
 & Conditional & 0.9012 & 0.7008 & 0.1678 & & 0.9988 & 0.9966 & 0.1844 & & 1.0000 & 1.0000 & 0.2490 & & 1.0000 & 1.0000 & 0.3538 \\
 & Target & 0.9208 & 0.7432 & 0.1652 & & 0.9994 & 0.9982 & 0.1718 & & 1.0000 & 1.0000 & 0.2348 & & 1.0000 & 1.0000 & 0.3244 \\
\addlinespace
25 & Profile & 0.9904 & 0.9836 & 0.1640 & & 1.0000 & 1.0000 & 0.1804 & & 1.0000 & 1.0000 & 0.1760 & & 1.0000 & 1.0000 & 0.2000 \\
 & Conditional & 0.9818 & 0.8924 & 0.1768 & & 1.0000 & 1.0000 & 0.2084 & & 1.0000 & 1.0000 & 0.2794 & & 1.0000 & 1.0000 & 0.3794 \\
 & Target & 0.9856 & 0.9080 & 0.1760 & & 1.0000 & 1.0000 & 0.2046 & & 1.0000 & 1.0000 & 0.2634 & & 1.0000 & 1.0000 & 0.3508 \\
\addlinespace
30 & Profile & 0.9990 & 0.9948 & 0.1758 & & 1.0000 & 1.0000 & 0.1766 & & 1.0000 & 1.0000 & 0.1884 & & 1.0000 & 1.0000 & 0.1924 \\
 & Conditional & 0.9974 & 0.9708 & 0.1984 & & 1.0000 & 1.0000 & 0.2428 & & 1.0000 & 1.0000 & 0.3110 & & 1.0000 & 1.0000 & 0.4276 \\
 & Target & 0.9980 & 0.9750 & 0.1962 & & 1.0000 & 1.0000 & 0.2364 & & 1.0000 & 1.0000 & 0.3024 & & 1.0000 & 1.0000 & 0.4016 \\
\addlinespace
40 & Profile & 0.9998 & 0.9992 & 0.2086 & & 1.0000 & 1.0000 & 0.2046 & & 1.0000 & 1.0000 & 0.2160 & & 1.0000 & 1.0000 & 0.2106 \\
 & Conditional & 0.9998 & 0.9974 & 0.2376 & & 1.0000 & 1.0000 & 0.2814 & & 1.0000 & 1.0000 & 0.3526 & & 1.0000 & 1.0000 & 0.4868 \\
 & Target & 0.9998 & 0.9978 & 0.2354 & & 1.0000 & 1.0000 & 0.2762 & & 1.0000 & 1.0000 & 0.3434 & & 1.0000 & 1.0000 & 0.4674 \\
\addlinespace
50 & Profile & 1.0000 & 1.0000 & 0.2490 & & 1.0000 & 1.0000 & 0.2554 & & 1.0000 & 1.0000 & 0.2556 & & 1.0000 & 1.0000 & 0.2584 \\
 & Conditional & 1.0000 & 0.9998 & 0.2838 & & 1.0000 & 1.0000 & 0.3116 & & 1.0000 & 1.0000 & 0.3936 & & 1.0000 & 1.0000 & 0.5296 \\
 & Target & 1.0000 & 0.9998 & 0.2826 & & 1.0000 & 1.0000 & 0.3068 & & 1.0000 & 1.0000 & 0.3846 & & 1.0000 & 1.0000 & 0.5110 \\
\bottomrule
\end{tabular}
\end{table}
\end{landscape}

%% file: simulations/panel_im_test_probit_hetero_sigma_review_table.tex
\begin{landscape}
\begin{table}[p]
\centering
\caption{Rejection Rates for Static Probit Model under Unit-Scale Heteroskedasticity}
\label{tab:probitheteroskedasticsigmarejectionrates}
\small
\setlength{\tabcolsep}{4.5pt}
\begin{tabular}{llcccccccccccccccc}
\toprule
 & & \multicolumn{3}{c}{$n=125$} & & \multicolumn{3}{c}{$n=250$} & & \multicolumn{3}{c}{$n=500$} & & \multicolumn{3}{c}{$n=1000$} \\
\cmidrule(lr){3-5} \cmidrule(lr){7-9} \cmidrule(lr){11-13} \cmidrule(lr){15-17}
$T$ & Method & IM & Boot & Wald & & IM & Boot & Wald & & IM & Boot & Wald & & IM & Boot & Wald \\
\midrule
5 & Profile & 0.0356 & 0.1130 & 0.1920 & & 0.4544 & 0.1754 & 0.2376 & & 0.9766 & 0.2024 & 0.3114 & & 1.0000 & 0.2674 & 0.4246 \\
 & Target & 0.0974 & 0.0260 & 0.4180 & & 0.1008 & 0.0432 & 0.6516 & & 0.2346 & 0.1566 & 0.8888 & & 0.7710 & 0.7378 & 0.9940 \\
\addlinespace
10 & Profile & 0.3192 & 0.5138 & 0.2388 & & 0.8928 & 0.8472 & 0.3066 & & 0.9984 & 0.9614 & 0.4388 & & 1.0000 & 0.9910 & 0.6302 \\
 & Target & 0.2478 & 0.0778 & 0.6532 & & 0.7322 & 0.4954 & 0.8850 & & 0.9910 & 0.9792 & 0.9906 & & 1.0000 & 1.0000 & 0.9998 \\
\addlinespace
15 & Profile & 0.7376 & 0.8136 & 0.4400 & & 0.9866 & 0.9876 & 0.6292 & & 0.9998 & 0.9990 & 0.8504 & & 1.0000 & 1.0000 & 0.9820 \\
 & Target & 0.6936 & 0.3324 & 0.7800 & & 0.9826 & 0.9304 & 0.9514 & & 0.9994 & 0.9996 & 0.9986 & & 1.0000 & 1.0000 & 1.0000 \\
\addlinespace
20 & Profile & 0.9142 & 0.9218 & 0.6152 & & 0.9984 & 0.9990 & 0.8182 & & 1.0000 & 1.0000 & 0.9600 & & 1.0000 & 1.0000 & 0.9992 \\
 & Target & 0.9076 & 0.6634 & 0.8452 & & 0.9982 & 0.9922 & 0.9778 & & 1.0000 & 1.0000 & 1.0000 & & 1.0000 & 1.0000 & 1.0000 \\
\addlinespace
25 & Profile & 0.9698 & 0.9584 & 0.7068 & & 0.9996 & 0.9998 & 0.9056 & & 1.0000 & 1.0000 & 0.9900 & & 1.0000 & 1.0000 & 1.0000 \\
 & Target & 0.9726 & 0.8354 & 0.8864 & & 0.9996 & 0.9990 & 0.9888 & & 1.0000 & 1.0000 & 0.9996 & & 1.0000 & 1.0000 & 1.0000 \\
\addlinespace
30 & Profile & 0.9882 & 0.9812 & 0.7886 & & 1.0000 & 1.0000 & 0.9504 & & 1.0000 & 1.0000 & 0.9974 & & 1.0000 & 1.0000 & 1.0000 \\
 & Target & 0.9886 & 0.9254 & 0.9112 & & 1.0000 & 1.0000 & 0.9932 & & 1.0000 & 1.0000 & 1.0000 & & 1.0000 & 1.0000 & 1.0000 \\
\addlinespace
40 & Profile & 0.9980 & 0.9938 & 0.8850 & & 1.0000 & 1.0000 & 0.9814 & & 1.0000 & 1.0000 & 0.9998 & & 1.0000 & 1.0000 & 1.0000 \\
 & Target & 0.9986 & 0.9834 & 0.9490 & & 1.0000 & 1.0000 & 0.9974 & & 1.0000 & 1.0000 & 1.0000 & & 1.0000 & 1.0000 & 1.0000 \\
\addlinespace
50 & Profile & 0.9996 & 0.9986 & 0.9172 & & 1.0000 & 1.0000 & 0.9920 & & 1.0000 & 1.0000 & 1.0000 & & 1.0000 & 1.0000 & 1.0000 \\
 & Target & 0.9994 & 0.9926 & 0.9582 & & 1.0000 & 1.0000 & 0.9990 & & 1.0000 & 1.0000 & 1.0000 & & 1.0000 & 1.0000 & 1.0000 \\
\bottomrule
\end{tabular}
\end{table}
\end{landscape}

%% file: simulations/panel_im_test_dyn_hetero_review_table.tex
\begin{landscape}
\begin{table}[p]
\centering
\caption{Rejection Rates for Dynamic Probit Model under Unit-Scale Heteroskedasticity}
\label{tab:dynamicprobitheteroskedasticrejectionrates}
\small
\setlength{\tabcolsep}{4.5pt}
\begin{tabular}{llcccccccccccccccc}
\toprule
 & & \multicolumn{3}{c}{$n=125$} & & \multicolumn{3}{c}{$n=250$} & & \multicolumn{3}{c}{$n=500$} & & \multicolumn{3}{c}{$n=1000$} \\
\cmidrule(lr){3-5} \cmidrule(lr){7-9} \cmidrule(lr){11-13} \cmidrule(lr){15-17}
$T$ & Method & IM & Boot & Wald & & IM & Boot & Wald & & IM & Boot & Wald & & IM & Boot & Wald \\
\midrule
5 & Profile & 0.0596 & 0.0786 & 0.9676 & & 0.1086 & 0.1446 & 1.0000 & & 0.7318 & 0.1860 & 1.0000 & & 0.9992 & 0.2000 & 1.0000 \\
 & Target & 0.1508 & 0.0290 & 0.4666 & & 0.1240 & 0.0346 & 0.7312 & & 0.1718 & 0.0880 & 0.9454 & & 0.4498 & 0.3592 & 0.9994 \\
\addlinespace
10 & Profile & 0.2022 & 0.1990 & 0.9022 & & 0.5868 & 0.6224 & 0.9948 & & 0.9858 & 0.9270 & 1.0000 & & 1.0000 & 0.9918 & 1.0000 \\
 & Target & 0.2424 & 0.0606 & 0.7114 & & 0.4498 & 0.1962 & 0.9216 & & 0.9236 & 0.8022 & 0.9966 & & 0.9992 & 0.9990 & 1.0000 \\
\addlinespace
15 & Profile & 0.4644 & 0.3490 & 0.8862 & & 0.9162 & 0.9128 & 0.9918 & & 0.9996 & 0.9988 & 1.0000 & & 1.0000 & 1.0000 & 1.0000 \\
 & Target & 0.4514 & 0.1300 & 0.8276 & & 0.8762 & 0.6342 & 0.9710 & & 0.9984 & 0.9950 & 1.0000 & & 1.0000 & 1.0000 & 1.0000 \\
\addlinespace
20 & Profile & 0.7338 & 0.5370 & 0.9046 & & 0.9868 & 0.9826 & 0.9958 & & 1.0000 & 1.0000 & 1.0000 & & 1.0000 & 1.0000 & 1.0000 \\
 & Target & 0.7196 & 0.2930 & 0.8912 & & 0.9818 & 0.9074 & 0.9906 & & 0.9998 & 0.9994 & 1.0000 & & 1.0000 & 1.0000 & 1.0000 \\
\addlinespace
25 & Profile & 0.8750 & 0.6882 & 0.9262 & & 0.9964 & 0.9948 & 0.9960 & & 1.0000 & 1.0000 & 1.0000 & & 1.0000 & 1.0000 & 1.0000 \\
 & Target & 0.8724 & 0.4654 & 0.9250 & & 0.9958 & 0.9788 & 0.9952 & & 1.0000 & 1.0000 & 1.0000 & & 1.0000 & 1.0000 & 1.0000 \\
\addlinespace
30 & Profile & 0.9392 & 0.7870 & 0.9386 & & 0.9978 & 0.9966 & 0.9964 & & 1.0000 & 1.0000 & 1.0000 & & 1.0000 & 1.0000 & 1.0000 \\
 & Target & 0.9430 & 0.6216 & 0.9442 & & 0.9984 & 0.9898 & 0.9946 & & 1.0000 & 1.0000 & 1.0000 & & 1.0000 & 1.0000 & 1.0000 \\
\addlinespace
40 & Profile & 0.9818 & 0.9006 & 0.9610 & & 1.0000 & 1.0000 & 0.9982 & & 1.0000 & 1.0000 & 1.0000 & & 1.0000 & 1.0000 & 1.0000 \\
 & Target & 0.9844 & 0.8254 & 0.9670 & & 1.0000 & 0.9992 & 0.9996 & & 1.0000 & 1.0000 & 1.0000 & & 1.0000 & 1.0000 & 1.0000 \\
\addlinespace
50 & Profile & 0.9934 & 0.9472 & 0.9770 & & 1.0000 & 1.0000 & 0.9992 & & 1.0000 & 1.0000 & 1.0000 & & 1.0000 & 1.0000 & 1.0000 \\
 & Target & 0.9938 & 0.9136 & 0.9806 & & 1.0000 & 1.0000 & 0.9998 & & 1.0000 & 1.0000 & 1.0000 & & 1.0000 & 1.0000 & 1.0000 \\
\bottomrule
\end{tabular}
\end{table}
\end{landscape}

%% file: simulations/panel_im_test_logit_omitted_review_table.tex
\begin{landscape}
\begin{table}[p]
\centering
\caption{Rejection Rates for Static Logit Model under Omitted Interaction}
\label{tab:logitomittedrejectionrates}
\small
\setlength{\tabcolsep}{4.5pt}
\begin{tabular}{llcccccccccccccccc}
\toprule
 & & \multicolumn{3}{c}{$n=125$} & & \multicolumn{3}{c}{$n=250$} & & \multicolumn{3}{c}{$n=500$} & & \multicolumn{3}{c}{$n=1000$} \\
\cmidrule(lr){3-5} \cmidrule(lr){7-9} \cmidrule(lr){11-13} \cmidrule(lr){15-17}
$T$ & Method & IM & Boot & Wald & & IM & Boot & Wald & & IM & Boot & Wald & & IM & Boot & Wald \\
\midrule
5 & Profile & 0.1784 & 0.1030 & 0.1404 & & 0.7496 & 0.1074 & 0.1724 & & 0.9976 & 0.1294 & 0.2576 & & 1.0000 & 0.1570 & 0.3842 \\
 & Conditional & 0.0660 & 0.0420 & 0.4150 & & 0.0464 & 0.0448 & 0.6756 & & 0.0512 & 0.0670 & 0.9264 & & 0.1000 & 0.1532 & 0.9992 \\
 & Target & 0.0480 & 0.0252 & 0.3230 & & 0.0382 & 0.0302 & 0.5418 & & 0.1060 & 0.1176 & 0.8320 & & 0.3836 & 0.4484 & 0.9848 \\
\addlinespace
10 & Profile & 0.0940 & 0.1188 & 0.2858 & & 0.3462 & 0.2116 & 0.4656 & & 0.8218 & 0.3140 & 0.7368 & & 0.9984 & 0.5190 & 0.9484 \\
 & Conditional & 0.0902 & 0.0546 & 0.6708 & & 0.1100 & 0.0908 & 0.9232 & & 0.1936 & 0.2010 & 0.9968 & & 0.4950 & 0.5318 & 1.0000 \\
 & Target & 0.0774 & 0.0448 & 0.6052 & & 0.1188 & 0.0914 & 0.8710 & & 0.2542 & 0.2564 & 0.9920 & & 0.6622 & 0.6782 & 1.0000 \\
\addlinespace
15 & Profile & 0.1040 & 0.1414 & 0.5270 & & 0.3092 & 0.3042 & 0.8120 & & 0.7310 & 0.5346 & 0.9790 & & 0.9866 & 0.8090 & 1.0000 \\
 & Conditional & 0.1092 & 0.0672 & 0.8334 & & 0.1602 & 0.1284 & 0.9858 & & 0.3364 & 0.3312 & 1.0000 & & 0.7200 & 0.7304 & 1.0000 \\
 & Target & 0.1056 & 0.0618 & 0.7934 & & 0.1652 & 0.1316 & 0.9750 & & 0.3894 & 0.3776 & 1.0000 & & 0.7960 & 0.8000 & 1.0000 \\
\addlinespace
20 & Profile & 0.1228 & 0.1350 & 0.7506 & & 0.3150 & 0.3484 & 0.9592 & & 0.7328 & 0.6580 & 0.9988 & & 0.9874 & 0.9238 & 1.0000 \\
 & Conditional & 0.1210 & 0.0702 & 0.9284 & & 0.2138 & 0.1666 & 0.9982 & & 0.4488 & 0.4312 & 1.0000 & & 0.8352 & 0.8302 & 1.0000 \\
 & Target & 0.1188 & 0.0670 & 0.9094 & & 0.2238 & 0.1776 & 0.9970 & & 0.4910 & 0.4696 & 1.0000 & & 0.8804 & 0.8752 & 1.0000 \\
\addlinespace
25 & Profile & 0.1434 & 0.1410 & 0.8626 & & 0.3370 & 0.3766 & 0.9890 & & 0.7622 & 0.7448 & 0.9998 & & 0.9856 & 0.9612 & 1.0000 \\
 & Conditional & 0.1382 & 0.0808 & 0.9630 & & 0.2482 & 0.1958 & 1.0000 & & 0.5630 & 0.5394 & 1.0000 & & 0.9098 & 0.9106 & 1.0000 \\
 & Target & 0.1370 & 0.0742 & 0.9526 & & 0.2560 & 0.2020 & 0.9994 & & 0.5960 & 0.5660 & 1.0000 & & 0.9298 & 0.9278 & 1.0000 \\
\addlinespace
30 & Profile & 0.1674 & 0.1574 & 0.9330 & & 0.3808 & 0.4036 & 0.9972 & & 0.8034 & 0.7996 & 1.0000 & & 0.9896 & 0.9808 & 1.0000 \\
 & Conditional & 0.1662 & 0.0956 & 0.9826 & & 0.2996 & 0.2324 & 0.9998 & & 0.6428 & 0.6094 & 1.0000 & & 0.9462 & 0.9410 & 1.0000 \\
 & Target & 0.1620 & 0.0944 & 0.9778 & & 0.3064 & 0.2454 & 0.9998 & & 0.6678 & 0.6360 & 1.0000 & & 0.9590 & 0.9558 & 1.0000 \\
\addlinespace
40 & Profile & 0.1972 & 0.1672 & 0.9862 & & 0.4762 & 0.4840 & 1.0000 & & 0.8706 & 0.8788 & 1.0000 & & 0.9992 & 0.9986 & 1.0000 \\
 & Conditional & 0.1896 & 0.1112 & 0.9968 & & 0.4036 & 0.3218 & 1.0000 & & 0.7768 & 0.7510 & 1.0000 & & 0.9906 & 0.9900 & 1.0000 \\
 & Target & 0.1878 & 0.1114 & 0.9956 & & 0.4114 & 0.3392 & 1.0000 & & 0.7936 & 0.7672 & 1.0000 & & 0.9932 & 0.9914 & 1.0000 \\
\addlinespace
50 & Profile & 0.2474 & 0.1868 & 0.9978 & & 0.5676 & 0.5582 & 1.0000 & & 0.9186 & 0.9198 & 1.0000 & & 0.9994 & 0.9994 & 1.0000 \\
 & Conditional & 0.2362 & 0.1290 & 1.0000 & & 0.5008 & 0.4214 & 1.0000 & & 0.8686 & 0.8442 & 1.0000 & & 0.9974 & 0.9974 & 1.0000 \\
 & Target & 0.2350 & 0.1342 & 0.9996 & & 0.5094 & 0.4238 & 1.0000 & & 0.8754 & 0.8564 & 1.0000 & & 0.9978 & 0.9976 & 1.0000 \\
\bottomrule
\end{tabular}
\end{table}
\end{landscape}

%% file: simulations/panel_im_test_probit_omitted_review_table.tex
\begin{landscape}
\begin{table}[p]
\centering
\caption{Rejection Rates for Static Probit Model under Omitted Interaction}
\label{tab:probitomittedrejectionrates}
\small
\setlength{\tabcolsep}{4.5pt}
\begin{tabular}{llcccccccccccccccc}
\toprule
 & & \multicolumn{3}{c}{$n=125$} & & \multicolumn{3}{c}{$n=250$} & & \multicolumn{3}{c}{$n=500$} & & \multicolumn{3}{c}{$n=1000$} \\
\cmidrule(lr){3-5} \cmidrule(lr){7-9} \cmidrule(lr){11-13} \cmidrule(lr){15-17}
$T$ & Method & IM & Boot & Wald & & IM & Boot & Wald & & IM & Boot & Wald & & IM & Boot & Wald \\
\midrule
5 & Profile & 0.1348 & 0.1800 & 0.5234 & & 0.7344 & 0.2632 & 0.7404 & & 0.9976 & 0.3658 & 0.9322 & & 1.0000 & 0.5240 & 0.9980 \\
 & Target & 0.0368 & 0.0070 & 0.9254 & & 0.0386 & 0.0098 & 0.9960 & & 0.2630 & 0.1786 & 1.0000 & & 0.8798 & 0.8482 & 1.0000 \\
\addlinespace
10 & Profile & 0.1474 & 0.2506 & 0.9788 & & 0.7192 & 0.5870 & 1.0000 & & 0.9952 & 0.8414 & 1.0000 & & 1.0000 & 0.9790 & 1.0000 \\
 & Target & 0.0594 & 0.0202 & 0.9994 & & 0.2476 & 0.1264 & 1.0000 & & 0.7966 & 0.7236 & 1.0000 & & 0.9984 & 0.9974 & 1.0000 \\
\addlinespace
15 & Profile & 0.2820 & 0.3648 & 0.9996 & & 0.8440 & 0.8496 & 1.0000 & & 0.9990 & 0.9912 & 1.0000 & & 1.0000 & 0.9998 & 1.0000 \\
 & Target & 0.1544 & 0.0552 & 1.0000 & & 0.5572 & 0.3966 & 1.0000 & & 0.9696 & 0.9496 & 1.0000 & & 1.0000 & 0.9998 & 1.0000 \\
\addlinespace
20 & Profile & 0.4164 & 0.4602 & 1.0000 & & 0.9258 & 0.9420 & 1.0000 & & 1.0000 & 0.9998 & 1.0000 & & 1.0000 & 1.0000 & 1.0000 \\
 & Target & 0.2784 & 0.1064 & 1.0000 & & 0.7798 & 0.6412 & 1.0000 & & 0.9974 & 0.9948 & 1.0000 & & 1.0000 & 1.0000 & 1.0000 \\
\addlinespace
25 & Profile & 0.5630 & 0.5424 & 1.0000 & & 0.9686 & 0.9726 & 1.0000 & & 0.9998 & 1.0000 & 1.0000 & & 1.0000 & 1.0000 & 1.0000 \\
 & Target & 0.4330 & 0.2046 & 1.0000 & & 0.9064 & 0.8206 & 1.0000 & & 0.9998 & 0.9994 & 1.0000 & & 1.0000 & 1.0000 & 1.0000 \\
\addlinespace
30 & Profile & 0.6850 & 0.6282 & 1.0000 & & 0.9884 & 0.9894 & 1.0000 & & 1.0000 & 1.0000 & 1.0000 & & 1.0000 & 1.0000 & 1.0000 \\
 & Target & 0.5682 & 0.3032 & 1.0000 & & 0.9628 & 0.9112 & 1.0000 & & 0.9998 & 0.9994 & 1.0000 & & 1.0000 & 1.0000 & 1.0000 \\
\addlinespace
40 & Profile & 0.8486 & 0.7658 & 1.0000 & & 0.9980 & 0.9980 & 1.0000 & & 1.0000 & 1.0000 & 1.0000 & & 1.0000 & 1.0000 & 1.0000 \\
 & Target & 0.7804 & 0.5238 & 1.0000 & & 0.9958 & 0.9856 & 1.0000 & & 1.0000 & 1.0000 & 1.0000 & & 1.0000 & 1.0000 & 1.0000 \\
\addlinespace
50 & Profile & 0.9328 & 0.8608 & 1.0000 & & 0.9998 & 0.9996 & 1.0000 & & 1.0000 & 1.0000 & 1.0000 & & 1.0000 & 1.0000 & 1.0000 \\
 & Target & 0.9026 & 0.7068 & 1.0000 & & 0.9998 & 0.9980 & 1.0000 & & 1.0000 & 1.0000 & 1.0000 & & 1.0000 & 1.0000 & 1.0000 \\
\bottomrule
\end{tabular}
\end{table}
\end{landscape}

%% file: simulations/panel_im_test_dyn_omitted_review_table.tex
\begin{landscape}
\begin{table}[p]
\centering
\caption{Rejection Rates for Dynamic Probit Model under Omitted Interaction}
\label{tab:dynamicprobitomittedrejectionrates}
\small
\setlength{\tabcolsep}{4.5pt}
\begin{tabular}{llcccccccccccccccc}
\toprule
 & & \multicolumn{3}{c}{$n=125$} & & \multicolumn{3}{c}{$n=250$} & & \multicolumn{3}{c}{$n=500$} & & \multicolumn{3}{c}{$n=1000$} \\
\cmidrule(lr){3-5} \cmidrule(lr){7-9} \cmidrule(lr){11-13} \cmidrule(lr){15-17}
$T$ & Method & IM & Boot & Wald & & IM & Boot & Wald & & IM & Boot & Wald & & IM & Boot & Wald \\
\midrule
5 & Profile & 0.0722 & 0.0940 & 0.9872 & & 0.1906 & 0.1972 & 1.0000 & & 0.8418 & 0.2674 & 1.0000 & & 1.0000 & 0.3222 & 1.0000 \\
 & Target & 0.0890 & 0.0184 & 0.9828 & & 0.0386 & 0.0116 & 1.0000 & & 0.0850 & 0.0316 & 1.0000 & & 0.4964 & 0.3846 & 1.0000 \\
\addlinespace
10 & Profile & 0.1210 & 0.1140 & 0.9932 & & 0.2970 & 0.2968 & 1.0000 & & 0.9000 & 0.6644 & 1.0000 & & 1.0000 & 0.9178 & 1.0000 \\
 & Target & 0.0958 & 0.0212 & 0.9998 & & 0.1186 & 0.0390 & 1.0000 & & 0.4690 & 0.3120 & 1.0000 & & 0.9716 & 0.9452 & 1.0000 \\
\addlinespace
15 & Profile & 0.1944 & 0.1446 & 0.9990 & & 0.5408 & 0.5262 & 1.0000 & & 0.9804 & 0.9470 & 1.0000 & & 1.0000 & 0.9996 & 1.0000 \\
 & Target & 0.1432 & 0.0372 & 1.0000 & & 0.3120 & 0.1308 & 1.0000 & & 0.8634 & 0.7424 & 1.0000 & & 1.0000 & 0.9984 & 1.0000 \\
\addlinespace
20 & Profile & 0.2866 & 0.1716 & 1.0000 & & 0.7554 & 0.7090 & 1.0000 & & 0.9962 & 0.9934 & 1.0000 & & 1.0000 & 1.0000 & 1.0000 \\
 & Target & 0.2194 & 0.0586 & 1.0000 & & 0.5634 & 0.3026 & 1.0000 & & 0.9802 & 0.9440 & 1.0000 & & 1.0000 & 1.0000 & 1.0000 \\
\addlinespace
25 & Profile & 0.3946 & 0.2298 & 1.0000 & & 0.8910 & 0.8400 & 1.0000 & & 0.9998 & 1.0000 & 1.0000 & & 1.0000 & 1.0000 & 1.0000 \\
 & Target & 0.3120 & 0.0888 & 1.0000 & & 0.7782 & 0.5348 & 1.0000 & & 0.9984 & 0.9940 & 1.0000 & & 1.0000 & 1.0000 & 1.0000 \\
\addlinespace
30 & Profile & 0.5180 & 0.2944 & 1.0000 & & 0.9522 & 0.9138 & 1.0000 & & 0.9998 & 0.9998 & 1.0000 & & 1.0000 & 1.0000 & 1.0000 \\
 & Target & 0.4356 & 0.1462 & 1.0000 & & 0.9064 & 0.7280 & 1.0000 & & 0.9998 & 0.9980 & 1.0000 & & 1.0000 & 1.0000 & 1.0000 \\
\addlinespace
40 & Profile & 0.7182 & 0.4276 & 1.0000 & & 0.9932 & 0.9830 & 1.0000 & & 1.0000 & 1.0000 & 1.0000 & & 1.0000 & 1.0000 & 1.0000 \\
 & Target & 0.6586 & 0.2760 & 1.0000 & & 0.9884 & 0.9350 & 1.0000 & & 1.0000 & 1.0000 & 1.0000 & & 1.0000 & 1.0000 & 1.0000 \\
\addlinespace
50 & Profile & 0.8522 & 0.5836 & 1.0000 & & 0.9990 & 0.9944 & 1.0000 & & 1.0000 & 1.0000 & 1.0000 & & 1.0000 & 1.0000 & 1.0000 \\
 & Target & 0.8306 & 0.4358 & 1.0000 & & 0.9976 & 0.9846 & 1.0000 & & 1.0000 & 1.0000 & 1.0000 & & 1.0000 & 1.0000 & 1.0000 \\
\bottomrule
\end{tabular}
\end{table}
\end{landscape}

%% file: simulations/panel_im_test_logit_link_review_table.tex
\begin{landscape}
\begin{table}[p]
\centering
\caption{Rejection Rates for Static Logit Model under Wrong Link Function}
\label{tab:logitlinkrejectionrates}
\small
\setlength{\tabcolsep}{4.5pt}
\begin{tabular}{llcccccccccccccccc}
\toprule
 & & \multicolumn{3}{c}{$n=125$} & & \multicolumn{3}{c}{$n=250$} & & \multicolumn{3}{c}{$n=500$} & & \multicolumn{3}{c}{$n=1000$} \\
\cmidrule(lr){3-5} \cmidrule(lr){7-9} \cmidrule(lr){11-13} \cmidrule(lr){15-17}
$T$ & Method & IM & Boot & Wald & & IM & Boot & Wald & & IM & Boot & Wald & & IM & Boot & Wald \\
\midrule
8 & Profile & 0.0280 & 0.0590 & 1.0000 & & 0.1454 & 0.0534 & 1.0000 & & 0.6774 & 0.0377 & 1.0000 & & 0.9944 & 0.0246 & 1.0000 \\
 & Conditional & 0.1337 & 0.0738 & 1.0000 & & 0.0986 & 0.0818 & 1.0000 & & 0.0715 & 0.0891 & 1.0000 & & 0.0837 & 0.1262 & 1.0000 \\
 & Target & 0.1529 & 0.0824 & 1.0000 & & 0.1168 & 0.0868 & 1.0000 & & 0.0985 & 0.1015 & 1.0000 & & 0.1191 & 0.1578 & 1.0000 \\
\addlinespace
10 & Profile & 0.0276 & 0.0626 & 1.0000 & & 0.0814 & 0.0450 & 1.0000 & & 0.4010 & 0.0312 & 1.0000 & & 0.9086 & 0.0232 & 1.0000 \\
 & Conditional & 0.1298 & 0.0774 & 1.0000 & & 0.1032 & 0.0820 & 1.0000 & & 0.0866 & 0.1016 & 1.0000 & & 0.0820 & 0.1228 & 1.0000 \\
 & Target & 0.1480 & 0.0812 & 1.0000 & & 0.1206 & 0.0976 & 1.0000 & & 0.1082 & 0.1140 & 1.0000 & & 0.1144 & 0.1418 & 1.0000 \\
\addlinespace
15 & Profile & 0.0396 & 0.0668 & 1.0000 & & 0.0450 & 0.0480 & 1.0000 & & 0.1312 & 0.0318 & 1.0000 & & 0.4252 & 0.0176 & 1.0000 \\
 & Conditional & 0.1408 & 0.0774 & 1.0000 & & 0.1122 & 0.0940 & 1.0000 & & 0.0958 & 0.1044 & 1.0000 & & 0.1122 & 0.1406 & 1.0000 \\
 & Target & 0.1466 & 0.0840 & 1.0000 & & 0.1204 & 0.0958 & 1.0000 & & 0.1088 & 0.1124 & 1.0000 & & 0.1370 & 0.1624 & 1.0000 \\
\addlinespace
20 & Profile & 0.0520 & 0.0712 & 1.0000 & & 0.0394 & 0.0500 & 1.0000 & & 0.0630 & 0.0348 & 1.0000 & & 0.1820 & 0.0198 & 1.0000 \\
 & Conditional & 0.1324 & 0.0772 & 1.0000 & & 0.1148 & 0.0908 & 1.0000 & & 0.1020 & 0.1026 & 1.0000 & & 0.1376 & 0.1526 & 1.0000 \\
 & Target & 0.1348 & 0.0784 & 1.0000 & & 0.1200 & 0.0966 & 1.0000 & & 0.1090 & 0.1088 & 1.0000 & & 0.1396 & 0.1692 & 1.0000 \\
\addlinespace
25 & Profile & 0.0622 & 0.0722 & 1.0000 & & 0.0454 & 0.0614 & 1.0000 & & 0.0502 & 0.0420 & 1.0000 & & 0.1150 & 0.0288 & 1.0000 \\
 & Conditional & 0.1378 & 0.0830 & 1.0000 & & 0.1228 & 0.0960 & 1.0000 & & 0.1158 & 0.1116 & 1.0000 & & 0.1404 & 0.1604 & 1.0000 \\
 & Target & 0.1388 & 0.0822 & 1.0000 & & 0.1258 & 0.1000 & 1.0000 & & 0.1194 & 0.1170 & 1.0000 & & 0.1564 & 0.1720 & 1.0000 \\
\addlinespace
30 & Profile & 0.0752 & 0.0772 & 1.0000 & & 0.0452 & 0.0640 & 1.0000 & & 0.0482 & 0.0490 & 1.0000 & & 0.0844 & 0.0370 & 1.0000 \\
 & Conditional & 0.1404 & 0.0846 & 1.0000 & & 0.1176 & 0.0928 & 1.0000 & & 0.1136 & 0.1136 & 1.0000 & & 0.1444 & 0.1546 & 1.0000 \\
 & Target & 0.1420 & 0.0872 & 1.0000 & & 0.1180 & 0.0938 & 1.0000 & & 0.1214 & 0.1138 & 1.0000 & & 0.1542 & 0.1622 & 1.0000 \\
\addlinespace
40 & Profile & 0.0740 & 0.0678 & 1.0000 & & 0.0556 & 0.0720 & 1.0000 & & 0.0538 & 0.0636 & 1.0000 & & 0.0700 & 0.0508 & 1.0000 \\
 & Conditional & 0.1278 & 0.0726 & 1.0000 & & 0.1210 & 0.0926 & 1.0000 & & 0.1150 & 0.1136 & 1.0000 & & 0.1404 & 0.1508 & 1.0000 \\
 & Target & 0.1322 & 0.0752 & 1.0000 & & 0.1222 & 0.0942 & 1.0000 & & 0.1198 & 0.1172 & 1.0000 & & 0.1462 & 0.1530 & 1.0000 \\
\addlinespace
50 & Profile & 0.0986 & 0.0800 & 1.0000 & & 0.0674 & 0.0808 & 1.0000 & & 0.0624 & 0.0792 & 1.0000 & & 0.0628 & 0.0592 & 1.0000 \\
 & Conditional & 0.1424 & 0.0868 & 1.0000 & & 0.1250 & 0.0946 & 1.0000 & & 0.1256 & 0.1166 & 1.0000 & & 0.1350 & 0.1410 & 1.0000 \\
 & Target & 0.1452 & 0.0880 & 1.0000 & & 0.1252 & 0.0972 & 1.0000 & & 0.1250 & 0.1222 & 1.0000 & & 0.1374 & 0.1514 & 1.0000 \\
\bottomrule
\end{tabular}
\end{table}
\end{landscape}

%% file: simulations/panel_im_test_neyman_scott_review_table.tex
\begin{landscape}
\begin{table}[p]
\centering
\caption{Rejection Rates for the Neyman--Scott Model}
\label{tab:ns_rejection_rates}
\small
\setlength{\tabcolsep}{4.5pt}
\begin{tabular}{llcccccccccccccccc}
\toprule
 & & \multicolumn{3}{c}{$n=125$} & & \multicolumn{3}{c}{$n=250$} & & \multicolumn{3}{c}{$n=500$} & & \multicolumn{3}{c}{$n=1000$} \\
\cmidrule(lr){3-5} \cmidrule(lr){7-9} \cmidrule(lr){11-13} \cmidrule(lr){15-17}
$T$ & Method & IM & Boot & Wald & & IM & Boot & Wald & & IM & Boot & Wald & & IM & Boot & Wald \\
\midrule
5 & Profile & 0.1662 & 0.0594 & 0.9748 & & 0.4250 & 0.0576 & 0.9996 & & 0.7942 & 0.0618 & 1.0000 & & 0.9834 & 0.0576 & 1.0000 \\
 & Target & 0.0772 & 0.0558 & 0.0510 & & 0.0654 & 0.0580 & 0.0480 & & 0.0632 & 0.0590 & 0.0476 & & 0.0576 & 0.0544 & 0.0552 \\
\addlinespace
10 & Profile & 0.0612 & 0.0592 & 0.7646 & & 0.1132 & 0.0616 & 0.9600 & & 0.2556 & 0.0566 & 0.9994 & & 0.5098 & 0.0648 & 1.0000 \\
 & Target & 0.0714 & 0.0580 & 0.0538 & & 0.0662 & 0.0592 & 0.0476 & & 0.0572 & 0.0604 & 0.0510 & & 0.0540 & 0.0590 & 0.0548 \\
\addlinespace
15 & Profile & 0.0522 & 0.0622 & 0.5804 & & 0.0660 & 0.0582 & 0.8548 & & 0.1280 & 0.0588 & 0.9846 & & 0.2612 & 0.0624 & 1.0000 \\
 & Target & 0.0748 & 0.0616 & 0.0500 & & 0.0592 & 0.0586 & 0.0572 & & 0.0562 & 0.0578 & 0.0572 & & 0.0532 & 0.0582 & 0.0518 \\
\addlinespace
20 & Profile & 0.0544 & 0.0604 & 0.4768 & & 0.0574 & 0.0562 & 0.7408 & & 0.0882 & 0.0580 & 0.9524 & & 0.1660 & 0.0654 & 0.9994 \\
 & Target & 0.0804 & 0.0630 & 0.0518 & & 0.0598 & 0.0574 & 0.0534 & & 0.0574 & 0.0600 & 0.0524 & & 0.0596 & 0.0626 & 0.0494 \\
\addlinespace
25 & Profile & 0.0566 & 0.0620 & 0.3818 & & 0.0548 & 0.0616 & 0.6354 & & 0.0676 & 0.0590 & 0.8992 & & 0.1126 & 0.0580 & 0.9956 \\
 & Target & 0.0730 & 0.0610 & 0.0508 & & 0.0616 & 0.0584 & 0.0486 & & 0.0516 & 0.0522 & 0.0506 & & 0.0564 & 0.0612 & 0.0494 \\
\addlinespace
30 & Profile & 0.0572 & 0.0604 & 0.3316 & & 0.0518 & 0.0598 & 0.5564 & & 0.0586 & 0.0610 & 0.8340 & & 0.1010 & 0.0678 & 0.9854 \\
 & Target & 0.0736 & 0.0606 & 0.0518 & & 0.0660 & 0.0600 & 0.0514 & & 0.0500 & 0.0556 & 0.0524 & & 0.0576 & 0.0640 & 0.0524 \\
\addlinespace
40 & Profile & 0.0560 & 0.0600 & 0.2662 & & 0.0524 & 0.0594 & 0.4416 & & 0.0526 & 0.0596 & 0.7338 & & 0.0624 & 0.0522 & 0.9454 \\
 & Target & 0.0708 & 0.0644 & 0.0536 & & 0.0632 & 0.0642 & 0.0506 & & 0.0582 & 0.0586 & 0.0560 & & 0.0472 & 0.0540 & 0.0486 \\
\addlinespace
50 & Profile & 0.0544 & 0.0546 & 0.2216 & & 0.0556 & 0.0636 & 0.3712 & & 0.0532 & 0.0600 & 0.6124 & & 0.0612 & 0.0588 & 0.8924 \\
 & Target & 0.0646 & 0.0564 & 0.0522 & & 0.0614 & 0.0596 & 0.0516 & & 0.0554 & 0.0624 & 0.0540 & & 0.0492 & 0.0562 & 0.0534 \\
\bottomrule
\end{tabular}
\end{table}
\end{landscape}

%% file: applied_example/FV_lfp_pim_results.tex
\begin{table}[H]
\centering
\caption{Profile PIM for the dynamic FE probit LFP specification of \citet{fernandez-val-probit}.}
\label{tab:fv-lfp-pim}
\begin{tabular}{lcc}
\toprule
 & Asymptotic & Bootstrap ($B=99$) \\
\midrule
PIM statistic $\mathcal{W}^{p}$ & \multicolumn{2}{c}{333.596} \\
Degrees of freedom $q$ & \multicolumn{2}{c}{120} \\
5\% critical value & 146.567 & 317.667 \\
Reject $H_0$ at 5\% & Yes & Yes \\
Bootstrap $p$-value & --- & 0.010 \\
\bottomrule
\end{tabular}
\begin{flushleft}\footnotesize
Notes: Balanced panel of $n=664$ women with time-varying LFP ($T=9$, years 1980--1988) from the PSID extract of \citet{fernandez-val-probit}. The fitted model is the dynamic FE probit in \eqref{eq:fv-lfp-spec}. The null is correct specification of that conditional likelihood. Asymptotic critical values are from $\chi^2_q$; bootstrap critical values are the 95th percentile of the parametric profile bootstrap with recursive lagged participation and redraws of uninformative paths.
\end{flushleft}
\end{table}

%% file: applied_example/FV_lfp_diagnostics.tex
\begin{table}[H]
\centering
\caption{Directional diagnostics after omnibus PIM rejection.}
\label{tab:fv-lfp-diagnostics}
\begin{tabular}{lcccc}
\toprule
Diagnostic & Statistic & df & $p$ (asym.) & $p$ (boot.) \\
\midrule
\multicolumn{5}{l}{\emph{A. Blockwise profile-IM contributions}} \\
lag & 13.70 & 15 & 0.549 & --- \\
fertility & 102.10 & 42 & 0.000 & --- \\
income & 25.51 & 15 & 0.043 & --- \\
age & 112.13 & 29 & 0.000 & --- \\
time & 314.66 & 92 & 0.000 & --- \\
\midrule
\multicolumn{5}{l}{\emph{B. Targeted alternatives}} \\
Heteroskedasticity (LM) & 13.26 & 4 & 0.0101 & 0.010 \\
kids02 $\times$ lag (LR) & 7.68 & 1 & 0.00557 & 0.010 \\
AR(2) lag2 (LR) & 3.18 & 1 & 0.0744 & 1.000 \\
\bottomrule
\end{tabular}
\begin{flushleft}\footnotesize
Notes: Blockwise statistics are profile-IM Wald statistics using the IEQ moments that involve the listed parameter block (possibly overlapping across blocks). Heteroskedasticity is an LM score test of $\sigma_{it}=\exp(W_{it}'\gamma)$ at $\gamma=0$ with $W$ containing fertility counts and log husband's earnings. The kids$\times$lag and AR(2) rows are likelihood-ratio tests against the nested FE probit (AR(2) on the $T=8$ subsample 1981--1988). Bootstrap $p$-values for targeted alternatives use parametric draws under the baseline null ($B=99$).
\end{flushleft}
\end{table}

%% file: applied_example/FV_lfp_repair_pim.tex
\begin{table}[H]
\centering
\caption{Profile PIM before and after the selected repair (kids0--2 $\times$ lag).}
\label{tab:fv-lfp-repair}
\begin{tabular}{lcc}
\toprule
 & Baseline & Repaired \\
\midrule
PIM statistic $\mathcal{W}^{p}$ & 333.596 & 341.692 \\
Degrees of freedom $q$ & 120 & 136 \\
Asymptotic 5\% critical value & 146.567 & 164.216 \\
Bootstrap 5\% critical value ($B=99$) & 317.667 & 319.435 \\
Reject at 5\% (asymptotic) & Yes & Yes \\
Reject at 5\% (bootstrap) & Yes & Yes \\
Bootstrap $p$-value & 0.010 & 0.000 \\
\bottomrule
\end{tabular}
\begin{flushleft}\footnotesize
Notes: Baseline is the dynamic FE probit \eqref{eq:fv-lfp-spec}. The repaired model is the alternative with the smallest diagnostic $p$-value in Table~\ref{tab:fv-lfp-diagnostics} (or the feasible fertility$\times$lag enrichment when the heteroskedasticity LM wins). Bootstrap critical values use the parametric profile bootstrap of Section~\ref{sec:bootstrap}.
\end{flushleft}
\end{table}

%% file: applied_example/union_pim_results.tex
\begin{table}[H]
\centering
\caption{Omnibus profile PIM for the baseline static FE probit union specification.}
\label{tab:union-pim}
\begin{tabular}{lcc}
\toprule
 & Asymptotic & Bootstrap ($B=99$) \\
\midrule
PIM statistic $\mathcal{W}^{p}$ & \multicolumn{2}{c}{194.327} \\
Degrees of freedom $q$ & \multicolumn{2}{c}{36} \\
5\% critical value & 50.998 & 191.836 \\
Reject $H_0$ at 5\% & Yes & Yes \\
Bootstrap $p$-value & --- & 0.010 \\
\bottomrule
\end{tabular}
\begin{flushleft}\footnotesize
Notes: Balanced panel of $n=246$ men with time-varying union status ($T=8$, 1980--1987) from the \citet{vella-verbeek-1998} NLSY extract. The baseline specification is a static FE probit with married, and year dummies. Bootstrap critical values are from the parametric profile bootstrap ($B=99$).
\end{flushleft}
\end{table}

%% file: applied_example/union_diagnostics.tex
\begin{table}[H]
\centering
\caption{Directional diagnostics after omnibus PIM rejection (union membership).}
\label{tab:union-diagnostics}
\begin{tabular}{lcccc}
\toprule
Diagnostic & Statistic & $q$ & $p$ (asym.) & $p$ (boot.) \\
\midrule
\multicolumn{5}{l}{\emph{A. Blockwise profile-IM contributions}} \\
marriage & 29.49 & 8 & 0.000 & 0.000 \\
time & 194.32 & 35 & 0.000 & 0.010 \\
\midrule
\multicolumn{5}{l}{\emph{B. Targeted alternatives}} \\
Hours worked (LR) & 4.10 & 1 & 0.0428 & 0.081 \\
Manufacturing (LR) & 18.74 & 1 & 1.5e-05 & 0.000 \\
Construction (LR) & 2.49 & 1 & 0.114 & 0.081 \\
Wholesale/retail trade (LR) & 15.97 & 1 & 6.42e-05 & 0.000 \\
Industry bundle (LR) & 26.56 & 3 & 7.3e-06 & 0.000 \\
Hours + industry bundle (LR) & 31.42 & 4 & 2.52e-06 & 0.000 \\
Heteroskedasticity (LM) & 3.11 & 1 & 0.078 & 0.030 \\
\bottomrule
\end{tabular}
\begin{flushleft}\footnotesize
Notes: Blockwise statistics are profile-IM Wald statistics using IEQ moments involving the listed parameter block. Bootstrap $p$-values for blockwise tests use the same $B=99$ omnibus draws. Targeted alternatives are LR or LM tests against the baseline; bootstrap $p$-values are from parametric draws under the baseline null ($B=99$).
\end{flushleft}
\end{table}

%% file: applied_example/union_repair_pim.tex
\begin{table}[H]
\centering
\caption{Profile PIM before and after repair (union membership).}
\label{tab:union-repair}
\begin{tabular}{lcc}
\toprule
 & Baseline & Repaired \\
\midrule
PIM statistic $\mathcal{W}^{p}$ & 194.327 & 21.589 \\
Degrees of freedom $q$ & 36 & 66 \\
Asymptotic 5\% critical value & 50.998 & 85.965 \\
Bootstrap 5\% critical value ($B=99$) & 191.836 & 163.965 \\
Reject at 5\% (asymptotic) & Yes & No \\
Reject at 5\% (bootstrap) & Yes & No \\
Bootstrap $p$-value & 0.010 & 0.919 \\
\bottomrule
\end{tabular}
\begin{flushleft}\footnotesize
Notes: Baseline is the static FE probit with married and year dummies. The repaired model adds hours worked, manufacturing, and construction indicators. Bootstrap critical values use the parametric profile bootstrap ($B=99$).
\end{flushleft}
\end{table}

%% file: applied_example/enia_pim_results.tex
\begin{table}[H]
\centering
\caption{Omnibus profile PIM for the baseline dynamic FE probit export specification (exploration subsample).}
\label{tab:enia-pim}
\begin{tabular}{lcc}
\toprule
 & Asymptotic & Bootstrap ($B=99$) \\
\midrule
PIM statistic $\mathcal{W}^{p}$ & \multicolumn{2}{c}{102.408} \\
Degrees of freedom $q$ & \multicolumn{2}{c}{66} \\
Informative plants ($n$) & \multicolumn{2}{c}{197} \\
5\% critical value & 85.965 & 105.073 \\
Reject $H_0$ at 5\% & Yes & No \\
Bootstrap $p$-value & --- & 0.152 \\
\bottomrule
\end{tabular}
\begin{flushleft}\footnotesize
Notes: Chile ENIA panel of export-status movers with $T=10$ (1996--2005). Exploration subsample: $n=279$ plants (70\% random split, seed=20260819). Baseline: lagged export status, log employment, year effects. 82 plants dropped due to non-finite IM indicators.
\end{flushleft}
\end{table}

%% file: applied_example/enia_diagnostics.tex
\begin{table}[H]
\centering
\caption{Directional diagnostics after omnibus PIM rejection (ENIA export, exploration).}
\label{tab:enia-diagnostics}
\begin{tabular}{lcccc}
\toprule
Diagnostic & Statistic & $q$ & $p$ (asym.) & $p$ (boot.) \\
\midrule
\multicolumn{5}{l}{\emph{A. Blockwise profile-IM contributions}} \\
lag & 30.13 & 11 & 0.002 & 0.323 \\
size & 27.03 & 11 & 0.005 & 0.394 \\
time & 96.99 & 63 & 0.004 & 0.101 \\
\midrule
\multicolumn{5}{l}{\emph{B. Targeted alternatives}} \\
Foreign ownership (LR) & 0.33 & 1 & 0.564 & 0.636 \\
Santiago location (LR) & 0.88 & 1 & 0.349 & 0.788 \\
Log capital stock (LR) & 0.02 & 1 & 0.878 & 0.919 \\
Public ownership (LR) & 0.00 & 1 & 0.954 & 0.960 \\
Individual/partnership (LR) & 2.93 & 1 & 0.0871 & 0.111 \\
High skill share (LR) & 4.75 & 1 & 0.0293 & 0.040 \\
Foreign + high skill (LR) & 5.02 & 2 & 0.0813 & 0.182 \\
Foreign + log capital (LR) & 0.35 & 2 & 0.84 & 0.889 \\
Heteroskedasticity (LM) & 3.43 & 2 & 0.18 & 0.051 \\
\bottomrule
\end{tabular}
\begin{flushleft}\footnotesize
Notes: Blockwise and targeted bootstrap $p$-values use parametric draws under the baseline null ($B=99$) on the exploration subsample.
\end{flushleft}
\end{table}

%% file: applied_example/enia_repair_pim.tex
\begin{table}[H]
\centering
\caption{Profile PIM before and after specification repair (ENIA export, exploration).}
\label{tab:enia-repair}
\begin{tabular}{lcc}
\toprule
 & Baseline & Fixed (S5) \\
\midrule
PIM statistic $\mathcal{W}^{p}$ & 102.408 & 113.613 \\
Degrees of freedom $q$ & 66 & 91 \\
5\% critical value (asymptotic) & 85.965 & 114.268 \\
5\% critical value (bootstrap, $B=99$) & 105.073 & 125.377 \\
Reject at 5\% (asymptotic) & Yes & No \\
Reject at 5\% (bootstrap) & No & No \\
Bootstrap $p$-value & 0.152 & 0.455 \\
\bottomrule
\end{tabular}
\begin{flushleft}\footnotesize
Notes: Fixed specification adds foreign ownership and high-skilled employment share to the baseline dynamic export probit.
\end{flushleft}
\end{table}

%% file: applied_example/enia_holdout_pim.tex
\begin{table}[H]
\centering
\caption{Hold-out validation of the fixed export specification (S5).}
\label{tab:enia-holdout}
\begin{tabular}{lcc}
\toprule
 & Asymptotic & Bootstrap ($B=99$) \\
\midrule
PIM statistic $\mathcal{W}^{p}$ & \multicolumn{2}{c}{32.000} \\
Degrees of freedom $q$ & \multicolumn{2}{c}{91} \\
Hold-out plants ($n$) & \multicolumn{2}{c}{32} \\
5\% critical value & 114.268 & 32.000 \\
Reject $H_0$ at 5\% & No & No \\
Bootstrap $p$-value & --- & 0.657 \\
\bottomrule
\end{tabular}
\begin{flushleft}\footnotesize
Notes: Hold-out sample contains 120 plants (30\% of movers), disjoint from the exploration subsample. Same fixed specification as selected on the exploration sample.
\end{flushleft}
\end{table}